\documentclass{article}

\usepackage{microtype}
\usepackage{graphicx}
\usepackage{booktabs} 

\usepackage{hyperref}

\usepackage[accepted]{icml2024}

\usepackage{amsmath}
\usepackage{amssymb}
\usepackage{mathtools}
\usepackage{amsthm}

\usepackage[capitalize,noabbrev]{cleveref}

\usepackage{multirow}
\usepackage{amsfonts}
\usepackage{amsmath, amssymb}
\usepackage{amsthm}
\usepackage{relsize}
\usepackage{enumitem}
\usepackage{upgreek}
\usepackage{caption}
\usepackage{subcaption}
\usepackage{microtype}
\usepackage{graphicx}
\usepackage{booktabs} 
\usepackage{algorithm}
\usepackage{algorithmic}[1]
\usepackage{wrapfig}
\usepackage{MnSymbol}
\usepackage{tikz}

\theoremstyle{plain}
\newtheorem{theorem}[]{Theorem}
\newtheorem{lemma}[]{Lemma}
\newtheorem{proposition}[]{Proposition}

\newtheorem{definition}[]{Definition}
\newtheorem{assumption}[]{Assumption}

\newtheorem{remark}[]{Remark}

\newcommand{\bigCI}{\mathrel{\text{\scalebox{1.07}{$\perp\mkern-10mu\perp$}}}}

\usepackage[textsize=tiny]{todonotes}

\icmltitlerunning{Open Ad Hoc Teamwork with Cooperative Game Theory}

\begin{document}

\twocolumn[
\icmltitle{Open Ad Hoc Teamwork with Cooperative Game Theory}



\icmlsetsymbol{equal}{*}

\begin{icmlauthorlist}
\icmlauthor{Jianhong Wang}{yyy}
\icmlauthor{Yang Li}{yyy}
\icmlauthor{Yuan Zhang}{comp}
\icmlauthor{Wei Pan}{yyy}
\icmlauthor{Samuel Kaski}{yyy,aaa}
\end{icmlauthorlist}

\icmlaffiliation{yyy}{Center for AI Fundamentals, University of Manchester, UK}
\icmlaffiliation{aaa}{Aalto University, Finland}
\icmlaffiliation{comp}{Neurorobotics Lab, University of Freiburg, Germany}

\icmlcorrespondingauthor{Jianhong Wang}{jianhong.wang@manchester.ac.uk}

\icmlkeywords{Machine Learning, ICML}

\vskip 0.3in
]



\printAffiliationsAndNotice{}  

\begin{abstract}
Ad hoc teamwork poses a challenging problem, requiring the design of an agent to collaborate with teammates without prior coordination or joint training. Open ad hoc teamwork (OAHT) further complicates this challenge by considering environments with a changing number of teammates, referred to as open teams. One promising solution in practice to this problem is leveraging the generalizability of graph neural networks to handle an unrestricted number of agents with various agent-types, named graph-based policy learning (GPL). However, its joint Q-value representation over a coordination graph lacks convincing explanations. In this paper, we establish a new theory to understand the representation of the joint Q-value for OAHT and its learning paradigm, through the lens of cooperative game theory. Building on our theory, we propose a novel algorithm named CIAO, based on GPL's framework, with additional provable implementation tricks that can facilitate learning. The demos of experimental results are available on \url{https://sites.google.com/view/ciao2024}, and the code of experiments is published on \url{https://github.com/hsvgbkhgbv/CIAO}.
\end{abstract}

\section{Introduction}
\label{sec:introduction}
    Multi-agent reinforcement learning (MARL) has achieved partial success on multiple tasks including playing strategy games \citep{rashid2020weighted}, power system operation \citep{wang2021multi}, and dynamic algorithm configuration \citep{xue2022multi}. These tasks fit to the training paradigm of MARL, which requires all agents to be controllable and to be coordinated during training. However, with this paradigm it is difficult to tackle many real-world tasks where not all agents are controllable and even prior coordination may not be possible. For example, in search and rescue, a robot must collaborate with other robots it has not seen before (e.g., manufactured by various companies without a common coordination protocol) or humans to rescue survivors \citep{barrett2015cooperating}. Similar situations occur in AI that helps trading markets \citep{albrecht2013game}, as well as in the human-machine and machine-machine collaboration emerging from the prevailing embodied AI settings \citep{smith2005development,duan2022survey} and large language models \citep{brown2020language,zhao2023survey}.

    To tackle the ad hoc teamwork problem, we explore a scenario where one agent, referred to as the \textbf{learner}, operates under our control and seeks to collaborate without prior coordination with \textbf{teammates} which have unknown types and policies \citep{stone2010ad}. When dealing with teams of dynamic sizes, commonly termed \textbf{open teams}, the research problem addressed in this paper is often referred to as \textbf{open ad hoc teamwork} (OAHT) \citep{mirsky2022survey}. One promising solution for OAHT in practice is graph-based policy learning (GPL) \citep{rahman2021towards}. GPL presents a three-fold framework, encompassing a type inference model, a joint action value model and an agent model. Although GPL reaps the success of performance due to the generalizability of graph neural networks to handle an unrestricted number of agents with various agent-types, its weakness is that the representation of the joint Q-value over a coordination graph lacks convincing explanations. This restricts its applicability to real-world problems requiring trustworthy algorithms \citep{bhat20216g,wang2021multi}.
    
    We propose to describe OAHT using a game model from cooperative game theory, namely the \textbf{coalitional affinity game} (CAG) \citep{branzei2009coalitional}. Specifically, we extend the CAG by incorporating Bayesian games \citep{harsanyi1967games} to depict uncertain agent-types and stochastic games \citep{shapley1953stochastic} to represent the long-horizon goal. The resulting game is termed the \textbf{open stochastic Bayesian coalitional affinity game} (OSB-CAG). In this game, the learner aims to influence other teammates (via its actions) to collaborate in achieving a shared goal. To formalize this, we extend the standard cooperative game theory notion of strict core to a novel solution concept which we call \textbf{dynamic variational strict core} (DVSC). The DVSC transforms collaboration in a temporary team into the task of forming a stable temporary team, where no agent has incentives to leave. We model the OAHT process under the learner's influence as a \textbf{dynamic affinity graph} (equivalent to a coordination graph), generalizing the classical static CAG. Based on the dynamic affinity graph, we further conceptualize an agent's preference for a temporary team to measure whether they prefer to stay in the team under the learner's influence. GPL's joint action value model is proven to be the sum of any temporary agents' preferences over a long horizon.
    
    The main contributions of this paper can be summarized as follows: (1) We conceptualize OAHT as a dynamic coalitional affinity game, OSB-CAG. In this model, the learner seeks to influence teammates through its actions, without prior coordination, to establish a stable temporary team. (2) The theoretical model of OSB-CAG gives an understanding of GPL's joint action value model. It ensures collaboration within any temporary team under open team settings. (3) Building on the OSB-CAG theory, we derive a constraint for representing the joint action value to facilitate learning, and an additional regularization term depending on the graph structure to rationalize solving DVSC as an RL problem. The novel algorithm, named CIAO (\textbf{C}ooperative game theory \textbf{I}nspired \textbf{A}d hoc teamwork in \textbf{O}pen teams), is implemented based on GPL and incorporates the above novel and provable tricks. (4) We discuss and understand the learning paradigm employed in GPL that aims to learn the joint Q-value for open team settings. (5) We conduct experiments, primarily comparing two instances of CIAO (CIAO-S and CIAO-C) based on GPL framework in two environments: Level-based Foraging (LBF) and Wolfpack under open team settings \citep{rahman2021towards}. Finally, we conduct a comprehensive review and discussion of related works on both theoretical and algorithmic aspects of AHT and explore its relationship to MARL in Appendix~\ref{sec:related_work}.
    
\section{Background}
\label{sec:background}
    Let $\Delta(\Omega)$ indicate the set of probability distributions over a random variable on a sample space $\Omega$ and let $\mathbb{P}(\mathcal{X})$ denote the power set of an arbitrary set $\mathcal{X}$. To simplify the notation, let $i$ exclusively denote the learner and $-i$ denote the set of all temporary teammates at any timestep. $P(\mathcal{X})$ indicates the generic probability distribution over a random variable $\mathcal{X}$ and $|\mathcal{X}|$ indicates the cardinality of an arbitrary set $\mathcal{X}$.
    
    \subsection{Coalitional Affinity Game}
    \label{subsec:hedonic_graph_games}
        As a subclass of non-transferable utility games, \textit{hedonic game} \citep{chalkiadakis2022computational} is defined as a tuple $\langle \mathcal{N}, \succeq \rangle$, where $\mathcal{N}$ is a set of all agents; and $\succeq = (\succeq_{1}, ..., \succeq_{n})$ is a sequence of agents' preferences over the subsets of $\mathcal{N}$ called coalitions. $\mathcal{C} \succeq_{j} \mathcal{C}'$ implies that coalition $\mathcal{C}$ is no less preferred by agent $j$ than coalition $\mathcal{C}'$. For each agent $j \in \mathcal{N}$, $\succeq_{j}$ describes a complete and transitive preference relation over a collection of all feasible coalitions $\mathcal{N}(j) = \{ \mathcal{C} \ \mathlarger{\mathlarger{\subseteq}} \ \mathcal{N} \ \vert \ j \in \mathcal{C} \}$. The outcome of a hedonic game is a coalition structure $\mathcal{CS}$, i.e., a partition of $\mathcal{N}$ into disjoint coalitions. We denote by $\mathcal{CS}(j)$ the coalition including agent $j$. The ordinal preferences can be represented as the cardinal form with preference values \citep{sliwinski2017learning}. More specifically, an agent $j$ has a preference value function such that $v_{j}: \mathcal{N}(j) \rightarrow \mathbb{R}_{\geq 0}$. $v_{j}(\mathcal{C}) \geq v_{j}(\mathcal{C}')$ if $\mathcal{C} \succeq_{j} \mathcal{C}'$, which implies that agent $j$ weakly prefers $\mathcal{C}$ to $\mathcal{C}'$; $v_{j}(\mathcal{C}) > v_{j}(\mathcal{C}')$ if $\mathcal{C} \succ_{j} \mathcal{C}'$, which implies that agent $j$ strictly prefers $\mathcal{C}$ to $\mathcal{C}'$.

        To concisely represent the preference value, a hedonic game is equipped with an affinity graph $G = \langle \mathcal{N}, \mathcal{E} \rangle$, where each edge $(j, k) \in \mathcal{E}$ describes an affinity relation between agents $j$ and $k$. For each edge $(j, k)$, it defines an \textit{affinity weight} $w(j, k) \in \mathbb{R}$ to indicate the value that agent $j$ can receive from agent $k$, while if $(j, k) \notin \mathcal{E}$, $w(j, k) = 0$. For any coalition $\mathcal{C} \ \mathlarger{\mathlarger{\subseteq}} \ \mathcal{N}_{j}$, the preference value of agent $j$ is specified as $v_{j}(\mathcal{C}) = \sum_{(j,k) \in \mathcal{E}, k \in \mathcal{C}} w(j, k)$ if $\mathcal{C} \neq \{j\}$, otherwise, $v_{j}(\{j\}) = b_{j} \in \mathbb{R}_{\geq 0}$.\footnote{In the original CAG setting \citep{sliwinski2017learning}, $v_{j}(\{j\})$ is conventionally set to zero. Herein, we extend it to non-negative values for generality (see Appendix~\ref{sec:generalization_of_preference_values}).} \textit{An affinity graph is symmetric if $w(j, k) = w(k, j)$, for all $(j,k), (k,j) \in \mathcal{E}$.} The hedonic game with an affinity graph is named as \textit{coalitional affinity game} (CAG) \citep{branzei2009coalitional}. Strict core stability is a principal solution concept of CAG (see Definition \ref{def:stability_concepts}).
        \begin{definition}
        \label{def:stability_concepts}
            We say that a blocking coalition $\mathcal{C}$ \textit{weakly blocks} a coalition structure $\mathcal{CS}$ if every agent $j \in \mathcal{C}$ weakly prefers $\mathcal{C}$ to $\mathcal{CS}(j)$ and there exists at least one agent $k \in \mathcal{C}$ who strictly prefers $\mathcal{C}$ to $\mathcal{CS}(j)$. A coalition structure admitting no weakly blocking coalition $\mathcal{C} \ \mathlarger{\mathlarger{\subseteq}} \ \mathcal{N}$ is called strict core stable.
        \end{definition}
        
    \subsection{Graph-Based Policy Learning}
    \label{subsec:graph-based_policy_optimization}
        We now briefly review GPL's empirical framework \citep{rahman2021towards} to solve OAHT (see Appendix~\ref{sec:gpl_framework} for more details). GPL consists of the following modules: the type inference model, the joint action value model and the agent model. To align with our motivation, we transform the framework to be adaptable to any coordination graph structure, as opposed to being restricted to only the complete graph as in GPL.

        \textbf{Type Inference Model.} This is modelled as a LSTM \citep{hochreiter1997long} to infer agent-types of a team at timestep $t$ given the teammates' agent-types and the state at timestep $t-1$. The agent-type is modelled as a fixed-length hidden-state vector of LSTM, referred to as agent-type embedding. To address the issue of variable team size, the embedding of the agents who leave a team would be removed at each timestep, while the type embedding of the newly added agents would be set to a zero vector.
    
        \textbf{Joint Action Value Model.} The joint Q-value $\hat{Q}^{\pi^{i}}(s_{t}, a_{t})$ is approximated as the sum of the individual utility $\hat{Q}_{j}^{\pi^{i}}(a_{t}^{j}, \vert s_{t})$ and pairwise utility $\hat{Q}_{jk}^{\pi^{i}}(a_{t}^{j}, a_{t}^{k} \vert s_{t})$:
        \begin{equation}
        \label{eq:joint_q-value_model}
            \hat{Q}^{\pi^{i}}(s_{t}, a_{t}) = \sum_{(j, k) \in \mathcal{E}_{t}} \hat{Q}_{jk}^{\pi^{i}}(a_{t}^{j}, a_{t}^{k} \vert s_{t}) + \sum_{j \in \mathcal{N}_{t}} \hat{Q}_{j}^{\pi^{i}}(a_{t}^{j} \vert s_{t}),
        \end{equation}
        where the superscript $\pi^{i}$ implies that the above terms can only be optimized by the learner's policy $\pi^{i}$.
        
        \textbf{Agent Model.} To address the open team setting, GNN is applied to process the joint agent-type embedding $\theta_{t}$ produced from the type inference model, where each agent is represented as a node and the coordination graph is consistent with that for the joint action value model. The resulting node representation $\bar{n}_{t}$ is applied as input to infer the estimated teammates' joint policy, denoted as $\hat{\pi}^{-i}(a_{t}^{-i} \vert s_{t}, \theta_{t}^{-i})$.
        
        \textbf{Learner's Decision Making.} The learner's approximate action value function $\hat{Q}^{\pi^{i}}(s_{t}, a_{t}^{i})$ is defined as follows:
        \begin{equation}
        \label{eq:objective_ad_hoc}
            \hat{Q}^{\pi^{i}}(s_{t}, a_{t}^{i}) = \mathbb{E}_{a_{t}^{-i} \sim \pi_{t}^{-i}} \left[ \hat{Q}^{\pi^{i}}(s_{t}, a_{t}^{i}, a_{t}^{-i}) \right],
        \end{equation}
        where $s_{t}$ is a state at timestep $t$, $a_{t}^{-i}$ is a joint action of teammates $-i$ at timestep $t$ and $a_{t}^{i}$ is the learner $i$'s action at timestep $t$. The learner's decision making is conducted by selecting the action that maximizes $\hat{Q}^{\pi^{i}}(s_{t}, a_{t}^{i})$. 
        
\section{A New Game Model to Formalize OAHT}
\label{sec:open_stochastic_bayesian_coalition_affinity_games }
    In this section, we generalize the coalitional affinity game framework to formalize OAHT, by integrating a graph to represent relationships among agents. It is essential to emphasize that, for the sake of brevity, our focus of this work is exclusively on fully observable scenarios. 
    \subsection{Problem Formulation}
    \label{subsec:problem_formulation}
        In an environment, the learner \(i\) interacts with other uncontrollable temporary teammates \(-i\) to achieve a shared goal. To model this process, we introduce Open Stochastic Bayesian Coalitional Affinity Game (OSB-CAG), defined as a tuple \(\langle \mathcal{N}, \mathcal{S}, (\mathcal{A}_{j})_{j \in \mathcal{N}}, \Theta, (R_{j})_{j \in {\scriptscriptstyle \mathcal{N}}}, P_{T}, P_{I}, P_{A}, \mathcal{E}, \gamma \rangle\). Here, \(\mathcal{N}\) represents the set of all possible agents, \(\mathcal{S}\) is the set of states, \(\mathcal{A}_{j}\) is the action set for agent \(j\), and \(\Theta\) denotes the set of all possible agent-types. Let the joint action set under a variable agent set \(\mathcal{N}_{t} \subseteq \mathcal{N}\) be defined as \(\mathcal{A}_{{\scriptscriptstyle \mathcal{N}}_{t}} = \times_{j \in \mathcal{N}_{t}} \mathcal{A}_{j}\). Therefore, the joint action space under a variable number of agents is defined as \(\mathcal{A}_{\scriptscriptstyle\mathcal{N}} = \bigcup_{\mathcal{N}_{t} \in \mathbb{P}(\mathcal{N})} \{a \vert a \in \mathcal{A}_{{\scriptscriptstyle \mathcal{N}}_{t}} \}\), while the joint agent-type space under a variable number of agents is defined as \(\Theta_{\scriptscriptstyle\mathcal{N}} = \bigcup_{\mathcal{N}_{t} \in \mathbb{P}(\mathcal{N})} \{\theta \vert \theta \in \Theta^{{\scriptscriptstyle |\mathcal{N}_{t}|}} \}\). A dynamic affinity graph, denoted as \(G_{t} = \langle \mathcal{N}_{t}, \mathcal{E}_{t} \rangle\), is introduced to describe the relationships among agents. Here, \(\mathcal{E}_{t} = \{ (j, k) \ \vert \ j, k \in \mathcal{N}_{t}\} \ \mathlarger{\mathlarger{\subseteq}} \ \mathcal{E}\), and \(\mathcal{E}\) is a set of possible edges represented by pairs \((j, k)\). This graph is referred to as the coordination graph in GPL.

        \textbf{Transition Function.} We now introduce three primitive probability distributions denoted as $P_{T}: \mathbb{P}(\mathcal{N}) \times \mathcal{S} \times \mathcal{A}_{\scriptscriptstyle\mathcal{N}} \rightarrow \Delta(\mathbb{P}(\mathcal{N}) \times \mathcal{S})$, $P_{I}: \mathbb{P}(\mathcal{N}) \times \mathcal{S} \rightarrow [0, 1]$, and $P_{A}: \mathcal{N} \times \mathcal{S} \rightarrow \Delta(\Theta)$. These probability functions characterize the dynamics of the environment in the following procedure: (1) At the initial timestep $0$, $P_{I}(\mathcal{N}_{0}, s_{0})$ generates an initial set of agents $\mathcal{N}_{0}$ and an initial state $s_{0}$. (2) $P_{A}(\theta_{t}^{j} \vert \{j\}, s_{t})$ represents a type assignment function that randomly assigns agent-types to the generated agent set. (3) $P_{T}(\mathcal{N}_{t}, s_{t} \vert \mathcal{N}_{t-1}, s_{t-1}, a_{t-1})$ generates the agent set $\mathcal{N}_{t}$ and state $s_{t}$ for the next time step $t$. (4) Stage 2 and 3 above are repeated. To succinctly represent the aforementioned process, we derive a composite transition function $T(\mathcal{N}_{t}, s_{t}, \theta_{t} \vert s_{t-1}, a_{t-1}, \theta_{t-1})$ (see Proposition \ref{prop:existence_transition_preference}) in place of stage 2 and 3 from timesteps $t \geq 1$. This function can be factorized, clarifying the GPL's framework, as follows:
        \begin{equation}
        \label{eq:preference_trans_decomposition}
        \begin{split}
            T(\mathcal{N}_{t}, &s_{t}, \theta_{t} \vert s_{t-1}, a_{t-1}, \theta_{t-1}) \\
            &= P_{E}(\theta_{t} \vert \mathcal{N}_{t}, s_{t}) P_{O}(\mathcal{N}_{t}, s_{t} \vert s_{t-1}, a_{t-1}, \theta_{t-1}).
        \end{split}
        \end{equation}
        Herein, $P_{O}(\mathcal{N}_{t}, s_{t} \vert s_{t-1}, a_{t-1}, \theta_{t-1})$ is a probability distribution composed of $P_{T}$, $P_{I}$ and $P_{A}$ (see the sketch of proof of Proposition~\ref{prop:existence_transition_preference}) that generates a variable agent set $\mathcal{N}_{t}$ and a state $s_{t}$, \textit{observable} by the learner. In contrast, a joint agent-type $\theta_{t}$ generated from $P_{E}(\theta_{t} \vert \mathcal{N}_{t}, s_{t}) = \prod_{j=1}^{|\mathcal{N}_{t}|} P_{A}(\theta_{t}^{j} \vert \{j\}, s_{t})$ is \textit{unobservable} by the learner. However, it plays a crucial role in the agent model for the learner's decision making in the empirical framework of GPL, motivating the estimation of this term in practice, as conducted by the type inference model (see Section~\ref{subsec:graph-based_policy_optimization}). To distinguish between and clarify the observation generated from $P_{O}$ and the agent-types generated from $P_{E}$ during the decision process, both functions will be concurrently utilized to describe the composite transition function $T$ in the subsequent sections. To simplify the notation, we would use $P_{O}$ in place of $P_{I}$ for $t=0$ in the following sections.
        \begin{assumption}
        \label{assm:conditional_independencies}
            The following conditional independencies are assumed to hold in any distribution $P$ over the set of variables in an OSB-CAG: (1) $( \theta_{t} \bigCI \theta_{t-1}, s_{t-1}, a_{t-1} \ \vert \ \mathcal{N}_{t}, s_{t} )$; (2) $(\mathcal{N}_{t}, s_{t} \bigCI \theta_{t-1} \ \vert \ \mathcal{N}_{t-1}, s_{t-1}, a_{t-1})$; (3) $(\mathcal{N}_{t} \bigCI a_{t} \vert s_{t}, \theta_{t})$; (4) $( \theta_{t}^{j} \bigCI -j, \theta_{t}^{-j} \ \vert \ \{j\}, s_{t} )$.
        \end{assumption}    
        \begin{proposition}
        \label{prop:existence_transition_preference}
            $T(\mathcal{N}_{t}, s_{t}, \theta_{t} \vert s_{t-1}, a_{t-1}, \theta_{t-1})$ for $t \geq 1$ can be expressed in terms of the following well-defined probability distributions: $P_{I}(\mathcal{N}_{0}, s_{0})$, $P_{T}(\mathcal{N}_{t}, s_{t} \vert \mathcal{N}_{t-1}, s_{t-1}, a_{t-1})$ for $t \geq 1$, and $P_{A}(\theta_{t}^{j} \vert \{j\}, s_{t})$ for $t \geq 0$.
        \end{proposition}
        \begin{proof}
            We show the sketch of proof here. The following derivation is obtained by Assumption~\ref{assm:conditional_independencies}. For validity of conditions in Assumption~\ref{assm:conditional_independencies}, please refer to Appendix~\ref{sec:appendix_assumptions}. About the complete version of proof, please refer to Appendix \ref{subsec:proof_of_proposition1}.
            \begin{equation*}
                \begin{split}
                    T(\mathcal{N}_{t}, s_{t}, \theta_{t} \vert s_{t-1}, a_{t-1}, \theta_{t-1})
                    &= \\
                    P_{E}(\theta_{t} \vert \mathcal{N}_{t}, s_{t}&) P_{O}(\mathcal{N}_{t}, s_{t} \vert s_{t-1}, a_{t-1}, \theta_{t-1}),
                \end{split}
            \end{equation*}
            where $P_{E}(\theta_{t} \vert \mathcal{N}_{t}, s_{t}) = \prod_{j=1}^{|\mathcal{N}_{t}|} P_{A}(\theta_{t}^{j} \vert \{j\}, s_{t})$ and 
            \begin{equation*}
                \begin{split}
                    P_{O}(&\mathcal{N}_{t}, s_{t} \vert s_{t-1}, a_{t-1}, \theta_{t-1}) = \\
                    &\sum_{\scriptscriptstyle{\mathcal{N}}_{t-1}} P_{T}(\mathcal{N}_{t}, s_{t} \vert \mathcal{N}_{t-1}, s_{t-1}, a_{t-1}) P(\mathcal{N}_{t-1} \vert s_{t-1}, \theta_{t-1}).
                \end{split}
            \end{equation*}
            We have
            \begin{equation*}
                P(\mathcal{N}_{t} | s_{t}, \theta_{t}) = \frac{\sum_{s_{t}} P_{E}(\theta_{t} \vert \mathcal{N}_{t}, s_{t}) P(\mathcal{N}_{t}, s_{t})}{\sum_{\scriptscriptstyle{\mathcal{N}}_{t}} \sum_{s_{t}} P_{E}(\theta_{t} \vert \mathcal{N}_{t}, s_{t}) P(\mathcal{N}_{t}, s_{t}) }.
            \end{equation*}
            Also, we have $P(\mathcal{N}_{0}, s_{0}) = P_{I}(\mathcal{N}_{0}, s_{0})$ and when $t \geq 1$,
            \begin{equation*}
                \begin{split}
                    P(&\mathcal{N}_{t}, s_{t}) \\
                    &= \sum_{\scriptscriptstyle{\mathcal{N}}_{t-1}} \sum_{s_{t-1}} \sum_{a_{t-1}} \sum_{\theta_{t-1}} P(\mathcal{N}_{t}, s_{t}, \mathcal{N}_{t-1}, s_{t-1}, a_{t-1}, \theta_{t-1}),
                \end{split}
            \end{equation*}
            where $P(\mathcal{N}_{t}, s_{t}, \mathcal{N}_{t-1}, s_{t-1}, a_{t-1}, \theta_{t-1})$ can be expressed in terms of $P(\mathcal{N}_{t-1}, s_{t-1})$ and the probability distributions we have defined. The sketch of proof is completed.
        \end{proof}
        
        \textbf{Preference Reward.} The function $R_{j}: \mathcal{A}_{\scriptscriptstyle \mathcal{N}} \times \mathcal{S} \rightarrow \mathbb{R}_{\geq 0}$ extends an agent $j$'s preference value, of the original stateless CAG, to the agent $j$’s preference reward $R_{j}$ which depends on the state and action. For example, $R_{j}(a_{t} \vert s_{t})$ indicates agent $j$'s preference reward for a temporary team $\mathcal{N}_{t} \ \mathlarger{\mathlarger{\subseteq}} \ \mathcal{N}$ with the corresponding joint action $a_{t} = \times_{j \in {\scriptscriptstyle \mathcal{N}_{t}}} a_{t}^{j}$, whereas $R_{j}(a_{t}^{j} \vert s_{t})$ indicates agent $j$'s preference reward for a coalition only including itself. To capture the relationship between agents \(j\) and \(k\) in terms of both the current state and the actions taken, the affinity weight is generalized accordingly as $w_{jk}:  \mathcal{A}_{j} \times \mathcal{A}_{k} \times \mathcal{S} \rightarrow \mathbb{R}$. Following the specification of preference values through affinity weights, the preference reward of any agent \(j\) for a coalition \(\mathcal{N}_{t}\) can be represented as \(R_{j}(a_{t} \vert s_{t}) = \sum_{(j,k) \in \mathcal{E}_{t}, k \in \mathcal{N}_{t}} w_{jk}(a_{t}^{j}, a_{t}^{k} \vert s_{t})\). This summation aggregates the affinity weights for all pairs of agents \((j,k)\) in the coalition, where \(k\) is a member of \(\mathcal{N}_{t}\). The learner's reward function $R(s_{t}, a_{t})$ for any $\mathcal{N}_{t}$ is specified by $R_{j}(a_{t} \vert s_{t})$, which will be introduced in Section~\ref{subsec:bayesian_core_ as_rl_obj}.
        
    \subsection{Dynamic Variational Strict Core}
    \label{subsec:dynamic_variational_Strict_core}
        We now extend the game theoretical concept of strict core from CAG to OSB-CAG as a criterion to evaluate the extent of collaboration among the agents in a temporary team (a coalition $\mathcal{N}_{t}$ at each timestep $t$), named as \textit{dynamic variational strict core} (DVSC). Unlike the strict core defined in CAG that evaluates coalition formation based on the given preference values, DVSC evaluates whether the learner $i$'s policy can influence temporary teammates' decisions (measured by preference rewards), so that they intend to collaborate (so called variational). This is analogous to forming a temporary team as a desired coalition. Next we derive a result on strict core stability to motivate a result on DVSC. The following two statements are equivalent when the affinity graph is \textit{symmetric}: Team maximizes social welfare, and team reaches strict core stability (see Lemma \ref{lemm:symmetry_strict_core_grand_coalition} in Appendix \ref{sec:derivation_details_of_definition_DVSC}). This inspires using the objective of maximizing social welfare as a surrogate criterion to evaluate strict core stability, and this criterion can be further generalized to dynamic scenarios to derive the DVSC (see Definition \ref{def:symmetry_strict_core_stability}).
        \begin{definition}
        \label{def:symmetry_strict_core_stability}
            If a dynamic affinity graph is symmetric, then maximizing the long-horizon social welfare is equivalent to reaching strict core stability under the variable teammates of uncertain agent-types generated by $P_{E}$ and the uncertain states generated by $P_{O}$.
        \end{definition}
        
        Following the inspiration shown in Definition \ref{def:symmetry_strict_core_stability}, DVSC can be equivalently expressed in the form shown in Eq.~\eqref{eq:dsvc}. The detailed derivation of DVSC is left in Appendix \ref{sec:derivation_details_of_definition_DVSC}.
        \begin{equation}
        \label{eq:dsvc}
           \begin{split}
               \texttt{DVSC} &:= \Big\{ \ \pi^{i,*} \ \Big\vert \ \mathbb{E}_{\pi^{i,*}}\big[ \sum_{t=0}^{\infty} \gamma^{t} \sum_{j \in \mathcal{N}_{t}} R_{j}(a_{t} \vert s_{t}) \big] \\ 
               &\geq \mathbb{E}_{\pi^{i}}\big[ \sum_{t=0}^{\infty} \gamma^{t} \sum_{j \in \mathcal{N}_{t}} R_{j}(a_{t} \vert s_{t}) \big], \forall s_{0} \in \mathcal{S}, \forall \pi^{i} \ \Big\},
           \end{split}
        \end{equation}
        where $a_{t}^{i} \sim \pi^{i}$ and $a_{t}^{-i} \sim \pi_{t}^{-i}$; $\mathbb{E}_{\pi^{i}}[\cdot]$ denotes the expectation that also implicitly depends on $\theta_{t} \sim P_{E}$ and $\mathcal{N}_{t}, s_{t} \sim P_{O}$, and $a_{t}^{-i} \sim \pi_{t}^{-i}$; and $\pi^{i,*}$ indicates the solution to DVSC.

    \subsection{Is Stability of any Temporary Team a Reasonable Metric for Describing Open Ad Hoc Collaboration?}
    \label{subsec:insight_behind_describing_collaboration_among_temporary_teammates}
        Recall that all agents in AHT have a shared goal, which implies that they intrinsically aim to collaborate on solving a shared task \citep{mirsky2022survey}, but their preferences for collaborating with each other are not necessarily compatible. This compatibility can be interpreted as stability of a temporary team, determined by the preferences of ad hoc agents for collaborating with each other. If those ad hoc agents are incompatible with each other, the temporary team becomes unstable but still with hope of collaborating as a team to solve the shared task. Thus, the learner's aim is to tweak the compatibility of a temporary team through its actions, to influence the temporary teammates' preferences, equivalent to maintaining the stability of the temporary team, across timesteps. Furthermore, the definition of DVSC in Eq.~\eqref{eq:dsvc} is invariant to team size, aligned to open team settings.
    
    \subsection{Solving DVSC by Reinforcement Learning}
    \label{subsec:bayesian_core_ as_rl_obj}
        We proceed to define the learner's reward function, initially left blank in Section~\ref{subsec:problem_formulation} and convert DVSC from Eq.~\eqref{eq:dsvc} into an RL problem. Since the learner's objective is to execute actions that influence any temporary teammates to collaboratively solve a shared task, we naturally interpret the learner's reward function as $R(s_{t}, a_{t}) = \sum_{j \in \mathcal{N}_{t}} R_{j}(a_{t} \vert s_{t})$. The reward function represents the social welfare of preference rewards for a temporary team $\mathcal{N}_{t}$, serving as a metric to measure agents' preferences for collaborating on a shared task.\footnote{In practical scenarios, $R(s_{t}, a_{t})$ only needs to implicitly encode the shared goal that multiple agents are required to achieve.} Substituting $R(s_{t}, a_{t})$ into Eq.~\eqref{eq:dsvc}, we derive an RL problem equivalent to solving DVSC:
        \begin{equation}
        \label{eq:rl_obj}
            \max_{\pi_{i}} \mathbb{E}_{{\scriptscriptstyle \mathcal{N}}_{t}, s_{t} \sim P_{O}, \theta_{t} \sim P_{E}, a_{t}^{-i} \sim \pi_{t}^{-i}, a_{t}^{i} \sim \pi^{i}} \Big[ \sum_{t=0}^{\infty} \gamma^{t} R(s_{t}, a_{t}) \Big]. 
        \end{equation}
        In the following section, we will explore how the optimization problem in Eq.~\eqref{eq:rl_obj} can be solved by a novel algorithm.

\section{A Novel Algorithm Building on OSB-CAG}
\label{sec:bayesian_core_ as_rl_obj}
    In this section, we derive a novel graph-based RL algorithm to solve OAHT based on the OSB-CAG, with DVSC as a solution concept. We first derive the joint Q-value's representation to narrow down its hypothesis space including the solution of DVSC. The representation aligns with and gives an interpretation to the GPL's heuristic joint action value model. Note that we also acquire a condition to further confine the joint Q-value's hypothesis space thanks to our theory (see Section~\ref{subsec:representation_of_preference_q-values}). With the estimated type inference model and agent model, the optimal learner's policy obtained by GPL's optimization problem approximately reaches DSVC (see Section~\ref{subsec:open_ad_hoc_bellman_optimality_equation}). Finally, we derive a novel practical algorithm, named CIAO (see Section~\ref{subsec:practical_implementation}).
    
    \subsection{Representation of Joint Q-Value}
    \label{subsec:representation_of_preference_q-values}
        \begin{figure}[ht!]
            \centering
            \includegraphics[width=0.95\columnwidth]{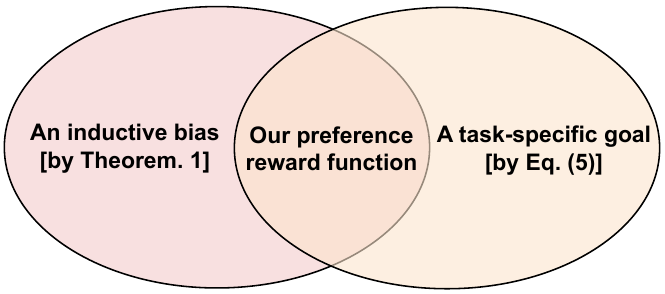}
            \caption{Illustration of the relationship between the conditions for our preference reward function, ensuring the existence of DVSC under its confined hypothesis space, and its alignment to a task-specific reward $R(s_{t}, a_{t})$ in Eq.~\eqref{eq:rl_obj}.}
        \label{fig:preference_rewards}
        \end{figure}
        Given the joint actions generated under the influence by the optimal learner's policy $\pi^{i,*}$, we have a sufficient condition, as an inductive bias, for \textit{any preference reward function} to narrow down its hypothesis space meeting DVSC in Theorem \ref{thm:dvsc_core_existence}. Solving the RL problem outlined in Eq.~\eqref{eq:rl_obj} based on this condition to specify $\pi^{i,*}$, the preference reward function is aligned to a task-specific reward $R(s_{t}, a_{t})$. The relationship between the above conditions to generate our preference reward function is shown in Fig.~\ref{fig:preference_rewards}.
        \begin{theorem}
        \label{thm:dvsc_core_existence}
            In an OSB-CAG, for any dynamic affinity graph $G_{t} = \langle \mathcal{N}_{t}, \mathcal{E}_{t} \rangle$ at any timestep $t$, if there exists a joint action $a_{t} \in \mathcal{A}_{\scriptscriptstyle \mathcal{N}_{t}}$, for any agent $j \in \mathcal{N}_{t}$, satisfying $R_{j}(a_{t} \vert s_{t}) \geq R_{j}(a_{t}^{j} \vert s_{t})$ for any $s_{t} \in \mathcal{S}$, then DVSC always exists.
        \end{theorem}

        To meet the condition that $R_{j}(a_{t} \vert s_{t}) \geq R_{j}(a_{t}^{j} \vert s_{t})$ as shown in Theorem~\ref{thm:dvsc_core_existence}, we derive a representation of $w_{jk}(a_{t}^{j}, a_{t}^{k} \vert s_{t})$ in Proposition~\ref{prop:affinity_weight_representation}. Recall that an agent $j$'s preference reward function for a temporary team $\mathcal{N}_{t}$ at timestep $t$ is defined as $R_{j}(a_{t} \vert s_{t}) = \sum_{(j,k) \in \mathcal{E}_{t}} w_{jk}(a_{t}^{j}, a_{t}^{k} \vert s_{t})$ (see Section~\ref{subsec:problem_formulation}). 
        \begin{proposition}
        \label{prop:affinity_weight_representation}
            In a dynamic affinity graph $G_{t} = \langle \mathcal{N}_{t}, \mathcal{E}_{t} \rangle$, for any state $s_{t} \in \mathcal{S}$ and any joint action $a_{t} \in \mathcal{A}_{\scriptscriptstyle \mathcal{N}_{t}}$, if for all $(j, k) \in \mathcal{E}_{t}$, $w_{jk}(a_{t}^{j}, a_{t}^{k} \vert s_{t}) = \alpha_{jk}(a_{t}^{j}, a_{t}^{k} \vert s_{t}) + \beta_{jk}(a_{t}^{j} \vert s_{t})$ with the conditions that $\alpha_{jk}(a_{t}^{j}, a_{t}^{k} \vert s_{t}) \geq 0$ and $R_{j}(a_{t}^{j} \vert s_{t}) = \sum_{(j,k) \in \mathcal{E}_{t}} \beta_{jk}(a_{t}^{j} \vert s_{t})$, then $R_{j}(a_{t} \vert s_{t}) \geq R_{j}(a_{t}^{j} \vert s_{t})$ for any agent $j \in \mathcal{N}_{t}$.
        \end{proposition}
        \begin{proof}
            This result can be directly obtained by the definition that $R_{j}(a_{t} \vert s_{t}) = \sum_{(j,k) \in \mathcal{E}_{t}, k \in \mathcal{N}_{t}} w_{jk}(a_{t}^{j}, a_{t}^{k} \vert s_{t})$.
        \end{proof}
        
        Plugging in the expression of $w_{jk}(a_{t}^{j}, a_{t}^{k} \vert s_{t})$, we can obtain the representation of an arbitrary agent $j$'s preference Q-value under the learner's optimal policy $\pi^{i,*}$,
        $Q_{j}^{\pi^{i,*}}(a_{t} \vert s_{t}) = \sum_{(j,k) \in \mathcal{E}_{t}} Q_{jk}^{\pi^{i,*}}(a_{t}^{j}, a_{t}^{k} \vert s_{t}) + Q_{j}^{\pi^{i,*}}(a_{t}^{j} \vert s_{t})$, and the joint Q-value under the learner's optimal policy $\pi^{i,*}$, $Q^{\pi^{i,*}}(s_{t}, a_{t}) = \sum_{j \in \mathcal{N}_{t}} Q_{j}^{\pi^{i,*}}(a_{t} \vert s_{t})$, outlined in Theorem~\ref{thm:joint_q_representation}.
        \begin{assumption}
        \label{assm:agent_leaves_env}
            Suppose that $\alpha_{jk}(a_{t}^{j}, a_{t}^{k} \vert s_{t}) = 0$ for $t \geq T$, where $T$ is the timestep when agent $j$ or $k$ leaves the environment, and $R_{j}(a_{t}^{j} \vert s_{t})=0$ for $t \geq T'$, where $T'$ is the timestep when agent $j$ leaves the environment.
        \end{assumption}
        
        \begin{theorem}
        \label{thm:joint_q_representation}
            Under Assumption \ref{assm:agent_leaves_env}, if $w_{jk}(s_{\tau}, a_{\tau}^{j}, a_{\tau}^{k}) = \alpha_{jk}(s_{\tau}, a_{\tau}^{j}, a_{\tau}^{k}) + \beta_{jk}(s_{\tau}, a_{\tau}^{j})$, then the joint Q-value of the learner's policy $\pi^{i}$ can be expressed as follows:
            \begin{equation*}
                \begin{split}
                    Q^{\pi^{i}}(s_{t}, a_{t}) &= \sum_{(j,k) \in \mathcal{E}_{t}} Q_{jk}^{\pi^{i}}(a_{t}^{j}, a_{t}^{k} \vert s_{t}) + \sum_{j \in \mathcal{N}_{t}} Q_{j}^{\pi^{i}}(a_{t}^{j} \vert s_{t}) \\
                    = \sum_{j \in \mathcal{N}_{t}} &\Big\{ \sum_{(j,k) \in \mathcal{E}_{t}} Q_{jk}^{\pi^{i}}(a_{t}^{j}, a_{t}^{k} \vert s_{t}) + Q_{j}^{\pi^{i}}(a_{t}^{j} \vert s_{t}) \Big\} \\
                    := \sum_{j \in \mathcal{N}_{t}} &Q_{j}^{\pi^{i}}(a_{t} \vert s_{t}),
                \end{split}
            \end{equation*}
            where $Q_{jk}^{\pi^{i}}(a_{t}^{j}, a_{t}^{k} \vert s_{t}) = \mathbb{E}_{\pi^{i}}[\sum_{\tau=t}^{\infty} \gamma^{\tau - t} \alpha_{jk}(a_{\tau}^{j}, a_{\tau}^{k} \vert s_{\tau})]$ and $Q_{j}^{\pi^{i}}(a_{t}^{j} \vert s_{t}) = \mathbb{E}_{\pi^{i}}[\sum_{\tau=t}^{\infty} \gamma^{\tau - t} R_{j}(a_{\tau}^{j} \vert s_{\tau})]$.
        \end{theorem}
        \begin{remark}
        \label{rmk:equivalence_to_GPL}
            The result of Theorem \ref{thm:joint_q_representation} verifies that the optimal joint Q-value representation derived from our theory is consistent with the GPL's joint action value model, as shown in Eq.~\eqref{eq:joint_q-value_model}, but additionally with $\hat{Q}_{jk}^{\pi^{i}}(a_{t}^{j}, a_{t}^{k} \vert s_{t}) \geq 0$, following our theory, which is requisite for satisfying $\alpha_{jk}(a_{t}^{j}, a_{t}^{k} \vert s_{t}) \geq 0$, as shown in Proposition~\ref{prop:affinity_weight_representation}.
        \end{remark}
        
        Recall that the condition for solving DSVC as a RL problem is \textit{the symmetry of a dynamic affinity graph} (see Definition \ref{def:symmetry_strict_core_stability}). To meet this condition, we outline in Proposition \ref{prop:condition_of_symmetry_star_graph} the constraints that must be fulfilled for the case of a dynamic affinity graph being a star graph (see Remark~\ref{rmk:ad_hoc_affinity_graph} for its validity in OAHT). Similarly, we provide the relevant constraints articulated in Proposition \ref{prop:condition_of_symmetry_full_graph} for situations where the dynamic affinity graph takes the form of a complete graph (as applied to GPL). The implementation of the constraints for these two cases are shown in Remark \ref{rmk:implement_symmetry_dynamic_affinity_graph}.
        \begin{definition}
        \label{def:ad_hoc_affinity_graph}
            In this paper, we introduce a novel dynamic affinity graph structured as a star graph, with the learner serving as the internal node and temporary teammates as the leaf nodes.
        \end{definition}
        
        \begin{remark}
        \label{rmk:ad_hoc_affinity_graph}
            We introduce a novel architecture for the dynamic affinity graph in the context of OAHT, assuming teammates lack prior coordination \citep{mirsky2022survey}. Given an additional assumption that teammates cannot adapt their policies or types in response to other agents,\footnote{For simplicity in presenting our theory in this paper, we tentatively disregard scenarios where temporary teammates can adapt to other agents (e.g. establishing an affinity model).} it is reasonable to presume the absence of relationships among any temporary teammates. Besides, this is also in line with the assumption in AHT that the learner's temporary teammates might not be familiar with one another before the interaction \citep{stone2010ad,mirsky2022survey}. In particular, this implies that no edges between any two teammates are necessary to form a dynamic affinity graph. However, the learner's goal is to establish collaboration with a variable number of temporary teammates at each timestep, necessitating the existence of edges between the learner and each teammate. To meet all these requirements, we design the learner's dynamic affinity graph as a star graph, as detailed in Definition \ref{def:ad_hoc_affinity_graph}. Consequently, the preference reward of any teammate $j$ for a temporary team $\mathcal{N}_{t}$ is determined as $R_{j}(s_{t}, a_{t}) = w_{ji}(s_{t}, a_{t}^{j}, a_{t}^{i})$, while the learner $i$'s preference reward for the temporary team $\mathcal{N}_{t}$ is expressed as $R_{i}(s_{t}, a_{t}) = \sum_{j \in -i} w_{ij}(s_{t}, a_{t}^{i}, a_{t}^{j})$.
        \end{remark}

        \begin{proposition}
        \label{prop:condition_of_symmetry_star_graph}
            For the learner $i$ and any teammate $j$ or $k$, the constraints $R_{i}(a_{t}^{i} \vert s_{t}) = \sum_{j \in -i} R_{j}(a_{t}^{j} \vert s_{t})$ and $\alpha_{jk}(a_{t}^{j}, a_{t}^{k} \vert s_{t}) = \alpha_{kj}(a_{t}^{k}, a_{t}^{j} \vert s_{t})$, for any $a_{t} \in \mathcal{A}_{\scriptscriptstyle \mathcal{N}_{t}}$ and $s_{t} \in \mathcal{S}$, are necessary for a star dynamic affinity graph to be symmetric. 
        \end{proposition}

        \begin{proposition}
        \label{prop:condition_of_symmetry_full_graph}
            For any two agents $j$ or $k$, the constraints $R_{j}(a_{t}^{j} \vert s_{t}) = R_{k}(a_{t}^{k} \vert s_{t})$ and $\alpha_{jk}(a_{t}^{j}, a_{t}^{k} \vert s_{t}) = \alpha_{kj}(a_{t}^{k}, a_{t}^{j} \vert s_{t})$, for any $a_{t} \in \mathcal{A}_{\scriptscriptstyle \mathcal{N}_{t}}$ and $s_{t} \in \mathcal{S}$, are necessary for the complete dynamic affinity graph to be symmetric.
        \end{proposition}

        \begin{remark}
        \label{rmk:implement_symmetry_dynamic_affinity_graph}
            The following implementation is necessary to satisfy the symmetry of a dynamic affinity graph: (1) meeting $Q_{jk}^{\pi^{i}}(a_{t}^{j}, a_{t}^{k} \vert s_{t}) = Q_{kj}^{\pi^{i}}(a_{t}^{k}, a_{t}^{j} \vert s_{t}) \geq 0$ in constructing preference Q-values; (2) If the dynamic affinity graph is a star graph with the learner as the internal node, $Q_{i}^{\pi^{i}}(a_{t}^{i} \vert s_{t}) = \sum_{j \in -i} Q_{j}^{\pi^{i}}(a_{t}^{j} \vert s_{t})$ is implemented as a regularizer. If the dynamic affinity graph is a complete graph, $Q_{i}^{\pi^{i}}(a_{t}^{i} \vert s_{t}) = Q_{j}^{\pi^{i}}(a_{t}^{j} \vert s_{t})$ is implemented as a regularizer.
        \end{remark}
        
    \subsection{Bellman Optimality Equation for OSB-CAG}
    \label{subsec:open_ad_hoc_bellman_optimality_equation}
        We now define the Bellman optimality equation for OSB-CAG to evaluate the learner $i$'s optimal policy $\pi^{i,*}$ as a solution of the DVSC following Theorem \ref{thm:open_team_ad_hoc_bellman_optimality}, such that
        \begin{equation}
        \label{eq:open_team_ad_hoc_bellman_optimality}
            \begin{split}
                Q^{\pi^{i,*}}(s_{t}, a_{t}) = R(s_{t}, a_{t}) + \gamma \mathbb{E}_{{\scriptscriptstyle \mathcal{N}_{t+1}}, s_{t+1} \sim P_{O}} \Big[ \\
                \max_{a^{i}} \mathbb{E}_{\substack{\theta_{t+1} \sim P_{E},\ a_{t+1}^{-i} \sim \pi_{t+1}^{-i}}} \big[ Q^{\pi^{i,*}}(s_{t+1}, a_{t+1}^{-i}, a^{i}) \big] \Big].
            \end{split}
        \end{equation}
        The regularity condition of Eq.~\eqref{eq:open_team_ad_hoc_bellman_optimality} is that $\mathcal{N}_{t+1} \subseteq \mathcal{N}_{t}$, since it is pathological to consider an agent $j \in \mathcal{N}_{t+1}$ but, $\notin \mathcal{N}_{t}$ at timestep $t$ when expanding $Q^{\pi^{i,*}}(s_{t}, a_{t})$ across timesteps, which is clarified in an illustrative example in Fig.~\ref{fig:bellman_optimality_equation_expansion}.
        \begin{figure}[ht!]
            \centering
            \includegraphics[width=0.95\columnwidth]{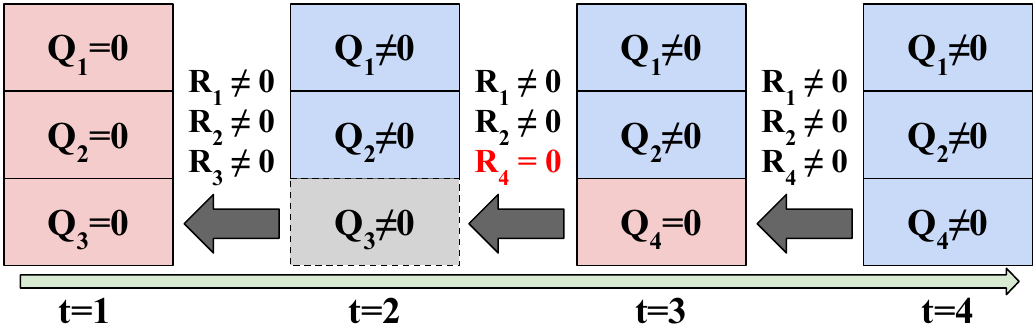}
            \caption{Illustration of the expansion of Bellman optimality equation for OSB-CAG. The thin green arrow indicates the time axis, while the thick black arrow indicates the expansion direction of Bellman optimality equation. In the theory of OSB-CAG, we have $Q^{\pi^{i,*}}(s_{t}, a_{t}) = \sum_{j \in \mathcal{N}_{t}} Q_{j}^{\pi^{i,*}}(a_{t} \vert s_{t})$, where $Q_{j}^{\pi^{i,*}}$ is denoted by $\mathbf{Q}_{j}$ in the figure. $R_{j}$ indicates the preference reward of agent $j$ to the team, measuring the agent's preference to stay in the team to solve the shared task. At each timestep, the preference Q-value $\mathbf{Q}_{j}$ of an agent $j$ that joins the game is filled in a red box, that remains in the game is filled in a blue box, and that leaves the game filled in a grey box with the dashed outline. At timestep $1$, agent $1$, $2$ and $3$ just join the game, so with all preference Q-values as zeros. At timestep $2$, agent $3$ leaves the game, and the expansion works as usual, since agent $3$ has influence to the team. At timestep $3$, agent $4$ joins the game, but it has not any influence to the team. For this reason, it is unnecessary to consider the expansion for agent $4$'s preference Q-value \textit{only at timestep $3$}, since it would be trivially zero. To satisfy the more generic representation of the joint Q-value with respect to preference Q-values (other than the linear decomposition described in our theory), we rule out the transition samples of $\mathcal{N}_{t} \subset \mathcal{N}_{t+1}$. At timestep $4$, the expansion considers all existing agents as usual.}
        \label{fig:bellman_optimality_equation_expansion}
        \end{figure}
        \begin{theorem}
        \label{thm:open_team_ad_hoc_bellman_optimality}
            Under Assumption \ref{assm:agent_leaves_env} and an arbitrary learner's deterministic stationary policy $\pi^{i}$, the Bellman equation for the OSB-CAG with DVSC as a solution concept is expressed as follows: $Q^{\pi^{i}}(s_{t}, a_{t}) = R(s_{t}, a_{t}) + \gamma \mathbb{E}_{{\scriptscriptstyle \mathcal{N}_{t+1}}, s_{t+1} \sim P_{O}} \Big[ 
\mathbb{E}_{\substack{\theta_{t+1} \sim P_{E}, \ a_{t+1} \sim \pi_{t+1}}} \big[ Q^{\pi^{i}}(s_{t+1}, a_{t+1}) \big] \Big]$.
        \end{theorem}
        To solve Eq.~\eqref{eq:open_team_ad_hoc_bellman_optimality}, we further propose an operator with the same regularity condition, such that $\Gamma: Q \mapsto \Gamma Q$, specified as follows:
        \begin{equation}
        \label{eq:open_team_hedonic_bellman_operator}
            \begin{split}
                \Gamma Q^{\pi^{i}} \left( s_{t+1},a_{t+1}^{-i}, a^{i} \right) := R(s_{t}, a_{t}) + \gamma \mathbb{E}_{{\scriptscriptstyle \mathcal{N}_{t+1}}, s_{t+1} \sim P_{O}} \Big[ \\
                \max_{a^{i}} \mathbb{E}_{\substack{\theta_{t+1} \sim P_{E},\ a_{t+1}^{-i} \sim \pi_{t+1}^{-i}}} \big[ Q^{\pi^{i}}(s_{t+1}, a_{t+1}^{-i}, a^{i}) \big] \Big].
            \end{split}
        \end{equation}
        Eq.~\eqref{eq:open_team_hedonic_bellman_operator} is a standard form of Bellman operator. Therefore, recursively running Eq.~\eqref{eq:open_team_hedonic_bellman_operator} converges to the Bellman optimality equation in Eq.~\eqref{eq:open_team_ad_hoc_bellman_optimality}, following the well-known value iteration algorithm \citep[Ch. 4]{sutton2018reinforcement}. 
        \begin{remark}
        \label{rmk:ignorance_effects_of_joining_team}
            In implementation, the effect of $\mathcal{N}_{t} \subset \mathcal{N}_{t+1}$ can be omitted, due to its low proportions during the process. Therefore, solving the GPL optimization problem of fitted Q-learning \citep{ernst2005tree} that omits the effect of $\mathcal{N}_{t} \subset \mathcal{N}_{t+1}$ is a reasonable approximation of Bellman operator in Eq.~\eqref{eq:open_team_hedonic_bellman_operator}, which reduces the computational cost of filtering out the transition samples of $\mathcal{N}_{t} \subset \mathcal{N}_{t+1}$ in practice. The GPL optimization problem is shown as follows:
            \begin{equation}
            \label{eq:fitted_q-learning}
                \begin{split}
                    &\min_{\beta} L(\beta) = \mathbb{E} \Big[ \frac{1}{2} \Big( R(s_{t}, a_{t}) + \gamma \max_{a^{i}} \mathbb{E}_{\substack{\theta_{t+1} \sim P_{E}, \\ a_{t+1}^{-i} \sim \pi_{t+1}^{-i}}} \big[ \\
                    &\hat{Q}^{\pi^{i}}(s_{t+1},a_{t+1}^{-i}, a^{i}; \beta^{-}) \big] - \hat{Q}^{\pi^{i}}(s_{t},a_{t}; \beta) \Big)^{2} \Big],
                \end{split}
            \end{equation}
            where $\hat{Q}^{\pi^{i}}(\cdot \ ; \ \beta^{-})$ is the approximate target optimal joint Q-value parameterised by $\beta^{-}$ and $\hat{Q}^{\pi^{i}}(\cdot \ ; \ \beta)$ is the approximate optimal joint Q-value parameterised by $\beta$.
        \end{remark}
    
    \subsection{Practical Implementation}
    \label{subsec:practical_implementation}
        Based on our theory, we introduce a novel algorithm, CIAO, representing the algorithm for \textbf{C}ooperative game theory \textbf{I}nspired \textbf{A}d hoc teamwork in \textbf{O}pen teams. We implement CIAO in dynamic affinity graphs as a star graph (refer to Remark~\ref{rmk:ad_hoc_affinity_graph} for more insights into this topology) and a complete graph, denoted as \textit{CIAO-S} and \textit{CIAO-C}, respectively, where ``S'' signifies \textbf{S}tar graph and ``C'' signifies \textbf{C}omplete graph. In addition to the joint Q-value representation model (derived from Theorem~\ref{thm:joint_q_representation}) and the training losses for estimating the unknown type inference model and the unknown agent model (as detailed in Section~\ref{subsec:graph-based_policy_optimization}), we introduce novel Q losses tailored for variant dynamic affinity graphs based on our theory. These losses incorporate regularization terms with multipliers $\lambda > 0$.

        \textbf{CIAO-S.} If the dynamic affinity graph is a star graph, the training loss with the regularizer is as follows:
        \begin{equation*}
            \begin{split}
                L_{s}(&\beta) = L(\beta) \\
                &+ \lambda \mathbb{E}_{s_{t}, a_{t}}\Big[ \frac{1}{2} \big( \sum_{j \in -i} \hat{Q}_{j}^{\pi^{i}}(a_{t}^{j} \vert s_{t}) - \hat{Q}_{i}^{\pi^{i}}(a_{t}^{i} \vert s_{t}; \beta) \big)^{2} \Big].
            \end{split}
        \end{equation*}
        \vspace{-5mm}

        \textbf{CIAO-C.} If the dynamic affinity graph is a complete graph, the training loss with the regularizer is as follows:
        \begin{equation*}
            \begin{split}
                L_{c}(&\beta) = L(\beta) \\
                &+ \lambda \mathbb{E}_{s_{t}, a_{t}}\Big[ \sum_{j \in -i} \frac{1}{2} \big( \hat{Q}_{i}^{\pi^{i}}(a_{t}^{i} \vert s_{t}) - \hat{Q}_{j}^{\pi^{i}}(a_{t}^{j} \vert s_{t}; \beta) \big)^{2} \Big].
            \end{split}
        \end{equation*}
        \vspace{-5mm}

        Note that it is also requisite to enforce that $\hat{Q}_{jk}^{\pi^{i}}(a_{t}^{j}, a_{t}^{k} \vert s_{t}) = \hat{Q}_{kj}^{\pi^{i}}(a_{t}^{k}, a_{t}^{j} \vert s_{t}) \geq 0$ by Remark \ref{rmk:implement_symmetry_dynamic_affinity_graph}. Following our theoretical model, the learner's reward $R(s_{t}, a_{t})$ ought to be non-negative, while the designated reward of an environment could be negative. However, this can be adjusted by adding the maximum difference between these two rewards among states and joint actions denoted by $\Delta R(s_{t}, a_{t})$ without changing the original goal. In practice, Eq.~\eqref{eq:fitted_q-learning} is solved by DQN \citep{mnih2013playing}. The learner's actions are decided by Eq.~\eqref{eq:objective_ad_hoc}, employing the estimated teammates' agent models $\hat{\pi}^{-i}$ (see Section~\ref{subsec:graph-based_policy_optimization}) to marginalize $a_{t}^{-i}$ of $\hat{Q}^{\pi^{i}}(s_{t},a_{t}; \beta)$, as implemented in GPL. The further implementation details are left to Appendix~\ref{sec:further_details_of_implementation}.
        
\section{Experiments}
\label{sec:experiments}
    We assess the effectiveness of the proposed algorithms CIAO-S and CIAO-C in two established environments, LBF and Wolfpack, featuring open team settings \citep{rahman2021towards}. In these settings, teammates are randomly selected to enter the environment and remain for a certain number of time steps. During experiments, the learner is trained in an environment with a maximum of 3 agents at each timestep. Subsequently, testing is conducted in environments with a maximum of 5 and 9 agents at each timestep, showcasing the model's ability to handle both unseen compositions and varied team sizes. All experiments are conducted with five random seeds, and the results are presented as the mean performance with a 95\% confidence interval. Our experimental design aims to answer the following questions: (1) Does the joint Q-value representation outlined in our theory effectively facilitate collaboration between the learner and temporary teammates? (2) Is it necessary to generalize the preference reward function from zero, as in CAG, to a non-negative range in our theory (see Appendix \ref{sec:generalization_of_preference_values})? (3) Is the claim in Remark \ref{rmk:ignorance_effects_of_joining_team} valid in practice? (4) Is CIAO able to improve the generalization of agent-type sets?
    
    \textbf{Baselines and Ablation Variants.} The state-of-the-art baseline we use in this experiment is GPL-Q (shortened as GPL) \citep{rahman2021towards}. The ablation variants of the proposed CIAO are as follows: \textbf{CIAO-X-FI}, \textbf{CIAO-X-ZI} and \textbf{CIAO-X-NI} are variants that remove enforcement of individual utility, enforce individual utility as zero and enforce individual utility as negative values, respectively. \textbf{CIAO-X-NP} is a variant that enforces negative pairwise utility. ``X'' above indicates either ``S'' or ``C''. Further details on experimental settings can be found in Appendix~\ref{sec:experimental_settings}. 
    
    \subsection{Main Results}
    \label{sec:main_results}
        \begin{figure}[ht!]
        \centering
            \begin{subfigure}[b]{0.238\textwidth}
            \centering
                \includegraphics[width=\textwidth]{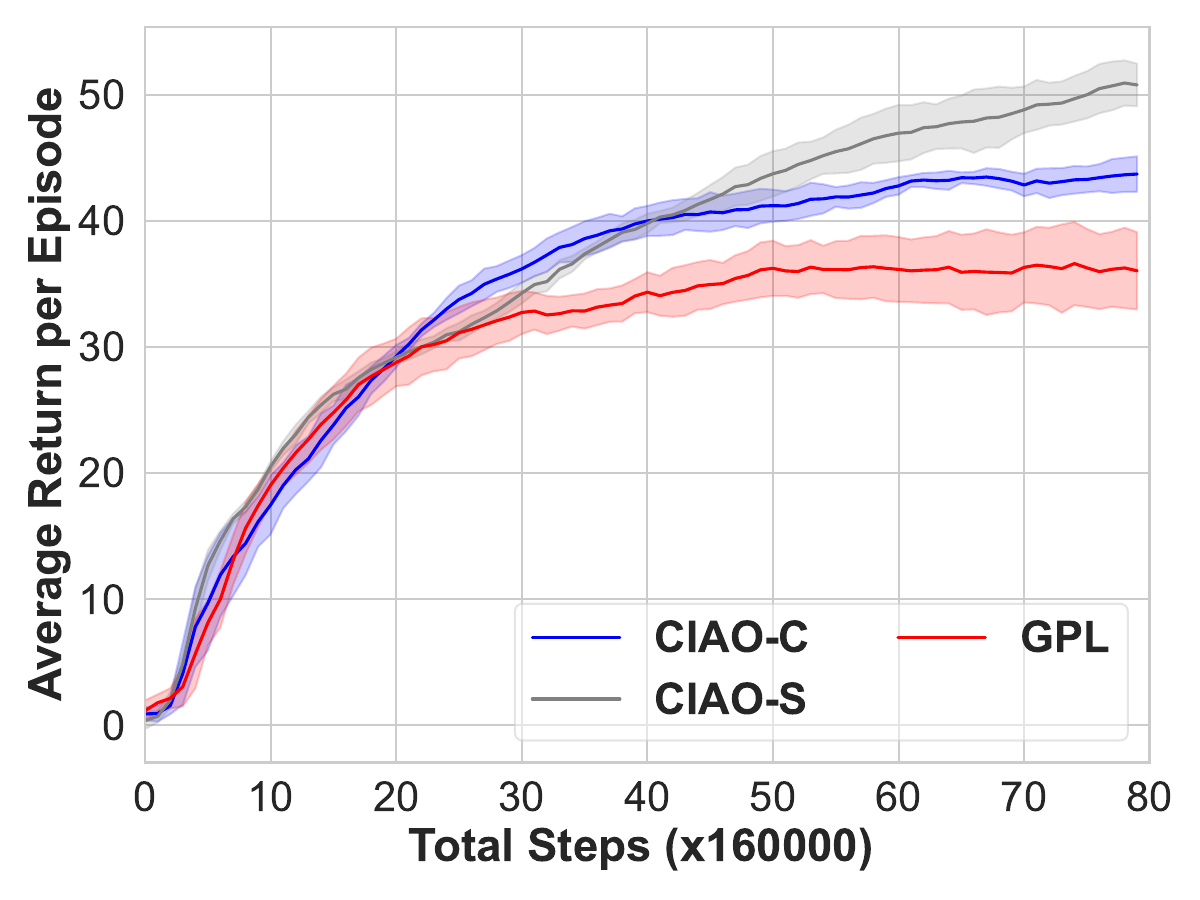}
                \caption{Wolfpack: max. of 5 agents.}
                \label{fig:wolfpack-test-5-normal}
            \end{subfigure}
            \hfill
            \begin{subfigure}[b]{0.238\textwidth}
            \centering
                \includegraphics[width=\textwidth]{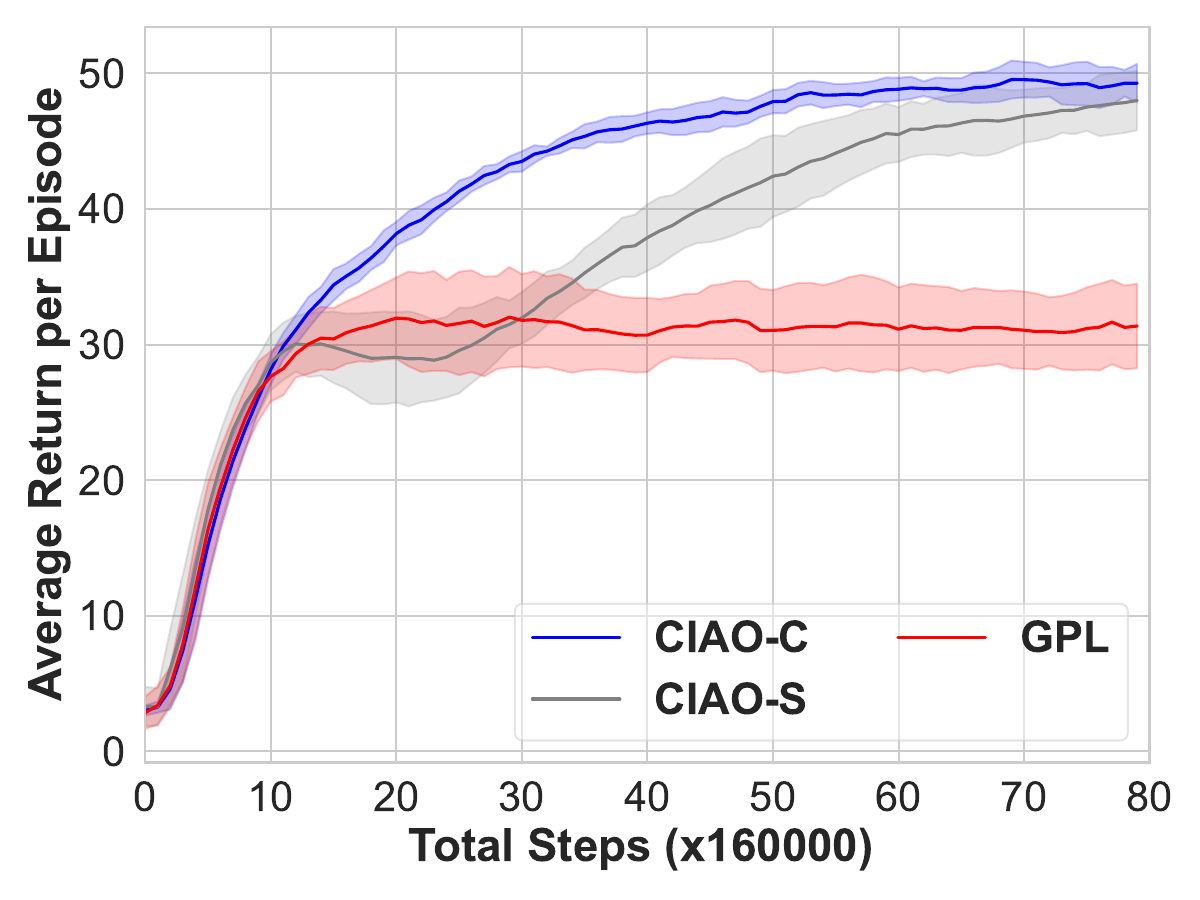}
                \caption{Wolfpack: max. of 9 agents.}
                \label{fig:wolfpack-test-9-normal}
            \end{subfigure}
            \hfill
            \begin{subfigure}[b]{0.238\textwidth}
            \centering
                \includegraphics[width=\textwidth]{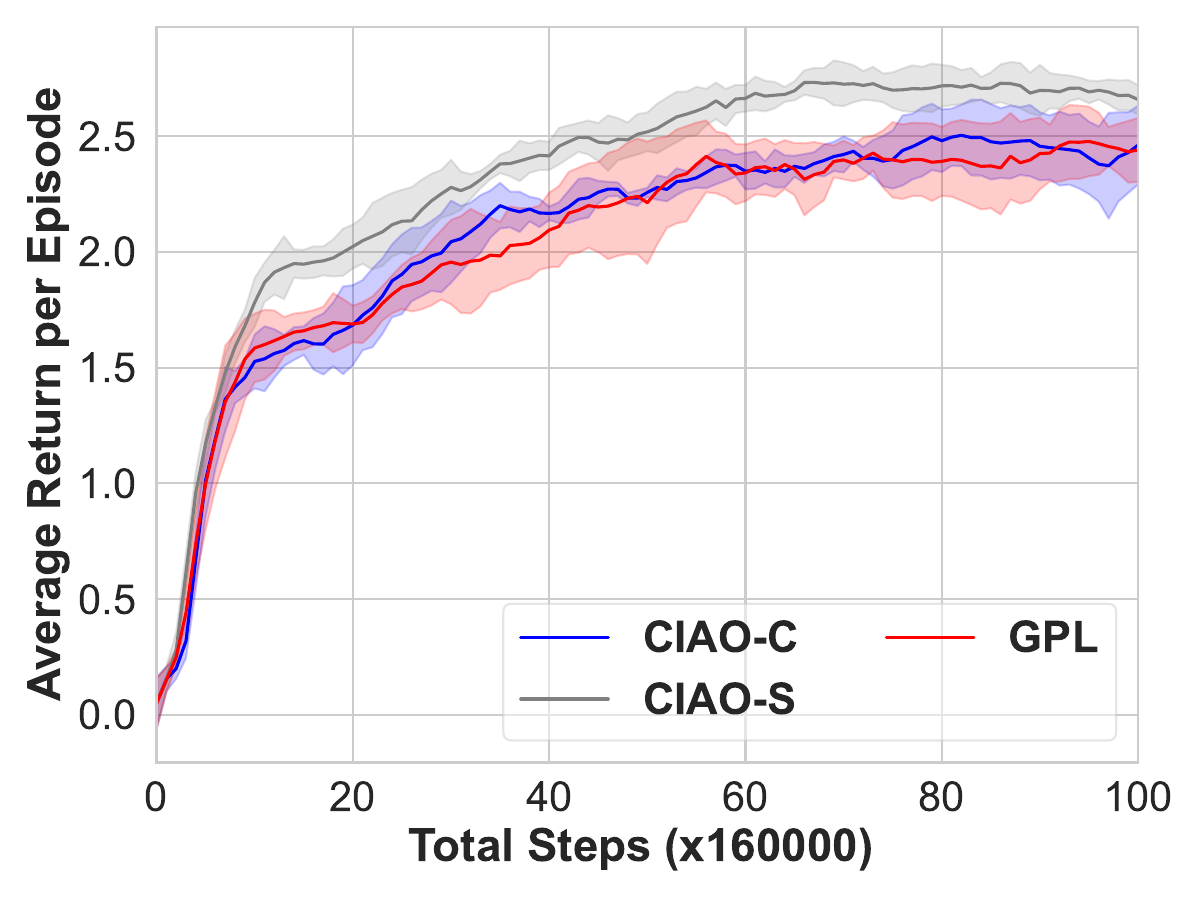}
                \caption{LBF: max. of 5 agents.}
                \label{fig:lbf-test-5-normal}
            \end{subfigure}
            \hfill
            \begin{subfigure}[b]{0.238\textwidth}
            \centering
                \includegraphics[width=\textwidth]{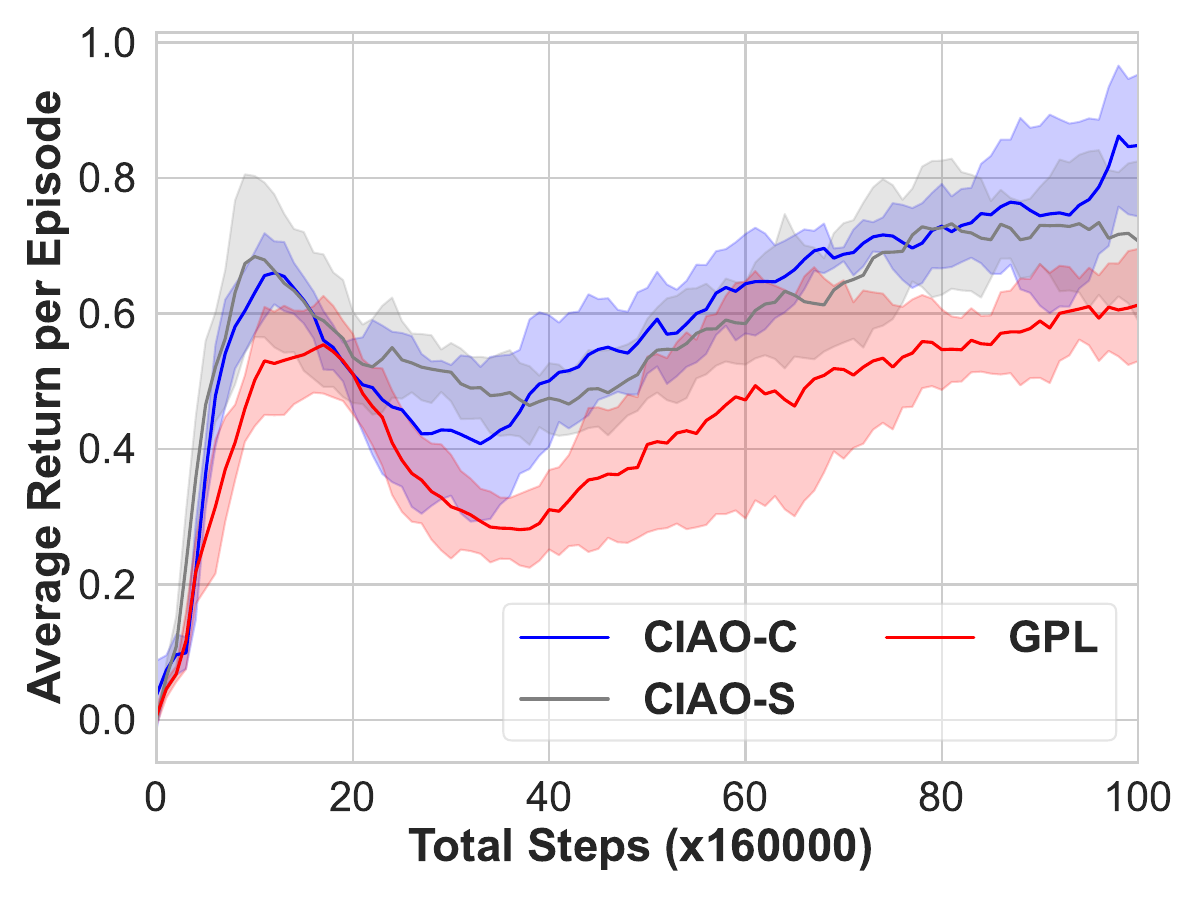}
                \caption{LBF: max. of 9 agents.}
                \label{fig:lbf-test-9-normal}
            \end{subfigure}
            \caption{Comparison between CIAO and GPL in Wolfpack and LBF with a maximum of 5 and 9 agents.}
            \label{fig:main_results}
        \end{figure}
        We initially address Question 1 through experiments conducted on the original versions of Wolfpack and LBF, as depicted in Fig.~\ref{fig:main_results}. It is evident that CIAO-C outperforms GPL in the majority of scenarios with varying maximum numbers of agents. This not only verifies the correctness and effectiveness of our theory, irrespective of dynamic affinity graph structures but also demonstrates its capability in facilitating collaboration between the learner and temporary teammates in the open ad hoc teamwork problem. Upon comparing CIAO-C and CIAO-S, it becomes apparent that the star graph may be more effective in scenarios with fewer agents, whereas the complete graph exhibits greater effectiveness in scenarios with more agents. This observation aligns with the intuition that the direct influence from the learner to each teammate may not suffice as the number of agents increases. Instead, indirect influence, where a teammate is influenced by the learner to subsequently influence another teammate, becomes crucial. 

    \subsection{Ablation Study}
    \label{subsec:ablation_study}
        \begin{figure}[ht!]
        \centering
            \begin{subfigure}[b]{0.238\textwidth}
            \centering
                \includegraphics[width=\textwidth]{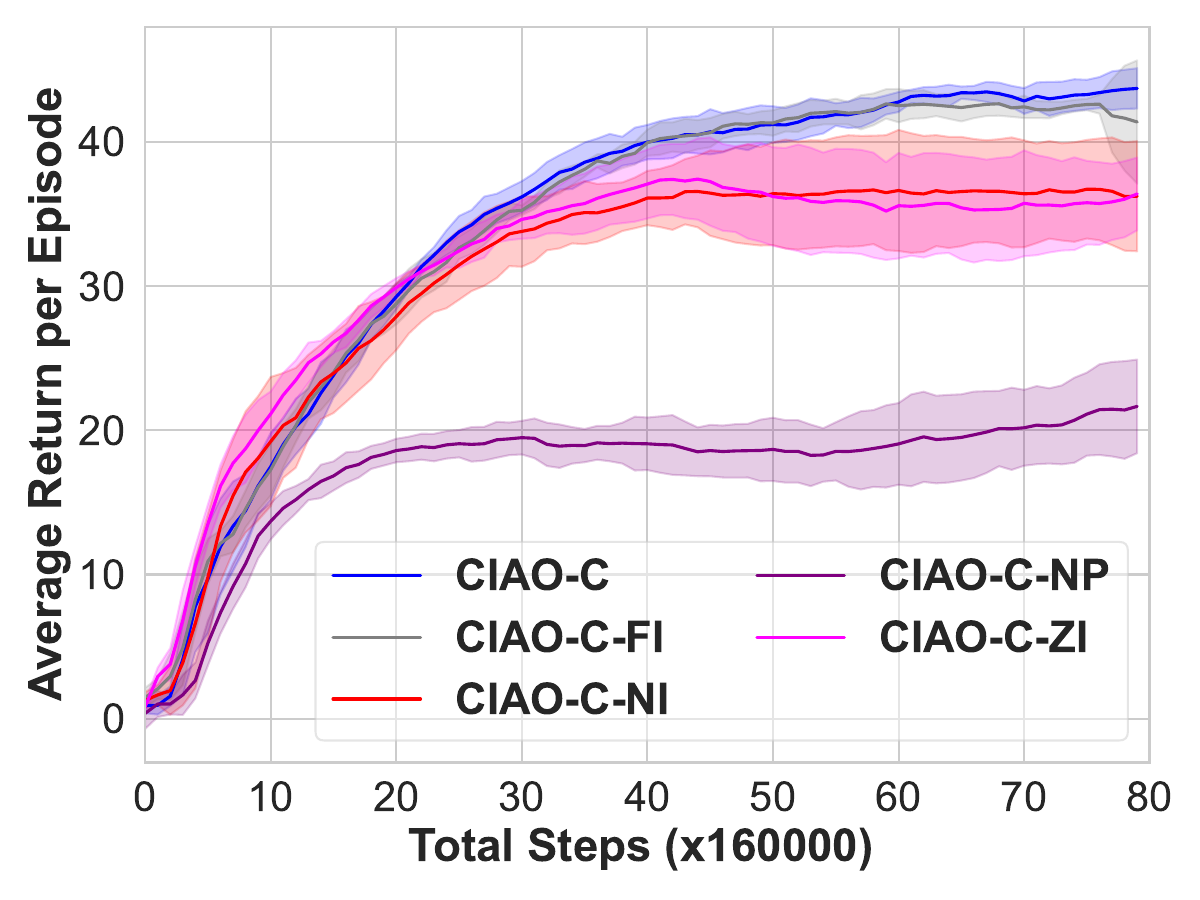}
                \caption{Wolfpack: max. of 5 agents.}
                \label{fig:wolfpack-test-5-ablation-complete}
            \end{subfigure}
            \hfill
            \begin{subfigure}[b]{0.238\textwidth}
            \centering
                \includegraphics[width=\textwidth]{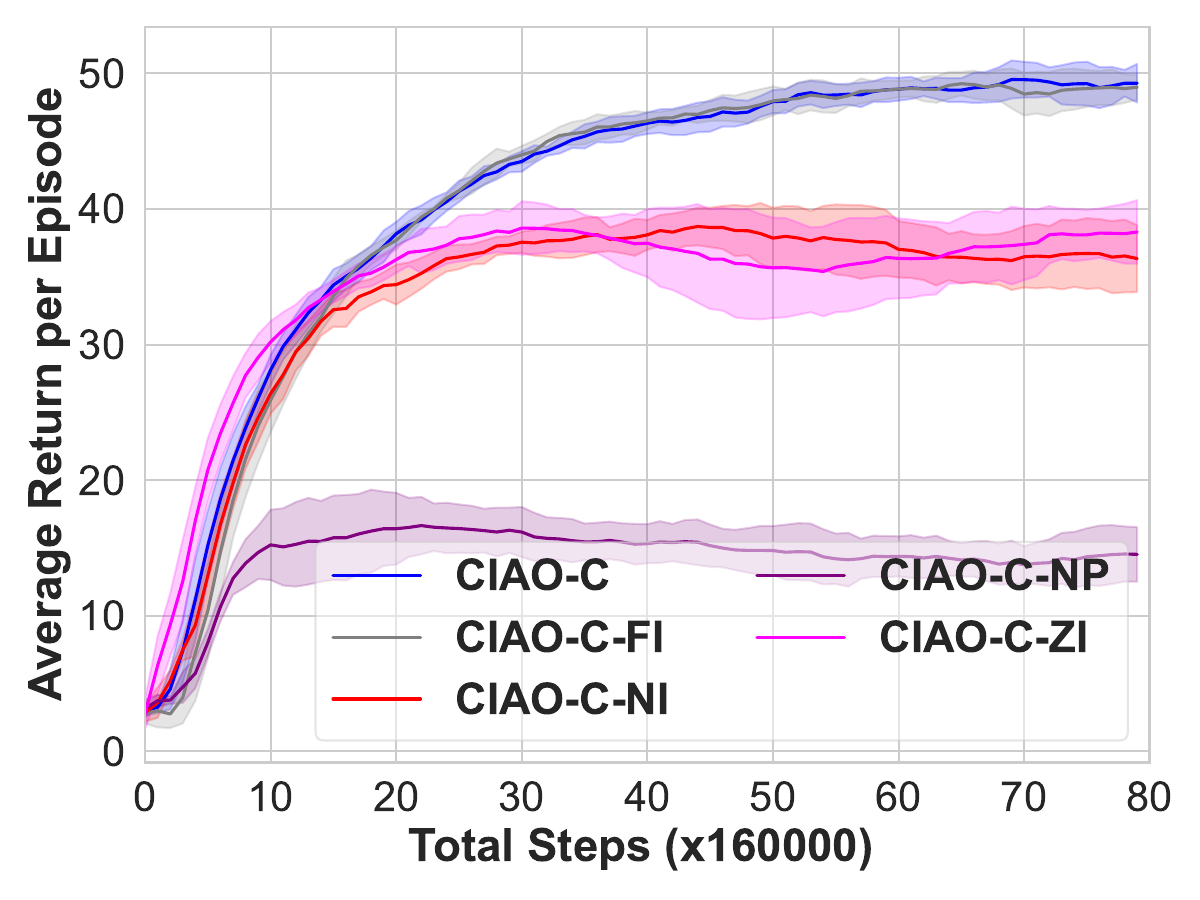}
                \caption{Wolfpack: max. of 9 agents.}
                \label{fig:wolfpack-test-9-ablation-complete}
            \end{subfigure}
            \hfill
            \begin{subfigure}[b]{0.238\textwidth}
            \centering
                \includegraphics[width=\textwidth]{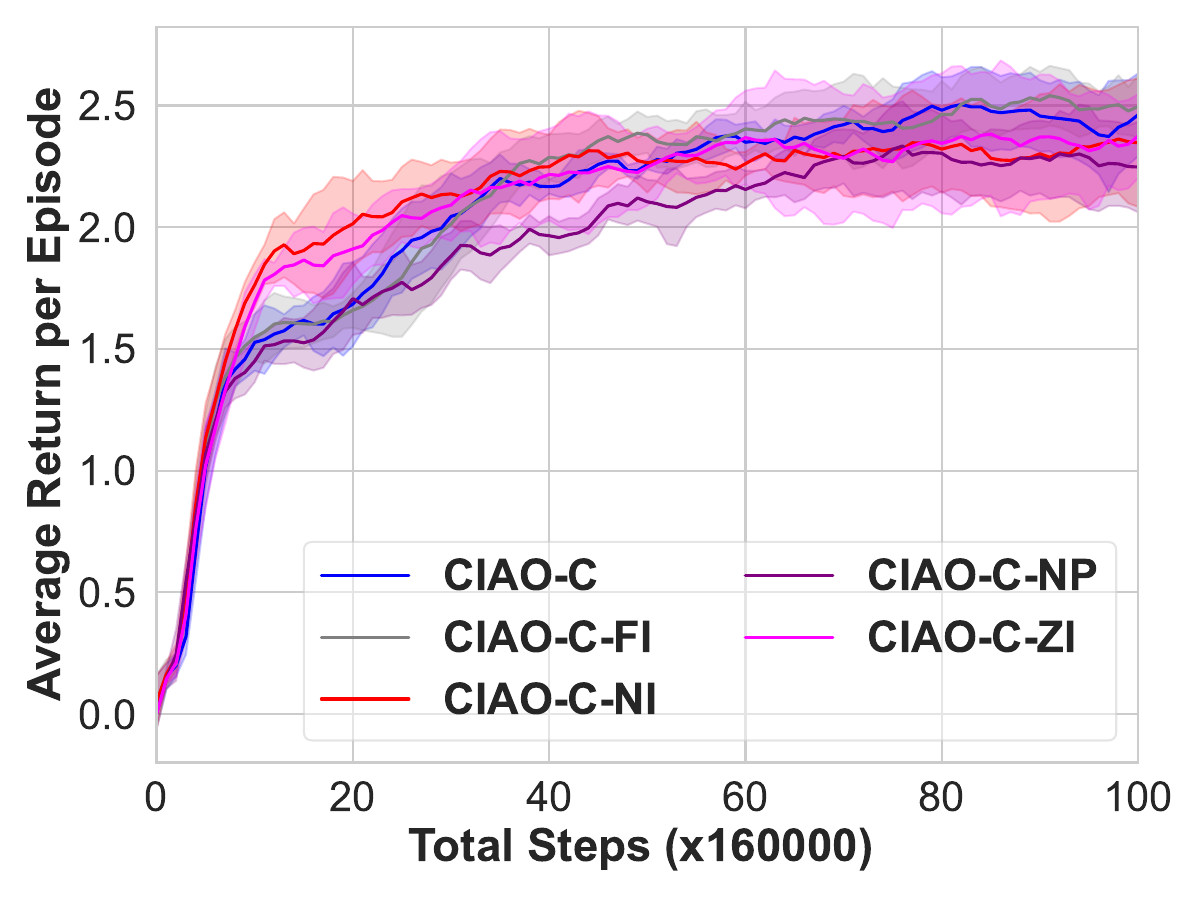}
                \caption{LBF: max. of 5 agents.}
                \label{fig:lbf-test-5-ablation-complete}
            \end{subfigure}
            \hfill
            \begin{subfigure}[b]{0.238\textwidth}
            \centering
                \includegraphics[width=\textwidth]{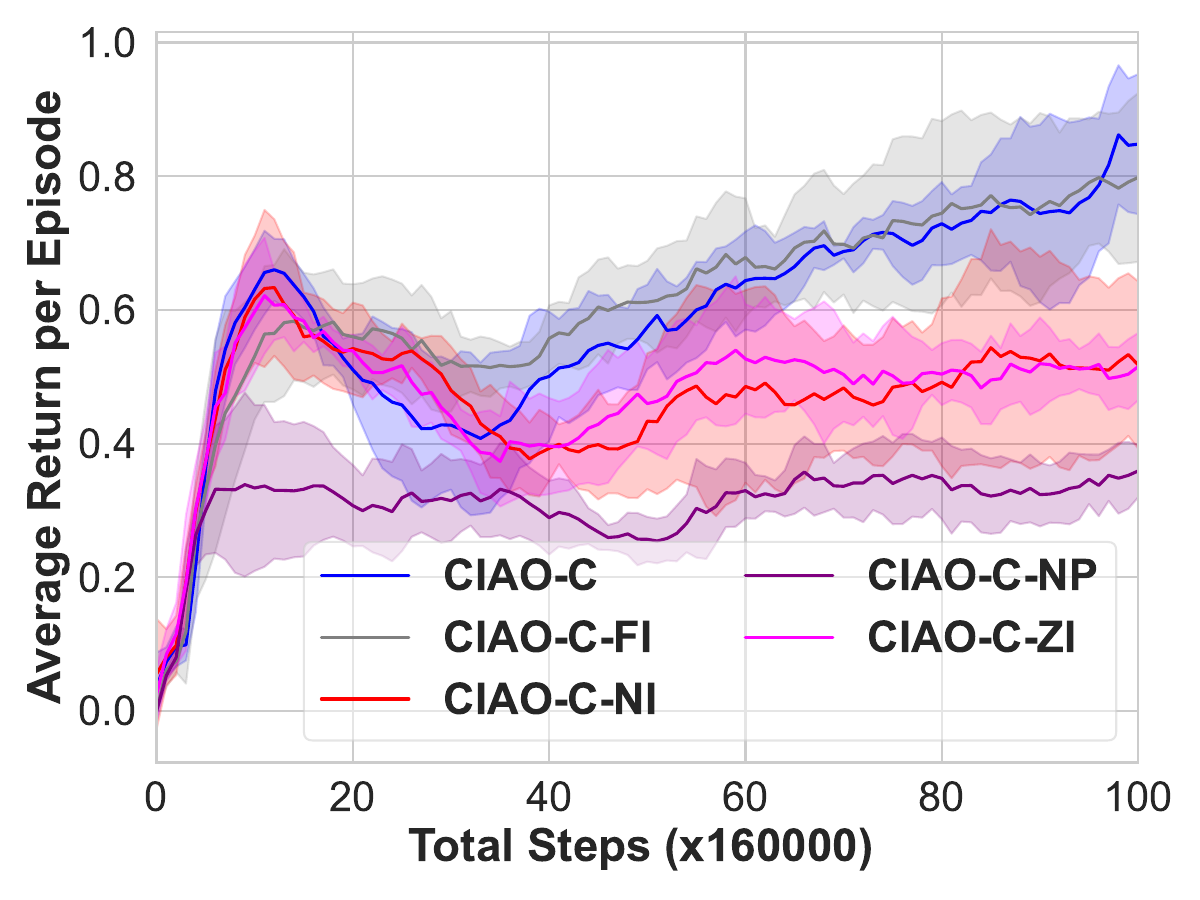}
                \caption{LBF: max. of 9 agents.}
                \label{fig:lbf-test-9-ablation-complete}
            \end{subfigure}
            \caption{Comparison between CIAO-C and its ablations in Wolfpack and LBF with a maximum of 5 and 9 agents.}
        \label{fig:ablation_study-complete}
        \end{figure}
        We present experimental results comparing CIAO-S and its ablations, as well as CIAO-C and its ablations. As illustrated in Figs.~\ref{fig:ablation_study-complete} and \ref{fig:ablation_study-star}, both CIAO-C-NP and CIAO-S-NP exhibit notably inferior performance compared to CIAO-C or CIAO-S. This observation demonstrates the validity of DVSC and confirms the accuracy of the joint Q-value representation based on our theory. This outcome provides an additional perspective in addressing Question 1. 
        \begin{figure}[ht!]
        \centering
            \begin{subfigure}[b]{0.238\textwidth}
            \centering
                \includegraphics[width=\textwidth]{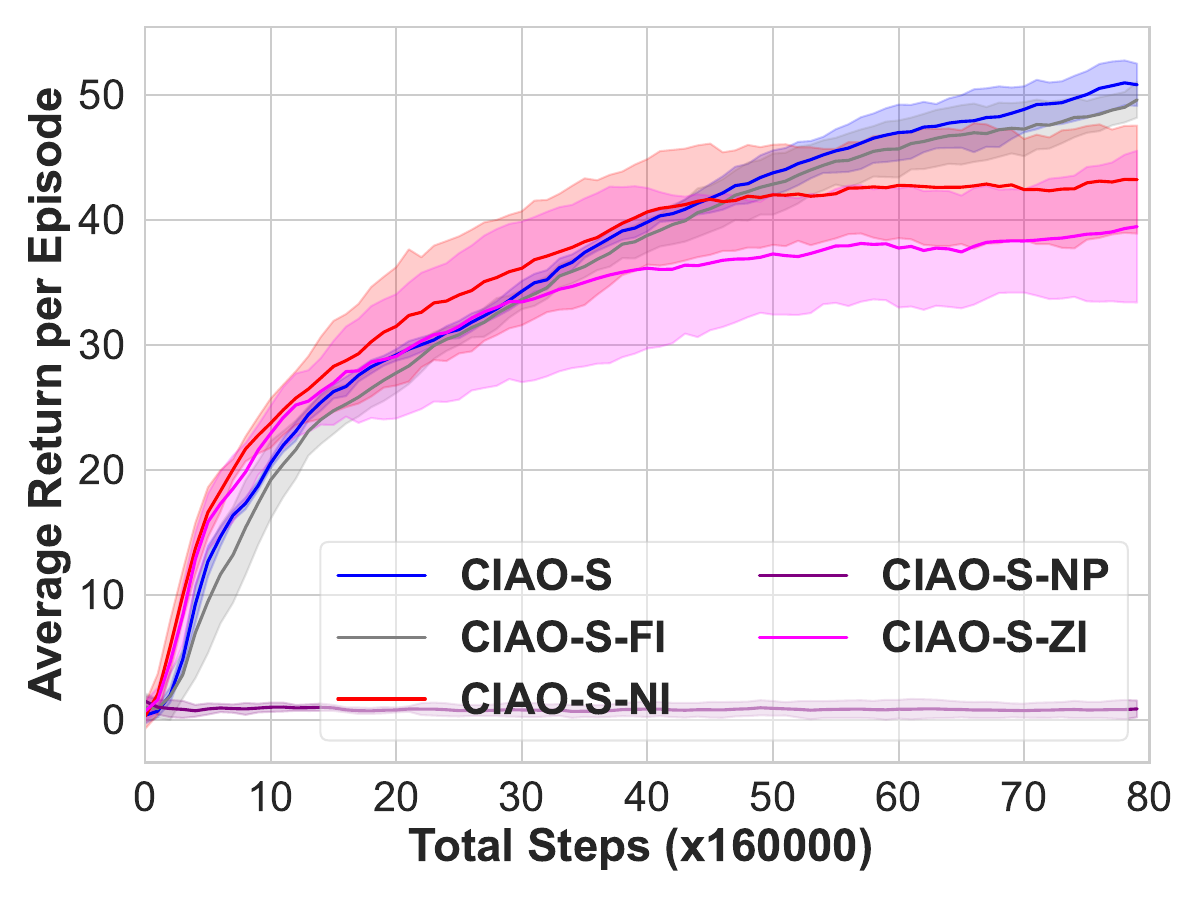}
                \caption{Wolfpack: max. of 5 agents.}
                \label{fig:wolfpack-test-5-ablation-star}
            \end{subfigure}
            \hfill
            \begin{subfigure}[b]{0.238\textwidth}
            \centering
                \includegraphics[width=\textwidth]{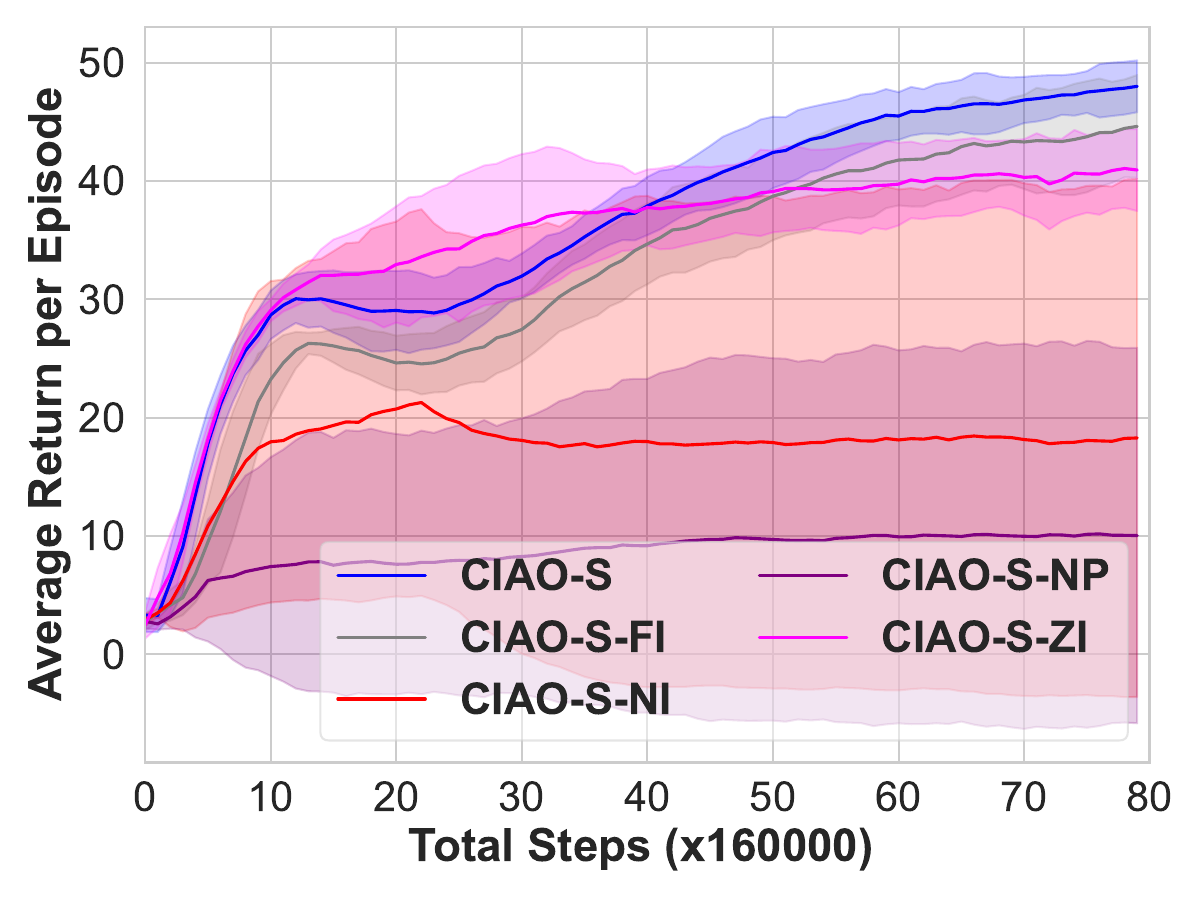}
                \caption{Wolfpack: max. of 9 agents.}
                \label{fig:wolfpack-test-9-ablation-star}
            \end{subfigure}
            \hfill
            \begin{subfigure}[b]{0.238\textwidth}
            \centering
                \includegraphics[width=\textwidth]{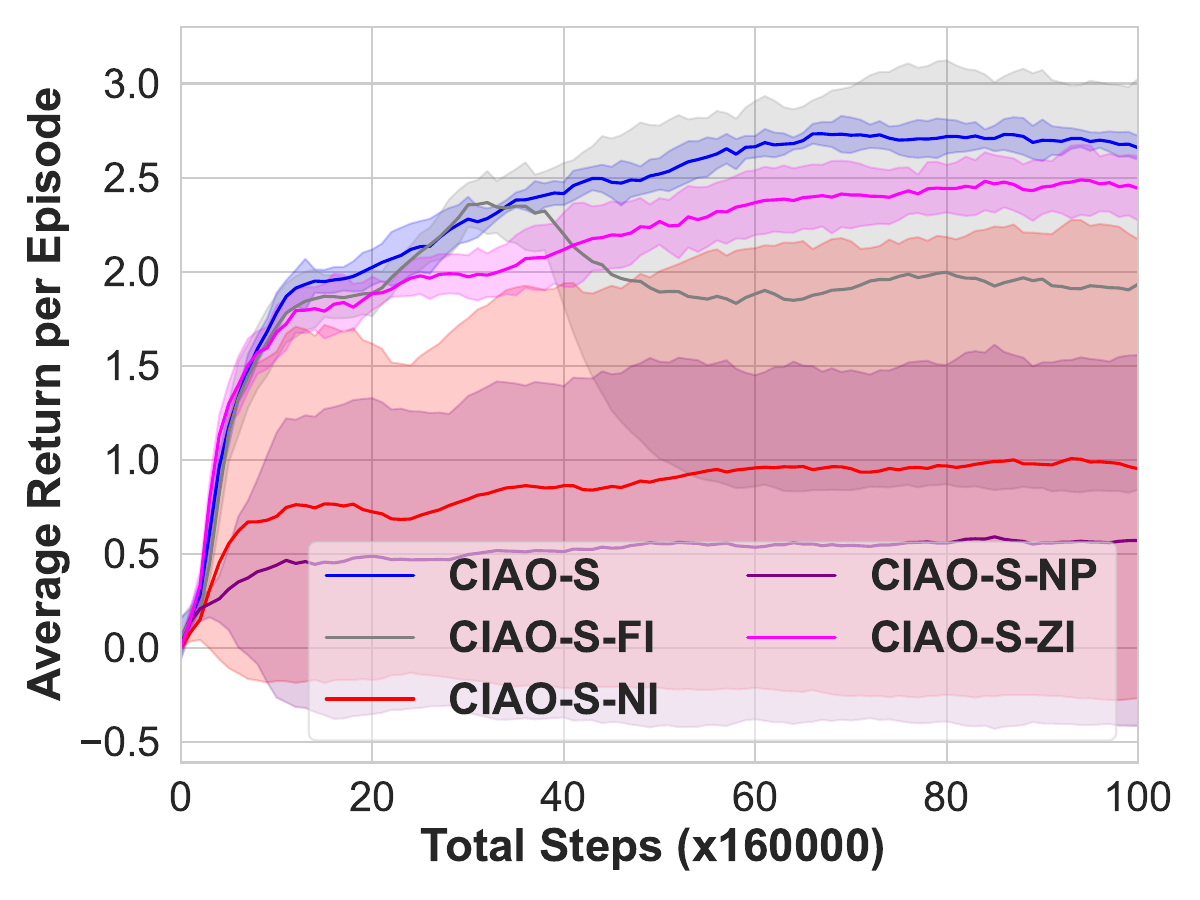}
                \caption{LBF: max. of 5 agents.}
                \label{fig:lbf-test-5-ablation-star}
            \end{subfigure}
            \hfill
            \begin{subfigure}[b]{0.238\textwidth}
            \centering
                \includegraphics[width=\textwidth]{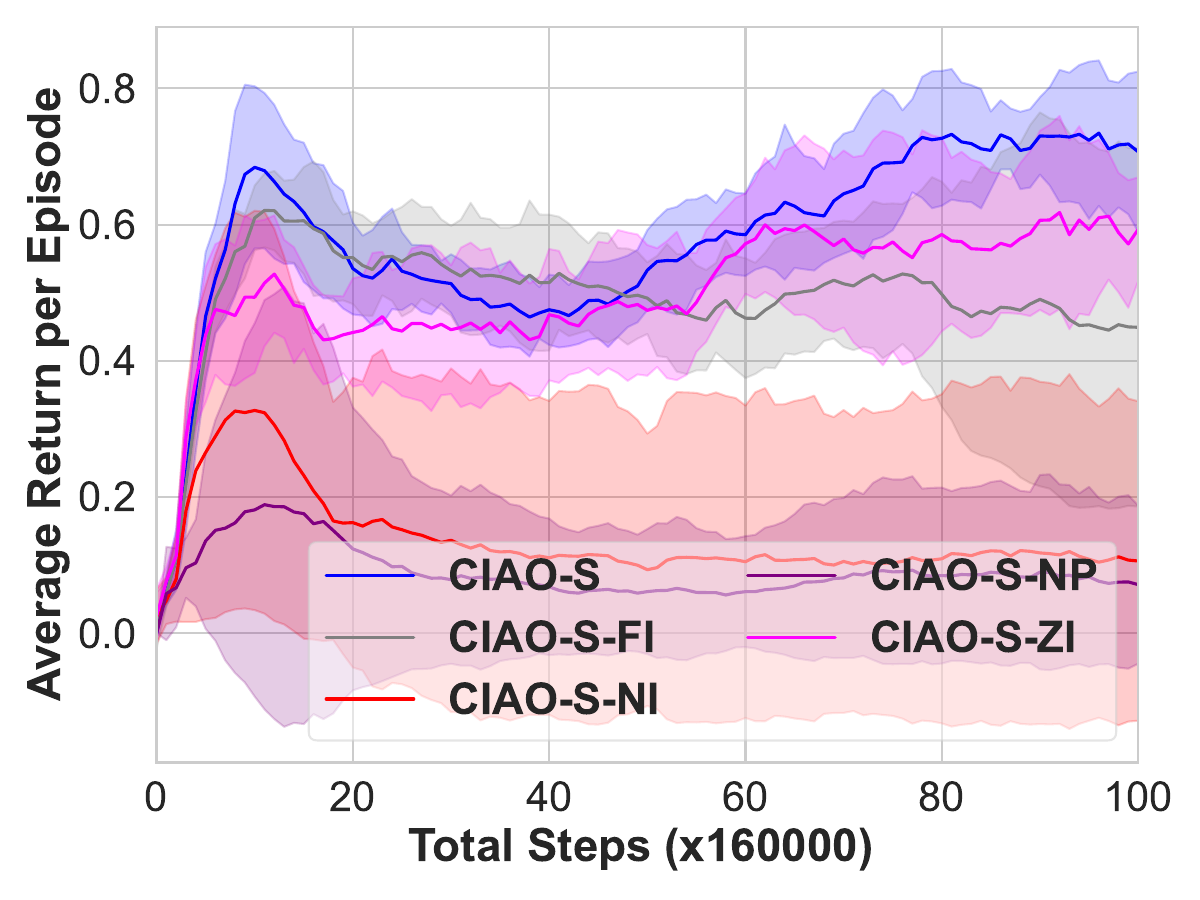}
                \caption{LBF: max. of 9 agents.}
                \label{fig:lbf-test-9-ablation-star}
            \end{subfigure}
            \caption{Comparison between CIAO-S and its ablations in Wolfpack and LBF with a maximum of 5 and 9 agents.}
        \label{fig:ablation_study-star}
        \end{figure}
        Adhering to the tradition of CAG, convention mandates setting individual utility to zero. However, in our theory, we extend its range to include non-zero values, enhancing its adaptability across diverse scenarios. This adaptability is demonstrated in the comparison between CIAO-C or CIAO-S and CIAO-C-ZI or CIAO-S-ZI in Figs.~\ref{fig:ablation_study-complete} and \ref{fig:ablation_study-star}. Although our theory does not inherently provide specific insights into the range of individual utility, we propose a hypothesis aligned with other definitions in CAG, asserting that individual utility is non-negative. This hypothesis ensures self-consistency in our generalization, as detailed in Definition~\ref{def:singleton_coalition_preference_value} in Appendix~\ref{sec:generalization_of_preference_values}. The superior performances of CIAO-C or CIAO-S over their ablations affirm the acceptability of our hypothesis.

    \subsection{Validation for Remark~\ref{rmk:ignorance_effects_of_joining_team}}
    \label{subsec:validity_of_remark2}
        \begin{figure}[ht!]
        \centering
            \begin{subfigure}[b]{0.238\textwidth}
            \centering
                \includegraphics[width=\textwidth]{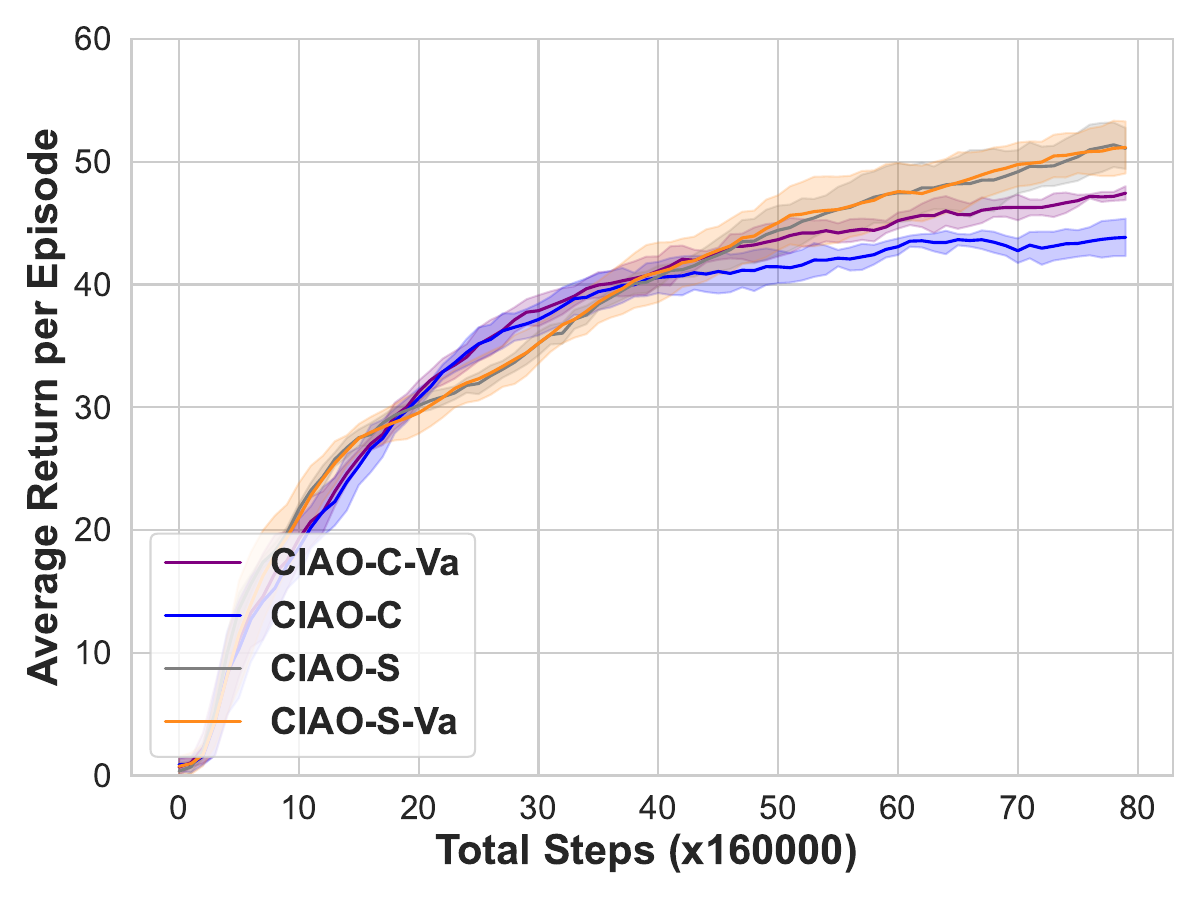}
                \caption{Wolfpack: max. of 5 agents.}
            \label{fig:lbf-test-5-va}
            \end{subfigure}
            \hfill
            \begin{subfigure}[b]{0.238\textwidth}
            \centering
                \includegraphics[width=\textwidth]{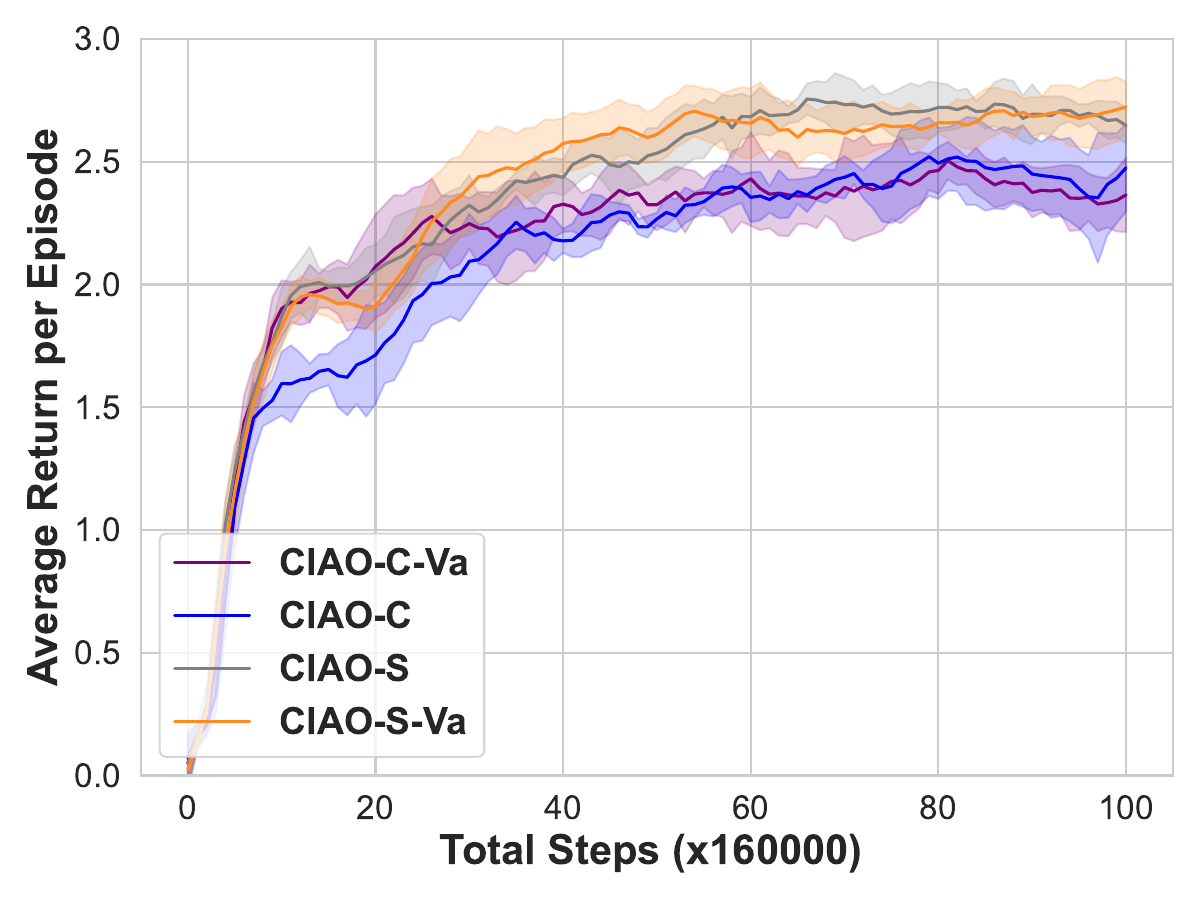}
                \caption{LBF: max. of 5 agents.}
            \label{fig:lbf-test-9-va}
            \end{subfigure}
            \caption{Comparison of training losses for CIAO between the implementations with omitting the effect of $\mathcal{N}_{t} \subset \mathcal{N}_{t+1}$ and those without (denoted as ``-Va'').}
        \label{fig:validity_algorithm}
        \end{figure}
        We now validate our claim in Remark \ref{rmk:ignorance_effects_of_joining_team} that minimizing the GPL training loss (omitting the effect of $\mathcal{N}_{t} \subset \mathcal{N}_{t+1}$) is an approximation of Eq.~\eqref{eq:open_team_hedonic_bellman_operator}. Based on the GPL training loss, we implement its variant that filters out the transition samples of $\mathcal{N}_{t} \subset \mathcal{N}_{t+1}$, following the suggestion from Remark~\ref{rmk:ignorance_effects_of_joining_team}, referred to as CIAO-C-Va and CIAO-S-Va. As shown in Fig.~\ref{fig:validity_algorithm}, in both LBF and Wolfpack with the maximum of 5 agents, CIAO-C and CIAO-S trained with the GPL training loss achieve the approximate performances to those with the variant training loss considering the effect of $\mathcal{N}_{t} \subset \mathcal{N}_{t+1}$.

    \subsection{Generalization of Agent-Type Sets}
    \label{subsec:further_experiments_on_variant_agent-type_sets}
        We now evaluate the generalizability of CIAO to agent-type sets through two scenarios: (1) the agent-type set for training has intersection of one agent-type with that for testing; (2) the agent-type set for training is mutually exclusive to that for testing. As seen from Fig.~\ref{fig:variant_agent-types}, the dynamic affinity graph as the star graph is more generalizable than the complete graph. One hypothesis for this phenomenon is that although the complete graph may be able to capture broader relationships among agents, it could be unnecessary for open ad hoc teamwork (see Remark~\ref{rmk:ad_hoc_affinity_graph}). The underlying principles behind this result deserve to be investigated in the future research. Although the generalizability of GNNs (which is also implemented in GPL) has featured prominently in generalization of agent-type sets~\citep{rahman2021towards}, the overall superior performance of CIAO to GPL still empirically shows CIAO's effectiveness to this problem.
            \begin{figure}[ht!]
            \centering
                \begin{subfigure}[b]{0.238\textwidth}
                \centering
                    \includegraphics[width=\textwidth]{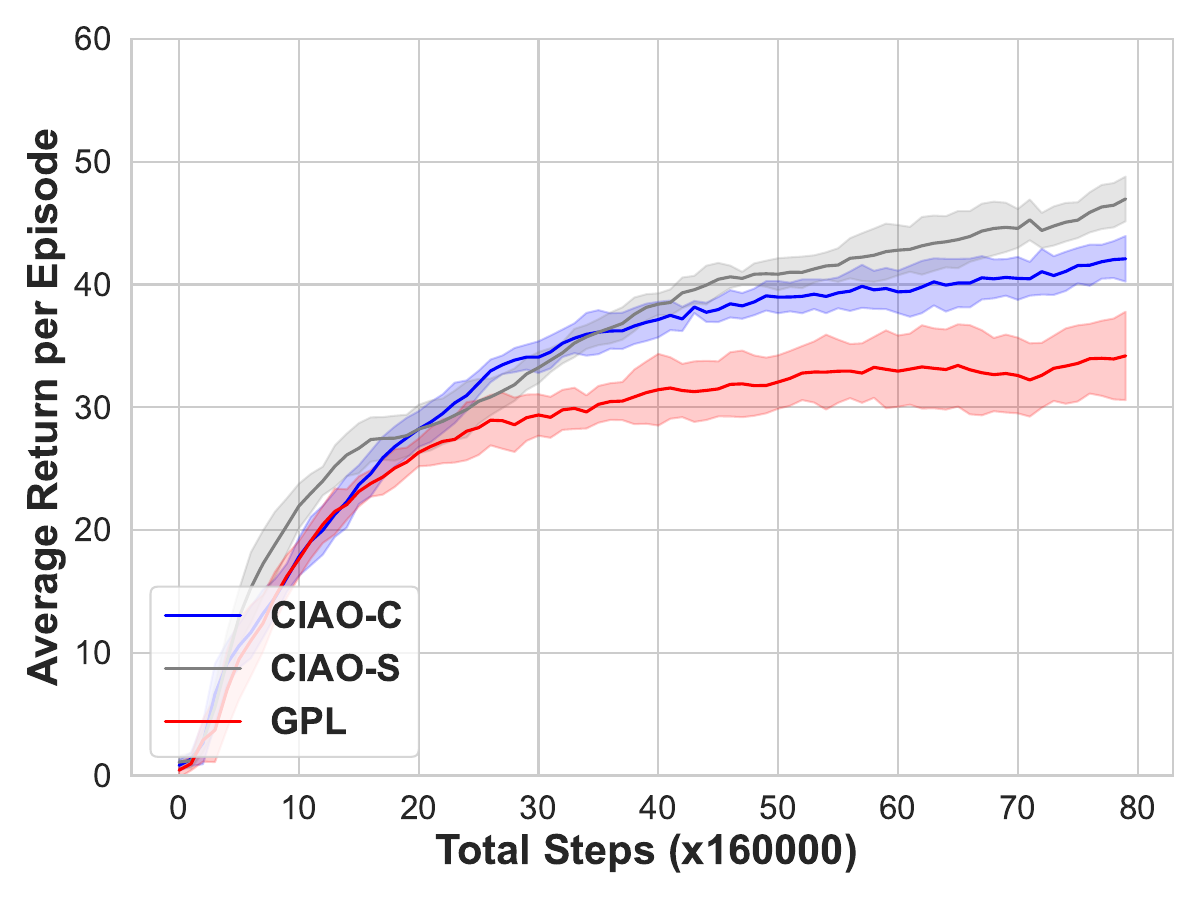}
                    \caption{Wolfpack: intersection.}
                    \label{fig:wolfpack-in}
                \end{subfigure}
                \begin{subfigure}[b]{0.238\textwidth}
                \centering
                    \includegraphics[width=\textwidth]{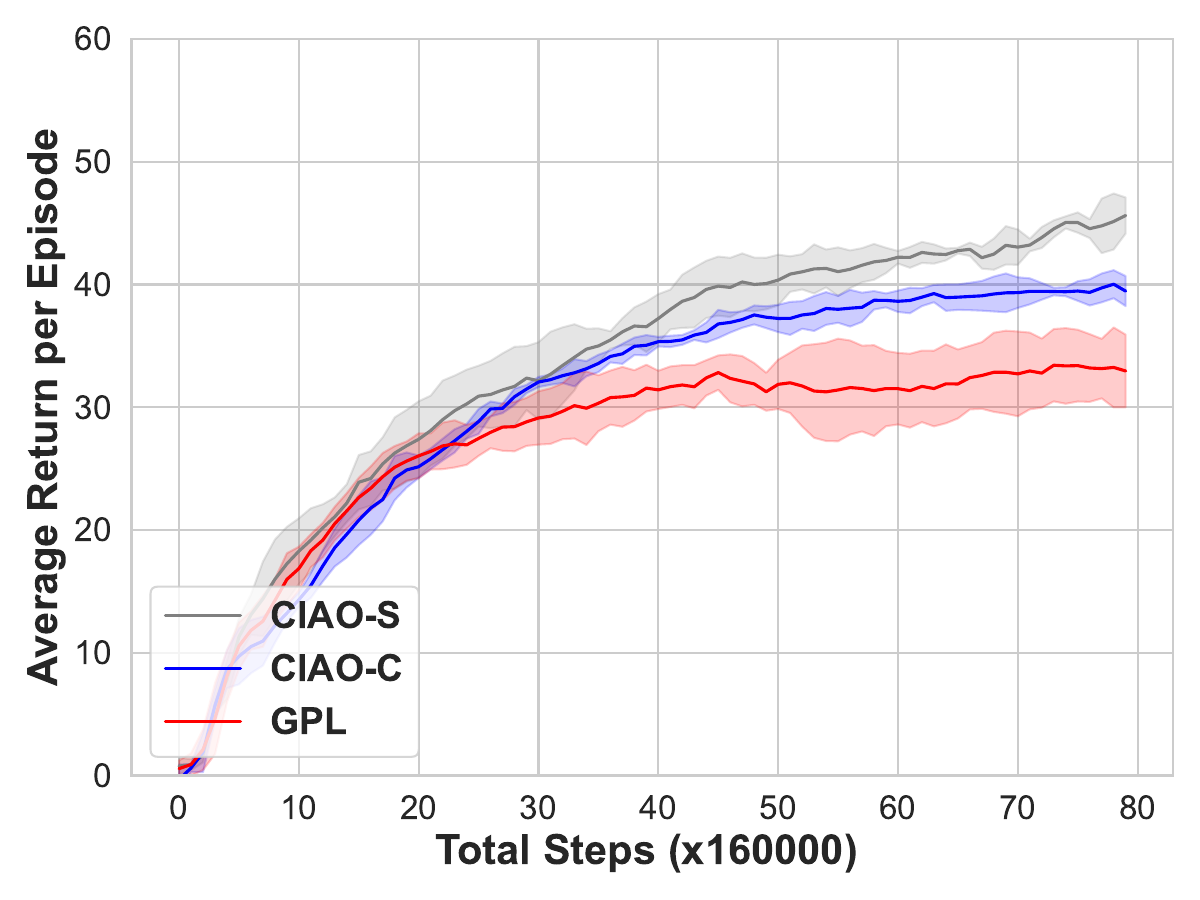}
                    \caption{Wolfpack: mutual exclusion.}
                    \label{fig:wolfpack-ex}
                \end{subfigure}
                \begin{subfigure}[b]{0.238\textwidth}
                \centering
                    \includegraphics[width=\textwidth]{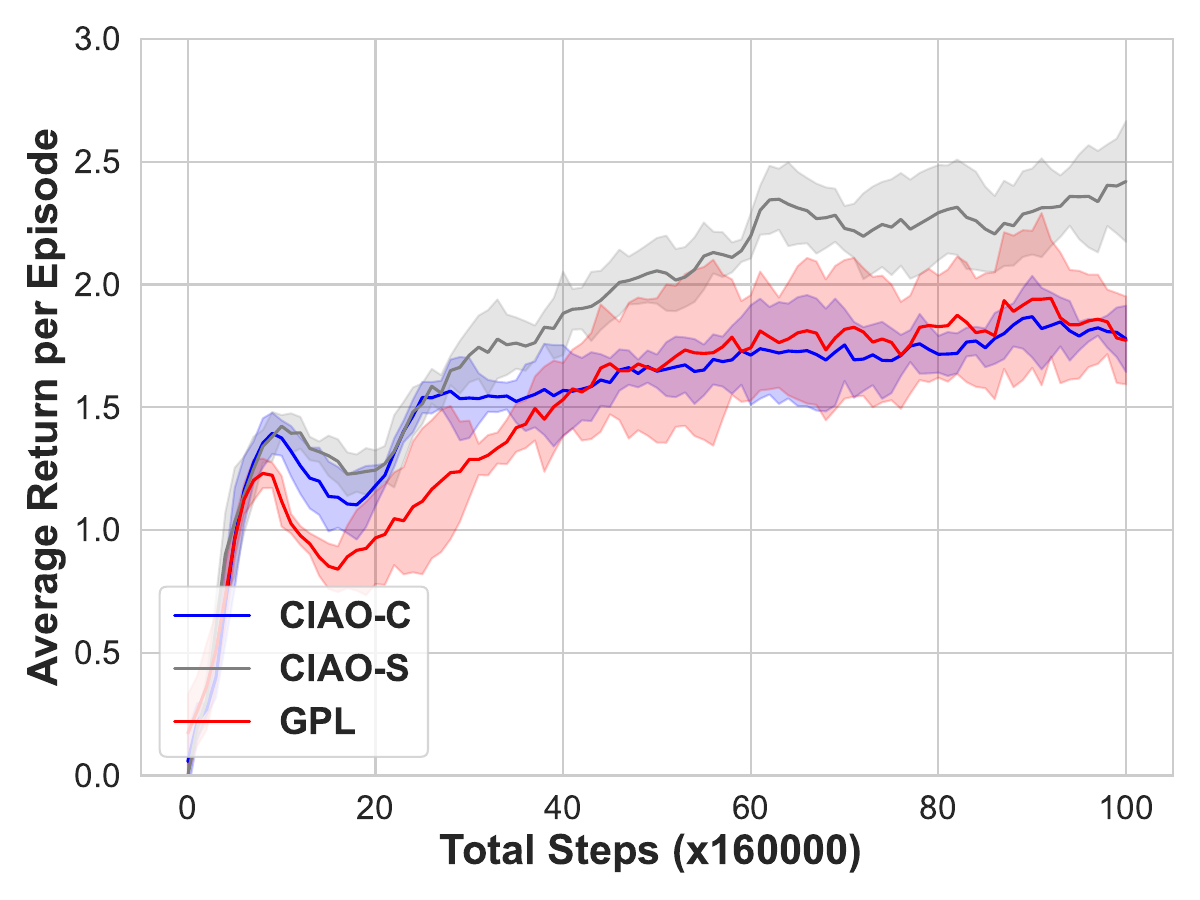}
                    \caption{LBF: intersection.}
                    \label{fig:lbf-in}
                \end{subfigure}
                \begin{subfigure}[b]{0.238\textwidth}
                \centering
                    \includegraphics[width=\textwidth]{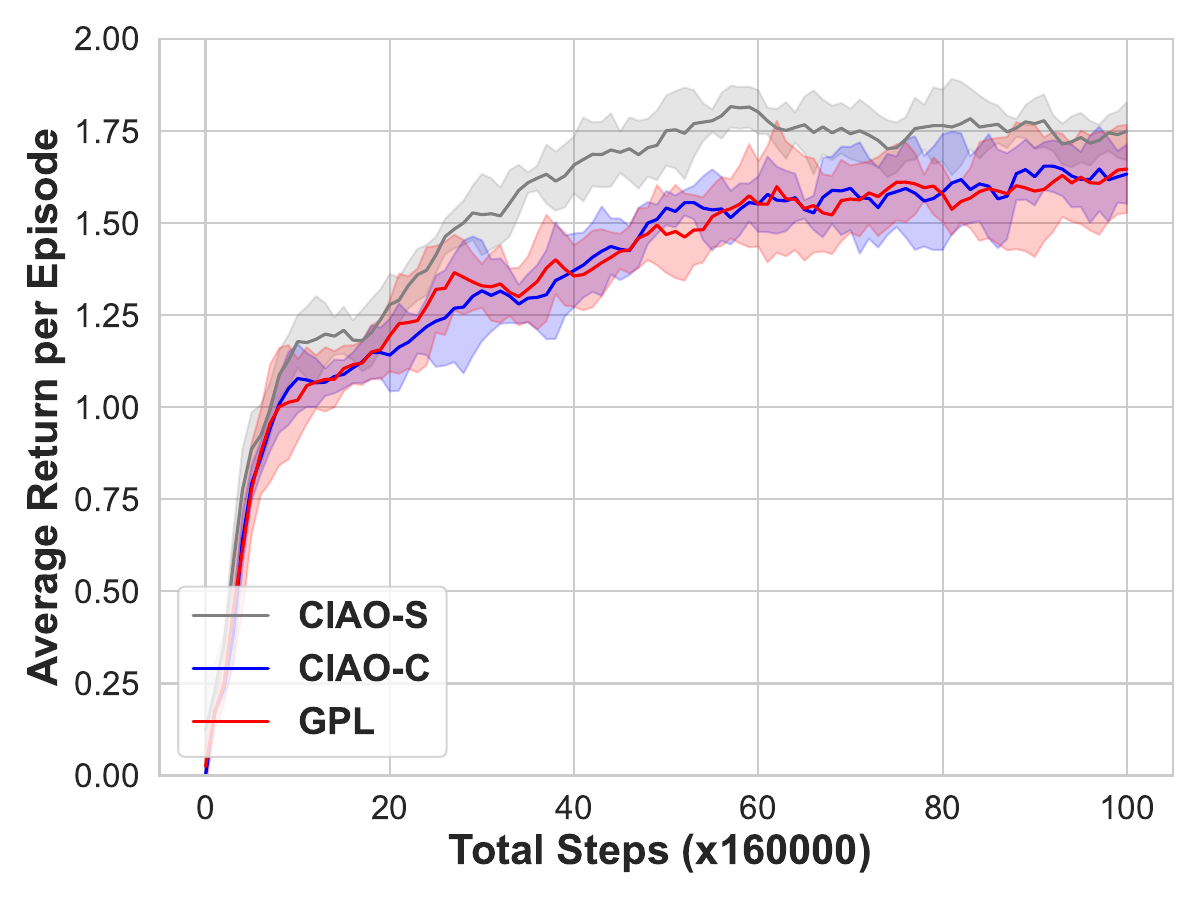}
                    \caption{LBF: mutual exclusion.}
                    \label{fig:lbf-ex}
                \end{subfigure}
                \caption{Comparison between CIAO and GPL on Wolfpack and LBF, with intersecting (denoted as ``intersection'') and mutually exclusive (denoted as ``mutual exclusion'') agent-type sets in training and testing, respectively. The maximum temporary team size is 5.}
            \label{fig:variant_agent-types}
            \end{figure}
    
\section{Conclusion}
\label{sec:discussion_and_future_work}
    \textbf{Discussion.} In this work we address the challenging problem of open ad hoc teamwork, aiming to design an agent capable of collaborating with teammates without prior coordination under dynamically changing team compositions. We propose a novel approach by incorporating cooperative game theory to develop a new theory. This theory effectively gives an interpretation to the joint Q-value representation leveraged in the state-of-the-art algorithm, GPL. In addition, the learning paradigm employed in GPL's framework is understood based on our theory. Building upon the empirical foundation of GPL, we introduce a novel algorithm, CIAO, which includes an additional regularizer and a constraint for representation thanks to our theory. Consequently, CIAO can be seen as an extension of GPL, providing extra information through our theory to narrow down the joint Q-value's hypothesis space, to improve learning efficiency. Furthermore, the incorporation of dynamic affinity graphs into OSB-CAG opens up a new avenue of designing graphs describing agent relationships aligned to game objectives. Experimental results validate the effectiveness of our theory and demonstrate the superior performance of CIAO.
    
    \textbf{Limitation and Future Work.} This work is the first in establishing both a theory and a practical algorithm rooted in cooperative game theory to address ad hoc teamwork. It opens up avenues of several promising future directions. Firstly, to enhance the scope and applicability of our theory, a logical next step involves exploring the adaptivity of teammates with time-varying agent-types, a factor currently omitted in our theory for simplicity. Another compelling direction is investigating the design of understandable joint Q-value representation for open ad hoc teamwork, other than linear decomposition with pairwise relationships and individual values justified in this work. This thread can push forward the potential deployment of ad hoc teamwork to safety-critical environments requiring trustworthy and cost-saving solutions, with fewer trial-and-error interactions.
    
\section*{Impact Statement}
\label{sec:impact_statement}
    The outcomes of this paper could significantly enhance the progress of autonomous vehicles, smart grids, and various decision-making scenarios involving multiple independently controlled agents under uncertainties. However, it is crucial to acknowledge potential drawbacks. Like many machine learning algorithms, our work may encounter challenges related to human value alignment, when the targets in interaction are humans in the potential applications. Addressing this concern is part of our ongoing research, building upon findings from related fields that emphasize alignment issues.

\section*{Acknowledgement}
    This work is partially supported by UKRI Turing AI World-Leading Researcher Fellowship, EP/W002973/1. The computational resources are supported by CSC -- IT Center for Science LTD., Finland. Yuan Zhang receives funding from the European Union’s Horizon 2020 research and innovation program under the Marie Skłodowska-Curie grant agreement No. 953348 (ELO-X).    
\bibliography{icml2024}
\bibliographystyle{icml2024}

\newpage
\appendix
\onecolumn

\section{Related Works}
\label{sec:related_work}
    \textbf{Theoretical Models for Ad Hoc Teamwork.} In our review of theoretical models for describing ad hoc teamwork (AHT), we begin by discussing foundational works. \citet{brafman1996partially} pioneered the study of ad hoc teamwork by investigating the repeated matrix game with a single teammate. Subsequent contributions extended this line of inquiry to scenarios involving multiple teammates, as exemplified by \citet{agmon2012leading}, who expanded the analysis to incorporate multiple teammates. \citet{agmon2014modeling} further relaxed assumptions by allowing teammates' policies to be drawn from a known set. Stone et al. \citet{stone2010teach} proposed collaborative multi-armed bandits, initially formalizing AHT but with notable assumptions, such as knowing teammates' policies and environments. \citet{albrecht2013game} introduced the stochastic Bayesian game (SBG) as the first complete theoretical model for addressing dynamic environments and unknown teammates in AHT. Building upon the SBG, \citet{rahman2021towards} proposed the open stochastic Bayesian game (OSBG) to address open ad hoc teamwork (OAHT). \citet{ZintgrafDCWH21} modelled AHT as interactive Bayesian reinforcement learning (IBRL) in Markov games, focusing on solving non-stationary teammates' policies within episodes. In contrast, \citet{xie2021learning} introduced a hidden parameter Markov decision process (HiP-MDP) to address scenarios where teammates' policies vary across episodes but remain stationary within each episode. In this paper, we contribute to the theoretical landscape of AHT by extending the coalitional affinity game (CAG) from the perspective of cooperative game theory, under the assumptions similar to SBG and OSBG. In more details, we introduce a novel theoretical model, referred to as \textbf{O}pen \textbf{S}tochastic \textbf{B}ayesian \textbf{C}oalitional \textbf{A}ffinity \textbf{G}ame (OSB-CAG), shedding light on the interactive process between the learner and temporary teammates. This theoretical model can be seen as an extension of OSBG (see Appendix~\ref{sec:general_ad_hoc_teamwork_framework}), where the relationship between agents is conceptualized as a dynamic affinity graph in theory, moving beyond treating the graph solely as an implementation tool.\footnote{If the dynamic affinity graph is with no edges, the OSB-CAG will degrade to a plain OSBG.} Our proposed solution concept, DVSC, provides a fresh perspective on how the learner can find optimal policies to attract temporary teammates for effective collaboration. Furthermore, we introduce a more specified transition function under our theoretical model in place of the one proposed by \citet{rahman2021towards}. The main benefit of our proposed transition function is that it enjoys a strong relationship to the underlying assumptions, and explicitly subsumes the concrete interactive process described by \citet{rahman2021towards}.
    
    \textbf{Algorithms for Ad Hoc Teamwork.} We now review AHT from an algorithmic standpoint. The best response algorithm \citep{stone2009leading}, initially proposed under the assumptions of a matrix game and well-known teammates' policies, laid the foundation for algorithmic solutions in this domain. Extending this work, REACT \citep{agmon2014modeling} emerged as a solution effective for matrices where teammates' policies are drawn from a known set. \citet{wu2011online} introduced a novel approach using biased adaptive play to estimate teammates' actions based on their historical actions. They combined this with Monte Carlo tree search to plan the ad hoc agent's actions. HBA \citep{albrecht2013game} expanded the scope beyond matrix games, maintaining a probability distribution of predetermined agent-types and maximizing long-term payoffs through an extended Bellman operator. PLASTIC-Policy \citep{barrett2017making} addressed more realistic scenarios, such as RoboCup \citep{kalyanakrishnan2007half}, by training teammates' policies through behavior cloning and the ad hoc agent's policy through FQI \citep{ernst2005tree}. AATEAM \citep{chen2020aateam} extended PLASTIC-Policy, incorporating an attention network \citep{bahdanau2014neural} to enhance the estimation of unseen agent-types. \citet{rahman2021towards} integrated modern deep learning techniques, including GNNs and RL algorithms, into HBA to address open ad hoc teamwork (OAHT) and introduced GPL. ODITS \citep{GuZH022} was proposed to handle teammates with rapidly changing behaviors under partial observability. In this paper, we introduce CIAO, a novel algorithm based on our proposed theory (OSB-CAG with DVSC as a solution concept). Specifically, CIAO extends the joint Q-value representation and training loss of GPL. Additionally, CIAO generalizes the implementation of training losses to various structures of the dynamic affinity graph, known as the coordination graph in GPL, with theoretical guarantees. This provides a design paradigm of training loss to facilitate the investigation of diverse dynamic affinity graph structures. This paradigm not only can cater for various scenarios of applications, but also can facilitate realizing the ideas inspired by other fields. Furthermore, we prove in theory and demonstrate in experiments that the existing GPL training loss is a viable approximation of the exact learning paradigm under our theory.
    
    \textbf{Relationship to Cooperative Multi-Agent Reinforcement Learning.} Cooperative multi-agent reinforcement learning (MARL) primarily aims at training and controlling agents altogether to optimally achieve a shared goal. The key research topics are credit assignment (also known as value decomposition in some literature) \citep{foerster2018counterfactual,SunehagLGCZJLSL18,RashidSWFFW18}, reward shaping \citep{du2019liir,mguniligs}, and communication \citep{foerster2016learning,sukhbaatar2016learning,jiang2018learning,kim2019learning}. In this paper, we shift the focus to AHT, where only one agent (referred to as learner) is controllable and trained to collaborate with an unknown set of uncontrollable agents to achieve a shared goal. Although the teammates' behaviours in AHT can be influenced by the learner's action (under assumption that they are capable of reacting to the learner's action) \citep{mirsky2022survey}, the joint policy may still be sub-optimal owing to either the reactivity of teammates or the effectiveness to attract teammates in implementation. On the other hand, a transferable utility game known as the convex game, belonging to cooperative game theory was introduced for employing Shapley value as a credit assignment scheme with theoretical guarantees and interpretation, to address credit assignment \citep{Wang_2020,wang2022shaq}. In this paper, we introduce CAG, belonging to non-transferable utility games (a broader class including transferable utility games), for establishing a graph-based joint Q-value representation with theoretical guarantees and understandings to address OAHT.

\section{Open Stochastic Bayesian Game}
\label{sec:general_ad_hoc_teamwork_framework}
    We now review the open stochastic Bayesian game (OSBG) that describes the open ad hoc teamwork for establishing GPL \citep{rahman2021towards}. It is defined as a tuple such that $\langle \mathcal{N}, \mathcal{S}, (\mathcal{A}_{j})_{j \in \mathcal{N}}, \Theta, R, T, \gamma \rangle$. $\mathcal{N}$ is a set of all possible agents; $\mathcal{S}$ is a set of states; $\mathcal{A}_{j}$ is agent $j$'s action set; $\Theta$ is a set of all possible agent-types. Let the joint action set under a variable agent set $\mathcal{N}_{t} \subseteq \mathcal{N}$ be defined as that $\mathcal{A}_{{\scriptscriptstyle \mathcal{N}}_{t}} = \times_{j \in \mathcal{N}_{t}} \mathcal{A}_{j}$. Therefore, the joint action space under the variable number of agents is defined as that $\mathcal{A}_{\scriptscriptstyle\mathcal{N}} = \bigcup_{\mathcal{N}_{t} \in \mathbb{P}(\mathcal{N})} \{a \vert a \in \mathcal{A}_{{\scriptscriptstyle \mathcal{N}}_{t}} \}$, while the joint agent-type space under the variable number of agents is defined as that $\Theta_{\scriptscriptstyle\mathcal{N}} = \bigcup_{\mathcal{N}_{t} \in \mathbb{P}(\mathcal{N})} \{\theta \vert \theta \in \Theta^{{\scriptscriptstyle |\mathcal{N}_{t}|}} \}$. $R: \mathcal{S} \times \mathcal{A}_{\scriptscriptstyle\mathcal{N}} \rightarrow \mathbb{R}$ is the learner's reward. $T: \mathcal{S} \times \Theta_{\scriptscriptstyle\mathcal{N}} \times \mathcal{A}_{\scriptscriptstyle\mathcal{N}} \rightarrow \mathcal{S} \times \Theta_{\scriptscriptstyle\mathcal{N}}$ is a transition function to describe the evolution of states and agents of variable types. The learner's action value function $Q^{\pi^{i}}(s_{t}, a_{t}^{i})$ is defined as follows:
    \begin{equation*}
        Q^{\pi^{i}}(s_{t}, a_{t}^{i}) = \mathbb{E}_{a_{t}^{-i} \sim \pi_{t}^{-i}} \left[ Q^{\pi^{i}}(s_{t}, a_{t}^{-i}, a_{t}^{i}) \right] = \mathbb{E}_{\substack{s_{t}, \theta_{t} \sim T, a_{t}^{-i} \sim \pi_{t}^{-i}, a_{t}^{i} \sim \pi^{i}}} \Big[ \sum_{t=0}^{\infty} \gamma^{t} R(s_{t}, a_{t})\Big],
    \end{equation*}
    where $\gamma \in [0, 1)$ is a discount factor; $s_{t}$ is a state at timestep $t$, $a_{t}^{-i}$ is a joint action of teammates $-i$ at timestep $t$ and $a_{t}^{i}$ is the learner $i$'s action at timestep $t$; $\pi^{i}$ is the learner's stationary policy and $\pi_{t}^{-i}$ is a joint policy of teammates $-i$; $Q^{\pi^{i}}(s_{t}, a_{t}^{-i}, a_{t}^{i})$ is a joint Q-value. The learner's policy $\pi^{i,*}$ is optimal, if and only if $Q^{\pi^{i,*}}(s_{t}, a_{t}^{i}) \geq Q^{\pi^{i}}(s_{t}, a_{t}^{i})$ for all $\pi^{i}, s_{t}, a_{t}^{i}$. The teammates' joint policy is expressed as that $\pi_{t}^{-i}: \mathcal{S} \times \Theta_{\scriptscriptstyle\mathcal{N}} \rightarrow \Delta(\mathcal{A}_{\scriptscriptstyle{\mathcal{N}}})$. Although learning an additional learner's policy is not necessary, the definition of teammates' joint policy still implies that the learner's policy can be expressed as that $\pi^{i}: \mathcal{S} \times \Theta \rightarrow \Delta(\mathcal{A}_{i})$ for consistency. As a result, any ad hoc team's joint policy (including the learner's policy) can be expressed as that $\pi_{t}: \mathcal{S} \times \Theta_{\scriptscriptstyle\mathcal{N}} \rightarrow \Delta(\mathcal{A}_{\scriptscriptstyle{\mathcal{N}}})$ for consistency. The learner is unable to observe the teammates' types and their policies, which can only be inferred through the history states and actions. The learner's decision making is conducted by selecting the actions that maximize $Q^{\pi^{i}}(s_{t}, a_{t}^{i})$.

\section{Further Details of Implementation}
\label{sec:further_details_of_implementation}
    Given the learner's lack of knowledge about $P_{E}$ and $\pi_{t}^{-i}$, it is essential to discuss strategies for estimating these terms to achieve the convergence of Eq.~\eqref{eq:open_team_hedonic_bellman_operator}. In the GPL framework, these two terms are implemented as the type inference model and the agent model, respectively. The implementation details are presented below.

    \subsection{GPL Framework}
    \label{sec:gpl_framework}
        We now review the GPL's empirical framework \citep{rahman2021towards}. This framework consists of the following modules: the type inference model, the joint action value model and the agent model. We only summarize the model specifications. Note that while the original GPL framework is oriented towards a fixed coordination graph, specifically a complete graph, we relax this constraint to accommodate any graph structures as needed. 
    
        \textbf{Type Inference Model.} This is modelled as a LSTM \citep{hochreiter1997long} to infer agent-types of a team at timestep $t$ given that of a team at timestep $t-1$. The agent-type is modelled as a fixed-length hidden-state vector of LSTM, named as agent-type embedding. At each timestep $t$, the state information of an emergent team $\mathcal{N}_{t}$ is reproduced to a batch of agents' information $B_{t} = [\langle u_{t}, x_{t,1} \rangle, ..., \langle u_{t}, x_{t,{\scriptscriptstyle | \mathcal{N}_{t} | }} \rangle]^{\top}$, where each agent is preserved a vector composing $u_{t}$ and $x_{t,i}$ which are observations and agent specific information extracted from state $s_{t}$. Along with additional information such as the agent-type embedding of $\mathcal{N}_{t-1}$ and the cell state, LSTM estimates the agent-type embedding of $\mathcal{N}_{t}$. To address the situation of changing team size, at each timestep the agent-type embedding of the agents who leave a team would be removed, while the new added agents' agent-type embedding would be set to a zero vector.
    
        \textbf{Joint Action Value Model.} The joint Q-value, denoted as $\hat{Q}^{\pi^{i}}(s_{t}, a_{t})$, is approximated as the sum of the corresponding individual utilities, $\hat{Q}_{j}^{\pi^{i}}(a_{t}^{j} \vert s_{t})$, and pairwise utilities, $\hat{Q}_{jk}^{\pi^{i}}(a_{t}^{j}, a_{t}^{k} \vert s_{t})$, based on the coordination graph structure. The approximation is expressed as follows:
        \begin{equation*}
            \hat{Q}^{\pi^{i}}(s_{t}, a_{t}) = \sum_{j \in \mathcal{N}_{t}} \hat{Q}_{j}^{\pi^{i}}(a_{t}^{j} \vert s_{t}) + \sum_{(j,k) \in \mathcal{E}_{t}} \hat{Q}_{jk}^{\pi^{i}}(a_{t}^{j}, a_{t}^{k} \vert s_{t}).
        \end{equation*}
        
        Both $\hat{Q}_{j}^{\pi^{i}}(a_{t}^{j} \vert s_{t})$ and $\hat{Q}_{jk}^{\pi^{i}}(a_{t}^{j}, a_{t}^{k} \vert s_{t})$ are implemented as multilayer perceptrons (MLPs) parameterised by $\beta$ and $\delta$, denoted as $\text{MLP}_{\beta}$ and $\text{MLP}_{\delta}$. The input of $\text{MLP}_{\beta}$ is the concatenation of the learner's agent-type embedding $\theta_{t}^{i}$ and the teammate $j$'s agent-type embedding $\theta_{t}^{j}$. Its output is a vector with a length of $|\mathcal{A}_{j}|$ estimating $Q_{j}^{\pi^{i}}(a_{t}^{j} \vert s_{t})$. The detailed expression is shown as follows:
        \begin{equation*}
            \hat{Q}_{j}^{\pi^{i}}(a_{t}^{j} \vert s_{t}) = \text{MLP}_{\beta}(\theta_{t}^{j}, \theta_{t}^{i})(a_{t}^{j}).
        \end{equation*}
    
        The pairwise utility $\hat{Q}_{jk}^{\pi^{i}}(a_{t}^{j}, a_{t}^{k} \vert s_{t})$ is approximated by low-rank factorization, as follows:
        \begin{equation*}
            \hat{Q}_{jk}^{\pi^{i}}(a_{t}^{j}, a_{t}^{k} \vert s_{t}) = \big( \text{MLP}_{\delta}(\theta_{t}^{j}, \theta_{t}^{i})^{\top} \text{MLP}_{\delta}(\theta_{t}^{k}, \theta_{t}^{i}) \big)(a_{t}^{j}, a_{t}^{k}),
        \end{equation*}
        where the input of $\text{MLP}_{\delta}$ is the same as $\text{MLP}_{\beta}$; the output of $\text{MLP}_{\delta}(\theta_{t}^{j}, \theta_{t}^{i})$ is a matrix with the shape $K \times |\mathcal{A}_{j}|$ and $K \ll |\mathcal{A}_{j}|$.
    
        \textbf{Agent Model.} It is assumed that all other connected agents, as described by a coordination graph, would influence an agent's actions. To model this situation, GNN is applied to process the agent-type embedding of a temporary team, denoted as $\theta_{t}$, where each team member is represented as a node. More specifically, a GNN model called relational forward model (RFM) \citep{TacchettiSMZKRG19} parameterised by $\eta$ is applied to transform $\theta_{t}$ (as the initial node representation) to $\bar{n}_{t}$ (as the new node representation) considering other agents' effects. Then, $\bar{n}_{t}$ is employed to infer $q_{\zeta,\eta}(a_{t}^{-i} \vert s_{t})$, as the approximation of teammates' joint policy, $\pi_{t}^{-i}(a_{t}^{-i} \vert s_{t}, \theta_{t}^{-i})$. The detailed expression is as follows:
        \begin{equation*}
            \begin{split}
                q_{\zeta,\eta}(a_{t}^{-i} \vert s_{t}) = \prod_{j \in -i} q_{\zeta, \eta}(a_{t}^{j} \vert s_{t}), \\
                q_{\zeta,\eta}(a_{t}^{j} \vert s_{t}) = \text{Softmax}(\text{MLP}_{\eta}(\bar{n}_{t}^{j}))(a_{t}^{j}).
            \end{split}
        \end{equation*}
    
        \textbf{Learner's Action Value Model.} Substituting the agent model and the joint action value model defined above into Eq.~\eqref{eq:objective_ad_hoc}, the learner's Q-value for its own decision making is approximated as follows:
        \begin{equation*}
            \begin{split}
                \hat{Q}^{\pi_{i}}(s_{t}, a_{t}^{i}) &= \hat{Q}_{i}^{\pi^{i}}(a_{t}^{i} \vert s_{t}) + \sum_{\substack{a_{t}^{j} \in \mathcal{A}_{j}, (j, i) \in \mathcal{E}_{t}}} \Big( \hat{Q}_{j}^{\pi^{i}}(a_{t}^{j} \vert s_{t}) + \hat{Q}_{ij}^{\pi^{i}}(a_{t}^{i}, a_{t}^{j} \vert s_{t}) \Big) q_{\zeta,\eta}(a_{t}^{j} \vert s_{t}) \\
                &+ \sum_{\substack{a_{t}^{j} \in \mathcal{A}_{j}, a_{t}^{k} \in \mathcal{A}_{k}, \\ (j,k) \in \mathcal{E}_{t}}} \hat{Q}_{jk}^{\pi^{i}}(a_{t}^{j}, a_{t}^{k} \vert s_{t}) q_{\zeta,\eta}(a_{t}^{j} \vert s_{t}) q_{\zeta,\eta}(a_{t}^{j} \vert s_{t}).
            \end{split}
        \end{equation*}

    \subsection{Overall Training Procedure of CIAO}
    \label{subsec:algo_ciao}
        We now summarize the overall training procedure of CIAO in Algorithm~\ref{alg:ciao}. Note that in the GPL framework, the type inference model is absorbed into the joint Q-value and the agent model as a LSTM, respectively. This construction aims to prevent these two models' gradients from interfering against each other during training \citep{rahman2021towards}.
        \begin{algorithm}[ht!]
        \caption{Overall training procedure of CIAO}
        \label{alg:ciao}
            \begin{algorithmic}
                \STATE {\bfseries Input:} dynamic affinity graph structure $G$, number of training episodes $e$, length of an episode $T$, replay buffer $\mathcal{B}$
                \REPEAT
                \STATE Clear the replay buffer $\mathcal{B}$.
                \STATE Reset the environment and receive the initial observations.
                    \FOR{$\text{timestep}=1$ {\bfseries to} $T$}
                       \STATE Execute learner's action by $\epsilon$-greedy policy.
                       \STATE Store observations (including teammates' actions) for the current timestep in the replay buffer $\mathcal{B}$.
                   \ENDFOR
                   \STATE Generate the joint Q-value and the agent model as per GPL framework, based on the dynamic affinity graph $G$.
                   \STATE Update parameters of pairwise utilities and individual utilities by the loss function proposed in Section~\ref{subsec:practical_implementation}.
                   \STATE Update parameters of the agent model by the following loss function:
                   \begin{equation*}
                       L(\zeta,\eta) = - \frac{1}{T} \sum_{t=1}^{T} \log q_{\zeta,\eta}(a_{t}^{-i} \vert s_{t}).
                   \end{equation*}
               \UNTIL{meeting the number of training episodes $m$}
            \end{algorithmic}
        \end{algorithm}

\section{Assumptions}
\label{sec:appendix_assumptions}
    In addition to the assumptions of ad hoc teamwork described in the survey~\citep{mirsky2022survey}, we also list the necessary assumptions for this paper as follow.
    \begingroup
    \def\theassumption{\ref{assm:conditional_independencies}}
        \begin{assumption}
            The following conditional independencies are assumed to hold in any distribution $P$ over the set of variables in an OSB-CAG: (1) $( \theta_{t} \bigCI \theta_{t-1}, s_{t-1}, a_{t-1} \ \vert \ \mathcal{N}_{t}, s_{t} )$; (2) $(\mathcal{N}_{t}, s_{t} \bigCI \theta_{t-1} \ \vert \ \mathcal{N}_{t-1}, s_{t-1}, a_{t-1})$; (3) $(\mathcal{N}_{t} \bigCI a_{t} \vert s_{t}, \theta_{t})$; (4) $( \theta_{t}^{j} \bigCI -j, \theta_{t}^{-j} \ \vert \ \{j\}, s_{t} )$.
        \end{assumption}
    \endgroup
    Assumption \ref{assm:conditional_independencies} indicates the assumptions encoding the relationships among random variables that are entailed by any probability distribution describing the \textit{open ad hoc teamwork} process, referred to as conditional independencies \citep[Ch. 2]{koller2009probabilistic}.
    
    As for conditional independence (1), it implies that the agent-types $\theta_{t}$ for the current timestep are conditionally independent of the related variables $\theta_{t-1}, s_{t-1}, a_{t-1}$ for the preceding timestep, given the agent set $\mathcal{N}_{t}$ and the state $s_{t}$ for the current timestep. This is reflected by $P_{E}(\theta_{t} \vert \mathcal{N}_{t}, s_{t}) = P(\theta_{t} \vert \mathcal{N}_{t}, s_{t}, s_{t-1}, a_{t-1}, \theta_{t-1})$.

    As for conditional independence (2), it implies that the agent set $\mathcal{N}_{t}$ and the state $s_{t}$ for the current timestep is independent of the agent-types $\theta_{t-1}$ for the preceding timestep, given the variables $\mathcal{N}_{t-1}, s_{t-1}, a_{t-1}$ for the preceding timestep. This is reflected by $P_{T}(\mathcal{N}_{t}, s_{t} \vert \mathcal{N}_{t-1}, s_{t-1}, a_{t-1}) = P(\mathcal{N}_{t}, s_{t} \vert \mathcal{N}_{t-1}, s_{t-1}, a_{t-1}, \theta_{t-1})$.

    As for conditional independence (3), it implies that the agent set $\mathcal{N}_{t}$ is independent of the joint action $a_{t}$, given the state $s_{t}$ and the agent-type set $\theta_{t}$ for the same timestep. This is reflected by $P(\mathcal{N}_{t} \vert s_{t}, \theta_{t}) = P(\mathcal{N}_{t} \vert s_{t}, a_{t}, \theta_{t})$. Note that this condition coincides with scenarios encoded by Assumption \ref{assm:agent_type_set}, where the agent $j$'s policy is able to be varied across timesteps, and the policy is only correlated with its agent-type. In turn, this implies that an agent's mind could be changed across timesteps, which is an evidence that open ad hoc teamwork is also suitable for modelling human-AI cooperation \citep{shneiderman2020human,de2023zero}. However, for clarity and simplicity to introduce our theory, we assume in this paper that the policy is fixed (time invariant or stationary) across timesteps, as shown in Assumption \ref{assm:agent_type_fixed_policy}.
    
    As for conditional independence (4), it implies that an agent $j$'s agent-type $\theta_{t}^{j}$ for some timestep is conditionally independent of other agents $-j$ and their agent-types $\theta_{t}^{-j}$, given itself denoted as $j$ and the state $s_{t}$ for that timestep. This is reflected by $\prod_{j=1}^{|\mathcal{N}_{t}|} P_{A}(\theta_{t}^{j} \vert \{j\}, s_{t}) = P_{E}(\theta_{t} \vert \mathcal{N}_{t}, s_{t})$.
    
    \begingroup
    \def\theassumption{\ref{assm:agent_leaves_env}}
        \begin{assumption}
            Suppose that $\alpha_{jk}(a_{t}^{j}, a_{t}^{k} \vert s_{t}) = 0$ for $t \geq T$, where $T$ is the timestep when agent $j$ or $k$ leaves the environment, and $R_{j}(a_{t}^{j} \vert s_{t})=0$ for $t \geq T'$, where $T'$ is the timestep when agent $j$ leaves the environment.
        \end{assumption}
    \endgroup
    Assumption~\ref{assm:agent_leaves_env} introduces a metric to quantify the impact of agents leaving the environment. Essentially, it posits that an agent that has departed from the environment no longer exerts any influence on the remaining agents within the environment.
    
    \setcounter{assumption}{2}
    \begin{assumption}
    \label{assm:agent_type_set}
        There exists an underlying agent-type set to generate ad hoc teammates in an environment which is unknown to the learner.
    \end{assumption}
    Assumption~\ref{assm:agent_type_set} provides a natural framework for describing the agent-types of teammates. In scenarios where the agent-type set is sufficiently large, traversing all possible agent-types or compositions becomes impractical. Therefore, this assumption ensures that the generalizability of open ad hoc teamwork is not compromised.
    
    \begin{assumption}
    \label{assm:influence_of_learner}
        Teammates can be influenced by the learner through its decision making. 
    \end{assumption}
    Assumption~\ref{assm:influence_of_learner} constitutes a fundamental and commonly assumed property essential for rationalizing the ad hoc teamwork problem. Often referred to as the reactivity of teammates \citep{barrett2017making}, this assumption posits that teammates must be capable of reacting to or being influenced by the learner. Without this interaction, the problem would regress to a scenario akin to a single-agent problem, where teammates merely function as moving `obstacles.' To avert such a pathological situation, maintaining this assumption serves as a crucial boundary for ad hoc teamwork.
    
    \begin{assumption}
    \label{assm:agents_stay_for_a_while}
        The agents stay in the environment at least for a period of timesteps.
    \end{assumption}
    Assumption~\ref{assm:agents_stay_for_a_while} is a prerequisite ensuring the feasibility of completing arbitrary tasks. Without this condition, wherein an agent joining at a given timestep remains in the environment for a non-instantaneous duration, there would be minimal opportunity for teams of agents to react to and influence each other effectively.
    
    \begin{assumption}
    \label{assm:agent_type_fixed_policy}
        Each teammate of an arbitrary agent-type is equipped with a fixed policy.
    \end{assumption}
    Assumption~\ref{assm:agent_type_fixed_policy} serves as a simplified condition for analyzing the learner's convergence to the optimal policy. By assuming fixed policies for teammates, the Markov process becomes stationary from the learner's perspective, facilitating a more tractable analysis of convergence dynamics. However, this can be further relaxed to cater for more realistic situations.

    \begin{assumption}
    \label{assm:generalizability_gnn}
        Graph neural networks are assumed to enjoy the generalizability to agent-types.
    \end{assumption}
    Assumption~\ref{assm:generalizability_gnn} serves as a justification to explain why graph neural networks as a critical component can result in better performance for GPL as reported in the previous work~\citep{rahman2021towards} and therefore CIAO (which is implemented upon GPL) in this paper, in generalization of agent-type sets. This assumption can be further relaxed through investigating the underlying reasons behind this phenomenon.
    
\section{Generalization of Preference Values for Coalitional Affinity Game}
\label{sec:generalization_of_preference_values}
    At the beginning, it is worth noting that in the original work of CAG \citep{branzei2009coalitional}, the definition of the preference value of an arbitrary agent $j$ is as follows:
    \begin{equation}
    \label{eq:org_def_coalition_value_cag}
        \bar{v}_{j}(\mathcal{C}) = 
        \begin{cases}
            0 & \text{if $\mathcal{C} = \{j\}$},\\
            \sum_{(j,k) \in \mathcal{E}, k \in \mathcal{C}} \bar{w}(j, k) & \text{otherwise}.
        \end{cases}
    \end{equation}

    While the condition that each agent's preference value of a coalition including only itself, equal to zero, is convenient and straightforward for analysis, it imposes limitations on the representational capacity for a wide variety of situations. To address this issue, we generalize the definition of the preference value function in Eq.~\eqref{eq:org_def_coalition_value_cag} to the form as follows:
    \begin{equation}
    \label{eq:new_def_coalition_value_cag}
        v_{j}(\mathcal{C}) = 
        \begin{cases}
            b_{j} \geq 0 & \text{if $\mathcal{C} = \{j\}$},\\
            \sum_{(j,k) \in \mathcal{E}, k \in \mathcal{C}} w(j, k) & \text{otherwise}.
        \end{cases}
    \end{equation}
    The main difference between the definitions in Eq.~\eqref{eq:org_def_coalition_value_cag} and Eq.~\eqref{eq:new_def_coalition_value_cag} is that the preference value of the coalition only including a single agent is not forced to be zero in Eq.~\eqref{eq:new_def_coalition_value_cag}. Albeit that the results shown in the original work of CAG \citep{branzei2009coalitional} are based on each agent's original preference value function shown in Eq.~\eqref{eq:org_def_coalition_value_cag}, we can still generalize and leverage the results through conducting translation to each agent's preference value function by its preference value of the coalition including itself, $v_{j}(\{j\}) = b_{j} \geq 0$ (where $b_{j}$ is a constant), to align with the condition of $\bar{v}_{j}(\mathcal{C})$ in Eq.~\eqref{eq:org_def_coalition_value_cag}. In more details, we can transform the newly defined preference value function in Eq.~\eqref{eq:new_def_coalition_value_cag} as follows:
    \begin{equation}
    \label{eq:transform_new_def_coalition_value_cag}
        \hat{v}_{j}(\mathcal{C}) = v_{j}(\mathcal{C}) - v_{j}(\{j\}) = 
        \begin{cases}
            0 & \text{if $\mathcal{C} = \{j\}$},\\
            \sum_{(j,k) \in \mathcal{E}, k \in \mathcal{C}} w(j, k) - v_{j}(\{j\}) & \text{otherwise}.
        \end{cases}
    \end{equation}
    Therefore, we can directly leverage the results from the previous work \citep{branzei2009coalitional} by replacing $\bar{v}_{j}(\mathcal{C})$ with $\hat{v}_{j}(\mathcal{C})$, and generalize the results to the newly defined preference value function in Eq.~\eqref{eq:new_def_coalition_value_cag} by conducting the change of variables according to Eq.~\eqref{eq:transform_new_def_coalition_value_cag}. 
    
    \textit{The generalised preference value $v_{j}(\mathcal{C})$ plays an important role of proving the results of OSB-CAG in the following sections.}

    \begin{definition}
    \label{def:singleton_coalition_preference_value}
        In a CAG with the generalised preference value function, for any agent $j$, its preference value of the coalition including only itself is defined as that $v_{j}(\{j\}) \geq 0$.
    \end{definition}
    In the conventional definition of a coalition value function\footnote{The preference value function of an agent can be seen as a coalition value function specifically defined for the agent.} in the cooperative game theory, the value of an empty set (empty coalition) is defined as zero \citep{chalkiadakis2022computational}. In the context of a CAG, we can formally extend the domain of an agent $j$'s preference value function by considering the empty set such that $v_{j}(\emptyset) = 0$. This extension can be interpreted as that an agent imagines a scenario where it is not included (with no incentives to join). If $v_{j}(\{j\}) < 0$, it may lead to a paradox that an agent $j$ would choose to disappear from the environment (e.g. suicide) to escape independence owing to $v_{j}(\{j\}) < v_{j}(\emptyset)$, which is apparently opposite to morality and ethics. To avoid the paradox, it is reasonable to generalise an agent $j$'s preference value of the coalition including itself to only the non-negative range such that $v_{j}(\{j\}) \geq 0$. 
        
\section{Derivation of Definition~\ref{def:symmetry_strict_core_stability}}
\label{sec:derivation_details_of_definition_DVSC}
    \begin{definition}
    \label{def:inner_stability}
        We say that a blocking coalition $\mathcal{C}$ \textit{weakly blocks} a coalition structure $\mathcal{CS}$ if every agent $j \in \mathcal{C}$ weakly prefers $\mathcal{C}$ to $\mathcal{CS}(j)$ and there exists at least one agent $k \in \mathcal{C}$ who strictly prefers $\mathcal{C}$ to $\mathcal{CS}(j)$. A coalition structure $\mathcal{CS} = \{ \mathcal{C}_{1}, ..., \mathcal{C}_{m} \}$ admitting no weakly blocking coalition $\mathcal{C} \ \mathlarger{\mathlarger{\subset}} \ \mathcal{C}_{k}$, for some $1 \leq k \leq m$, is called inner stable.
    \end{definition}
    \begin{theorem}[\citet{branzei2009coalitional}]
    \label{thm:symmetry_inner_stability}
        If a CAG is symmetric, then the social-welfare maximizing partition exhibits inner stability.
    \end{theorem}
    Theorem \ref{thm:symmetry_inner_stability} directly holds for the newly defined $v_{j}(\mathcal{C})$ in this paper, since it is irrelevant to the detailed representation (feasible domain) of a preference value function (see Theorem 2 and 5 in the previous work \citep{branzei2009coalitional}).
    
    \begin{lemma}
    \label{lemm:symmetry_strict_core_grand_coalition}
        If a CAG is symmetric, then maximizing the social welfare under a grand coalition results in strict core stability.
    \end{lemma}
    \begin{proof}
        Following Definition \ref{def:inner_stability}, it is not difficult to observe that a grand coalition exhibiting strict core stability is equivalent to a grand coalition exhibiting inner stability. Therefore, we can directly obtain the result by Theorem \ref{thm:symmetry_inner_stability}.
    \end{proof}             

    \subsection{Derivation of Dynamic Variational Strict Core}
        In an OSB-CAG, at any timestep $t$, under an arbitrary state $s_{t} \in \mathcal{S}$ along with a temporary team (including the learner $i$), denoted as $\mathcal{N}_{t} \ \mathlarger{\mathlarger{\subseteq}} \ \mathcal{N}$, and the temporary team's joint action $a_{t} \in \mathcal{A}_{\scriptscriptstyle \mathcal{N}_{t}}$, the coalition reward can be equivalently expressed as a preference value of an agent belonging to a temporary team $\mathcal{N}_{t}$ such that $R_{j}(s_{t}, a_{t}) = v_{j}(\mathcal{N}_{t})$. The temporary team $\mathcal{N}_{t}$ can be interpreted as the grand coalition at any timestep $t$. Thereby, \textit{reaching the strict core stability at any timestep $t$ is equivalent to maximizing the social welfare at the timestep}. Different from the previous work \citep{branzei2009coalitional} that given the predetermined preference values, the coalition structure is as a decision variable to reach the strict core; in this paper, we predetermine a temporary team, as the target coalition structure, and the learner $i$'s action is as an extended decision variable to change the preference values (coalition rewards) in order to reach the \textit{variational strict core} (VSC) that is defined with the same criterion as the strict core, but with different target variables as elements to form the solution set. The learner $i$'s action is generated by its policy $\pi^{i}$. By Assumption \ref{assm:influence_of_learner}, we can get that the learner's action is able to influence teammates' actions. Therefore, the teammates' coalition rewards as an evaluation of their policies will also be varied accordingly. This explains that the learner's action can be seen as a decision variable that is able to indirectly change teammates' coalition rewards. By Lemma \ref{lemm:symmetry_strict_core_grand_coalition}, if a dynamic affinity graph at timestep $t$ is symmetric, we can express the VSC for any timestep $t$ (under an arbitrary $s_{t} \in \mathcal{S}$ along with a temporary team $\mathcal{N}_{t} \ \mathlarger{\mathlarger{\subseteq}} \ \mathcal{N}$, a joint agent-type $\theta_{t} \in \Theta^{\scriptscriptstyle |\mathcal{N}_{t}|}$, the teammates' policies $\pi_{t}^{-i}(a_{t}^{-i} \vert s_{t}, \theta_{t}^{-i})$ with respect to their agent-types and the state) to find the learner's optimal action (rather than find a coalition structure in the previous work) as follows:
        \begin{equation}
        \label{eq:variational_strict_core_single_timestep}
            \texttt{VSC} := \Big\{ a^{i,*} \ \Big\vert \ \sum_{j \in \mathcal{N}_{t}} R_{j}(s_{t}, a_{t}^{i,*}, a_{t}^{-i}) \geq \sum_{j \in \mathcal{N}_{t}} R_{j}(s_{t}, a_{t}^{i}, a_{t}^{-i}), \ \forall a_{t}^{i} \in \mathcal{A}_{i} \Big\}.
        \end{equation}
        Note that the strict core defined in Eq.~\eqref{eq:variational_strict_core_single_timestep} implicitly assumes that the teammates' reaction is instantaneous (happening at the same timestep). Recall that our aim is to find the learner's optimal stationary policy $\pi^{i,*}$ that generates actions across timesteps (in a long horizon), in order to influence the temporary teammates occurring at any timestep to collaborate (meeting the strict core stability). We now generalize the VSC defined in Eq.~\eqref{eq:variational_strict_core_single_timestep} by considering the process of generating states, teammates, agent-types and teammates' actions, named as \textit{dynamic variational strict core} (DVSC). The DSVC is defined as follows:
        \begin{equation}
        \label{eq:variational_strict_core_single_long_horizon}
            \texttt{DVSC} := \Big\{ \ \pi^{i,*} \ \Big\vert \ \mathbb{E}_{\pi^{i,*}}\big[ \sum_{t=0}^{\infty} \gamma^{t} \sum_{j \in \mathcal{N}_{t}} R_{j}(s_{t}, a_{t}) \big] \geq \mathbb{E}_{\pi^{i}}\big[ \sum_{t=0}^{\infty} \gamma^{t} \sum_{j \in \mathcal{N}_{t}} R_{j}(s_{t}, a_{t}) \big], \forall s_{0} \in \mathcal{S}, \forall \pi^{i} \ \Big\},
        \end{equation}
        where $a_{t}^{i} \sim \pi^{i}$ and $a_{t}^{-i} \sim \pi_{t}^{-i}$; $\mathbb{E}_{\pi^{i}}[\cdot]$ denotes the expectation that also implicitly depends on $\theta_{t} \sim P_{E}$, $\mathcal{N}_{t}, s_{t} \sim P_{O}$. 

        Note that the VSC defined in Eq.~\eqref{eq:variational_strict_core_single_long_horizon} weakens the implicit assumption of the strict core defined in Eq.~\eqref{eq:variational_strict_core_single_timestep}. In more details, it allows the teammates to react at the successor timesteps instead of the mandatory instantaneous reaction at the same timestep. Nevertheless, this requires that the learner has potential for adapting to the teammates (through interaction with teammates for a period). By Assumption \ref{assm:agents_stay_for_a_while} and \ref{assm:agent_type_fixed_policy}, the learner's adaption to the temporary teammates is possible. 
            
\section{Mathematical Proofs}
\label{sec:appendix_proofs}
    \subsection{The Proof of Proposition \ref{prop:existence_transition_preference}}
    \label{subsec:proof_of_proposition1}
        \begingroup
        \def\theproposition{\ref{prop:existence_transition_preference}}
            \begin{proposition}
                $T(\mathcal{N}_{t}, s_{t}, \theta_{t} \vert s_{t-1}, a_{t-1}, \theta_{t-1})$ for $t \geq 1$ can be expressed in terms of the following well-defined probability distributions: $P_{I}(\mathcal{N}_{0}, s_{0})$, $P_{T}(\mathcal{N}_{t}, s_{t} \vert \mathcal{N}_{t-1}, s_{t-1}, a_{t-1})$ for $t \geq 1$, and $P_{A}(\theta_{t}^{j} \vert \{j\}, s_{t})$ for $t \geq 0$.
            \end{proposition}
        \endgroup
        \begin{proof}
            To ease the proof, we assume that $s_{t}$ and $a_{t}$ are discrete variables with no loss of generality. We prove that $T(\mathcal{N}_{t}, s_{t}, \theta_{t} \vert s_{t-1}, a_{t-1}, \theta_{t-1})$ can be expressed as the probability distributions we have defined, as follows:
            \begin{equation*}
                \begin{split}
                    &\quad T(\mathcal{N}_{t}, s_{t}, \theta_{t} \vert s_{t-1}, a_{t-1}, \theta_{t-1}) \\ 
                    &= P(\theta_{t} \vert \mathcal{N}_{t}, s_{t}, s_{t-1}, a_{t-1}, \theta_{t-1}) P_{O}(\mathcal{N}_{t}, s_{t} \vert s_{t-1}, a_{t-1}, \theta_{t-1}) \\
                    &= P_{E}(\theta_{t} \vert \mathcal{N}_{t}, s_{t}) P_{O}(\mathcal{N}_{t}, s_{t} \vert s_{t-1}, a_{t-1}, \theta_{t-1}) \quad \text{(By conditional independence (1) in Assumption \ref{assm:conditional_independencies}.)} \\
                    &= P_{E}(\theta_{t} \vert \mathcal{N}_{t}, s_{t})\sum_{\scriptscriptstyle{\mathcal{N}}_{t-1}} P(\mathcal{N}_{t}, s_{t}, \mathcal{N}_{t-1} \vert s_{t-1}, a_{t-1}, \theta_{t-1}) \\
                    &= P_{E}(\theta_{t} \vert \mathcal{N}_{t}, s_{t})\sum_{\scriptscriptstyle{\mathcal{N}}_{t-1}} P(\mathcal{N}_{t}, s_{t} \vert \mathcal{N}_{t-1}, s_{t-1}, a_{t-1}, \theta_{t-1}) P(\mathcal{N}_{t-1} \vert s_{t-1}, a_{t-1}, \theta_{t-1}) \\
                    &= P_{E}(\theta_{t} \vert \mathcal{N}_{t}, s_{t}) \sum_{\scriptscriptstyle{\mathcal{N}}_{t-1}} P_{T}(\mathcal{N}_{t}, s_{t} \vert \mathcal{N}_{t-1}, s_{t-1}, a_{t-1}) P(\mathcal{N}_{t-1} \vert s_{t-1}, \theta_{t-1}) \\
                    &\quad \text{(By conditional independence (2) and (3) in Assumption \ref{assm:conditional_independencies}.)} \\
                    &= \prod_{j=1}^{|\mathcal{N}_{t}|} P_{A}(\theta_{t}^{j} \vert \{j\}, s_{t}) \sum_{\scriptscriptstyle{\mathcal{N}}_{t-1}} P_{T}(\mathcal{N}_{t}, s_{t} \vert \mathcal{N}_{t-1}, s_{t-1}, a_{t-1}) P(\mathcal{N}_{t-1} \vert s_{t-1}, \theta_{t-1}). \\
                    &\quad \text{(By conditional independence (4) in Assumption \ref{assm:conditional_independencies}.)}
                \end{split}
            \end{equation*}
            To complete the above proof, we need to further show the expression of $P(\mathcal{N}_{t} \vert s_{t}, \theta_{t})$ as follows:
            \begin{equation}
            \label{eq:induction_condition}
                P(\mathcal{N}_{t} | s_{t}, \theta_{t}) = \frac{\sum_{s_{t}} P_{E}(\theta_{t} \vert \mathcal{N}_{t}, s_{t}) P(\mathcal{N}_{t}, s_{t})}{\sum_{\scriptscriptstyle{\mathcal{N}}_{t}} \sum_{s_{t}} P_{E}(\theta_{t} \vert \mathcal{N}_{t}, s_{t}) P(\mathcal{N}_{t}, s_{t}) }.
            \end{equation}
            Apparently, we require to prove that $P(\mathcal{N}_{t}, s_{t})$ admits factorization into the probability distributions we have defined. We now conduct this by mathematical induction as follows:
            
            \textit{Base case}: As per the definition, $P_{I}(\mathcal{N}_{0}, s_{0})$ is a predefined probability distribution to express $P(\mathcal{N}_{0}, s_{0})$ for $t = 0$.
    
            \textit{Induction case}: Assume the induction hypothesis that $P(\mathcal{N}_{t}, s_{t})$ admits factorization into the probability distributions we have defined, for any $t \geq 0$. 
            
            Next, we aim to prove that $P(\mathcal{N}_{t+1}, s_{t+1})$ admits factorization into the probability distributions we have defined and $P(\mathcal{N}_{t}, s_{t})$ as the induction hypothesis, such that
            \begin{equation*}
                P(\mathcal{N}_{t+1}, s_{t+1}) = \sum_{\scriptscriptstyle{\mathcal{N}}_{t}} \sum_{s_{t}} \sum_{a_{t}} \sum_{\theta_{t}} P(\mathcal{N}_{t+1}, s_{t+1}, \mathcal{N}_{t}, s_{t}, a_{t}, \theta_{t}),
            \end{equation*}
            where
            \begin{equation*}
                \begin{split}
                    P(\mathcal{N}_{t+1}, s_{t+1}, \mathcal{N}_{t}, s_{t}, a_{t}, \theta_{t}) &= P(\mathcal{N}_{t+1}, s_{t+1} \vert \mathcal{N}_{t}, s_{t}, a_{t}, \theta_{t}) P(\mathcal{N}_{t} \vert s_{t}, a_{t}, \theta_{t}) P(a_{t} \vert s_{t}, \theta_{t}) P(s_{t}, \theta_{t}) \\
                    &= P_{T}(\mathcal{N}_{t+1}, s_{t+1} \vert \mathcal{N}_{t}, s_{t}, a_{t}) P(\mathcal{N}_{t} \vert s_{t}, \theta_{t}) P(a_{t} \vert s_{t}, \theta_{t}) P(s_{t}, \theta_{t}) \\
                    &\quad \text{(By conditional independence (2) and (3) in Assumption \ref{assm:conditional_independencies}.)} \\
                    &= P_{T}(\mathcal{N}_{t+1}, s_{t+1} \vert \mathcal{N}_{t}, s_{t}, a_{t}) P(\mathcal{N}_{t} \vert s_{t}, \theta_{t}) \pi_{t}(a_{t} \vert s_{t}, \theta_{t}) P(s_{t}, \theta_{t}). \\
                    &\quad \text{(By the definition of $\pi_{t}$ in Appendix~\ref{sec:general_ad_hoc_teamwork_framework}, we can use $\pi_{t}$ to specify $P(a_{t} \vert s_{t}, \theta_{t})$.)}
                \end{split}
            \end{equation*}
            Since $P(\mathcal{N}_{t}, s_{t})$ is the induction hypothesis and $P_{E}(\theta_{t} \vert \mathcal{N}_{t}, s_{t})$ is a probability distribution we have defined, we can derive $P(\mathcal{N}_{t} \vert s_{t}, \theta_{t})$ by Eq.~\eqref{eq:induction_condition}. Also, by that $P(\mathcal{N}_{t}, s_{t})$ is the induction hypothesis and $P_{E}(\theta_{t} \vert \mathcal{N}_{t}, s_{t})$ is a probability distribution we have defined, we can obtain $P(s_{t}, \theta_{t})$ as follows:
            \begin{equation*}
                \begin{split}
                    P(s_{t}, \theta_{t}) &= \sum_{\mathcal{N}_{t}} P(\mathcal{N}_{t}, s_{t}, \theta_{t}) \\
                    &= \sum_{\mathcal{N}_{t}} P_{E}(\theta_{t} \vert \mathcal{N}_{t}, s_{t}) P(\mathcal{N}_{t}, s_{t}).
                \end{split}
            \end{equation*}
            \textit{Conclusion}: Since both the base case and the induction step have been proved as true, $P(\mathcal{N}_{t}, s_{t})$ is proved to admit factorization into the probability distributions we have defined for any $t \geq 0$.
        \end{proof}
        
    \subsection{The Proof of Theorem \ref{thm:dvsc_core_existence}}
        \begin{theorem}[\citet{branzei2009coalitional}]
        \label{thm:grand_coalition_core_cag}
            In a CAG with an affinity graph $G=\langle \mathcal{N}, \mathcal{E} \rangle$, if for all $(j, k) \in \mathcal{E}$, $\bar{w}(j, k) \geq 0$, then the grand coalition is in the strict core.
        \end{theorem}

        \begin{lemma}
        \label{lemm:grand_coalition_core_cag_generalised}
            In a CAG with an affinity graph $G=\langle \mathcal{N}, \mathcal{E} \rangle$ and the generalised preference value function $v_{j}(\mathcal{C})$, if the following conditions are satisfied such that
            \begin{equation}
            \label{eq:conditions_grand_coalition_core}
                \begin{split}
                    w(j, k) \geq z_{jk}(\{j\}), \\
                    v_{j}(\{j\}) = \sum_{(j,k) \in \mathcal{E}, k \in \mathcal{N}} z_{jk}(\{j\}), \\
                    \forall (j, k) \in \mathcal{E},
                \end{split}
            \end{equation}
            then the grand coalition is in the strict core.
        \end{lemma}
        \begin{proof}
            Recall that we have generalised the preference value function in this paper (see Appendix \ref{sec:generalization_of_preference_values}). Theorem \ref{thm:grand_coalition_core_cag} only holds for the case where the preference value function is defined as $\bar{v}_{j}(\mathcal{C})$ in Eq.~\eqref{eq:org_def_coalition_value_cag}. As a result, we first investigate the conditions that makes Theorem \ref{thm:grand_coalition_core_cag} still hold for the generalised preference value function $v_{j}(\mathcal{C})$ in Eq.~\eqref{eq:new_def_coalition_value_cag}. As discussed before, we can transform the generalised preference value function $v_{j}(\mathcal{C})$ to the feasible domain of the original preference value function $\bar{v}_{j}(\mathcal{C})$ by translation such that
            \begin{equation*}
                \hat{v}_{j}(\mathcal{C}) = v_{j}(\mathcal{C}) - v_{j}(\{j\}) = 
                \begin{cases}
                    0 & \text{if $\mathcal{C} = \{j\}$},\\
                    \sum_{(j,k) \in \mathcal{E}, k \in \mathcal{C}} w(j, k) - v_{j}(\{j\}) & \text{otherwise}.
                \end{cases}
            \end{equation*}
            It is apparent that the domain of $\hat{v}_{j}(\mathcal{C})$ is aligned with that of $\bar{v}_{j}(\mathcal{C})$. Therefore, we can substitute $\hat{v}_{j}(\mathcal{C})$ for $\bar{v}_{j}(\mathcal{C})$. Since Theorem \ref{thm:grand_coalition_core_cag} only considers the grand coalition, we can temporarily ignore the case of that $\mathcal{C} = \{j\}$. For any $\mathcal{C} \neq \{j\}$ of $\hat{v}_{j}(\mathcal{C})$, we can rewrite the expression $\sum_{(j,k) \in \mathcal{E}, k \in \mathcal{C}} w(j, k) - v_{j}(\{j\})$ as follows:
            \begin{equation*}
                \begin{split}
                    \sum_{(j,k) \in \mathcal{E}, k \in \mathcal{C}} w(j, k) - v_{j}(\{j\}) &= \sum_{(j,k) \in \mathcal{E}, k \in \mathcal{C}} \left\{ w(j, k) - z_{jk}(\{j\}) \right\} \\
                    &:= \sum_{(j,k) \in \mathcal{E}, k \in \mathcal{C}} \hat{w}(j, k),
                \end{split}
            \end{equation*}
            where
            \begin{equation*}
                \begin{split}
                    \hat{w}(j, k) := w(j, k) - z_{jk}(\{j\}), \\
                    v_{j}(\{j\}) = \sum_{(j,k) \in \mathcal{E}, k \in \mathcal{C}} z_{jk}(\{j\}).
                \end{split}
            \end{equation*}
            By the condition that $\hat{w}(j, k) \geq 0$, for $(j, k) \in \mathcal{E}$, from Theorem \ref{thm:grand_coalition_core_cag}, we can directly obtain the conditions to enable the grand coalition $\mathcal{N}$ being in the strict core such that
            \begin{equation}
                \begin{split}
                    w(j, k) \geq z_{jk}(\{j\}), \\
                    v_{j}(\{j\}) = \sum_{(j,k) \in \mathcal{E}, k \in \mathcal{N}} z_{jk}(\{j\}), \\
                    \forall (j, k) \in \mathcal{E}.
                \end{split}
            \end{equation}
        \end{proof}

        \begingroup
        \def\thetheorem{\ref{thm:dvsc_core_existence}}
            \begin{theorem}
                In an OSB-CAG, for any dynamic affinity graph $G_{t} = \langle \mathcal{N}_{t}, \mathcal{E}_{t} \rangle$ at any timestep $t$, if there exists a joint action $a_{t} \in \mathcal{A}_{\scriptscriptstyle \mathcal{N}_{t}}$, for any agent $j \in \mathcal{N}_{t}$, satisfying $R_{j}(a_{t} \vert s_{t}) \geq R_{j}(a_{t}^{j} \vert s_{t})$ for any $s_{t} \in \mathcal{S}$, then DVSC always exists.
            \end{theorem}
        \endgroup
        \begin{proof}
            To avoid losing the generality, we consider an arbitrary dynamic affinity graph $G_{t} = \langle \mathcal{N}_{t}, \mathcal{E}_{t} \rangle$ for a temporary team $\mathcal{N}_{t} \ \mathlarger{\mathlarger{\subseteq}} \ \mathcal{N}$ at an arbitrary timestep $t$. For any state $s_{t} \in \mathcal{S}$ and any joint action $a_{t} \in \mathcal{A}_{\scriptscriptstyle \mathcal{N}_{t}}$, the affinity weight $w_{jk}(a_{t}^{j}, a_{t}^{k} \vert s_{t})$ of any $(j, k) \in \mathcal{E}_{t}$ can be represented as a corresponding $w(j, k)$ such that $w(j, k) = w_{jk}(a_{t}^{j}, a_{t}^{k} \vert s_{t})$. Similarly, each agent $j$'s preference reward for the coalition including only itself $R_{j}(a_{t}^{j} \vert s_{t})$ can also be represented as a corresponding $v_{j}(\{j\})$ such that $v_{j}(\{j\}) = R_{j}(a_{t}^{j} \vert s_{t})$. Thereafter, we can apply Lemma \ref{lemm:grand_coalition_core_cag_generalised} to the situation here at a single timestep $t$. Substituting the above variables into Eq.~\eqref{eq:conditions_grand_coalition_core} in Lemma \ref{lemm:grand_coalition_core_cag_generalised}, it is not difficult to observe that if for any state $s_{t} \in \mathcal{S}$, there exists a joint action $a_{t} \in \mathcal{A}_{\scriptscriptstyle \mathcal{N}_{t}}$ such that $\sum_{(j,k) \in \mathcal{E}_{t}, k \in \mathcal{N}_{t}} w_{jk}(a_{t}^{j}, a_{t}^{k} \vert s_{t}) \geq R_{j}(a_{t}^{j} \vert s_{t})$, then there always exists a $R_{j}(a_{t}^{j} \vert s_{t}) = \sum_{(j,k) \in \mathcal{E}_{t}, k \in \mathcal{N}_{t}} \beta_{jk}(a_{t}^{j} \vert s_{t})$ satisfying the condition that $w_{jk}(a_{t}^{j}, a_{t}^{k} \vert s_{t}) \geq \beta_{jk}(a_{t}^{j} \vert s_{t})$, for all $(j, k) \in \mathcal{E}_{t}$, for any state $s_{t} \in \mathcal{S}$. Analogously, we can obtain the same results for all timesteps as above, which achieves the long-horizon objective as defined in the DVSC. Therefore, we can conclude that for any dynamic affinity graph $G_{t} = \langle \mathcal{N}_{t}, \mathcal{E}_{t} \rangle$ at any timestep $t$, if there exists a joint action $a_{t} \in \mathcal{A}_{\scriptscriptstyle \mathcal{N}_{t}}$, for any agent $j \in \mathcal{N}_{t}$, satisfying $R_{j}(a_{t} \vert s_{t}) = \sum_{(j,k) \in \mathcal{E}_{t}, k \in \mathcal{N}_{t}} w_{jk}(a_{t}^{j}, a_{t}^{k} \vert s_{t}) \geq R_{j}(a_{t}^{j} \vert s_{t})$ for any $s_{t} \in \mathcal{S}$, then the DVSC defined in Eq.~\eqref{eq:dsvc} always exists.
        \end{proof}

    \subsection{The Proof of Theorem \ref{thm:joint_q_representation}}
        \begin{lemma}
        \label{lemm:basic_agent_truncated_episode}
            Under Assumption \ref{assm:agent_leaves_env}, it is valid to have the expressions that $Q_{jk}^{\pi^{i}}(a_{t}^{j}, a_{t}^{k} \vert s_{t}) = \mathbb{E}_{\pi^{i}}[\sum_{\tau=t}^{\infty} \gamma^{\tau - t} \alpha_{jk}(a_{\tau}^{j}, a_{\tau}^{k} \vert s_{\tau})]$ and $Q_{j}^{\pi^{i}}(a_{t}^{j} \vert s_{t}) = \mathbb{E}_{\pi^{i}}[\sum_{\tau=t}^{\infty} \gamma^{\tau - t} R_{j}(a_{\tau}^{j} \vert s_{\tau})]$, with the learner $i$'s policy $\pi^{i}$.
        \end{lemma}
        \begin{proof}
            Suppose that agent $j$ or $k$ leaves the environment at timestep $T$, then we can have the expression that $Q_{jk}^{\pi^{i}}(a_{t}^{j}, a_{t}^{k} \vert s_{t}) = \mathbb{E}_{\pi^{i}}[\sum_{\tau=t}^{\infty} \gamma^{\tau - t} \alpha_{jk}(a_{\tau}^{j}, a_{\tau}^{k} \vert s_{\tau})]$ by the condition in Assumption \ref{assm:agent_leaves_env} that $\alpha_{jk}(a_{\tau}^{j}, a_{\tau}^{k} \vert s_{\tau}) = 0$ for $\tau \geq T$ if agent $j$ or $k$ leaves the environment at timestep $T$ as follows:
            \begin{equation*}
                \begin{split}
                    Q_{jk}^{\pi^{i}}(a_{t}^{j}, a_{t}^{k} \vert s_{t}) &= \mathbb{E}_{\pi^{i}}[\sum_{\tau=t}^{T} \gamma^{\tau - t} \alpha_{jk}(a_{\tau}^{j}, a_{\tau}^{k} \vert s_{\tau})] \\
                    &= \mathbb{E}_{\pi^{i}}[\sum_{\tau=t}^{T} \gamma^{\tau - t} \alpha_{jk}(a_{\tau}^{j}, a_{\tau}^{k} \vert s_{\tau}) + \underbrace{\sum_{\tau=T}^{\infty} \gamma^{\tau - t} \alpha_{jk}(a_{\tau}^{j}, a_{\tau}^{k} \vert s_{\tau})}_{\qquad\qquad\quad =0 \text{ by Assumption \ref{assm:agent_leaves_env}.}}] \\
                    &= \mathbb{E}_{\pi^{i}}[\sum_{\tau=t}^{\infty} \gamma^{\tau - t} \alpha_{jk}(a_{\tau}^{j}, a_{\tau}^{k} \vert s_{\tau})].
                \end{split}
            \end{equation*}
            Similarly, by the condition in Assumption \ref{assm:agent_leaves_env} that $R_{j}(a_{\tau}^{j} \vert s_{\tau})=0$ for $\tau \geq T'$ if agent $j$ leaves the environment at timestep $T'$, we can derive the result that $Q_{j}^{\pi^{i}}(a_{t}^{j} \vert s_{t}) = \mathbb{E}_{\pi^{i}}[\sum_{\tau=t}^{\infty} \gamma^{\tau-t} R_{j}(a_{\tau}^{j} \vert s_{\tau})]$.
        \end{proof}

        \begingroup
        \def\thetheorem{\ref{thm:joint_q_representation}}
        \begin{theorem}
            Under Assumption \ref{assm:agent_leaves_env}, if $w_{jk}(a_{\tau}^{j}, a_{\tau}^{k} \vert s_{\tau}) = \alpha_{jk}(a_{\tau}^{j}, a_{\tau}^{k} \vert s_{\tau}) + \beta_{jk}(a_{\tau}^{j} \vert s_{\tau})$, then the joint Q-value of the learner's policy $\pi^{i}$ can be expressed as follows:
            \begin{equation*}
                \begin{split}
                    Q^{\pi^{i}}(s_{t}, a_{t}) &= \sum_{(j,k) \in \mathcal{E}_{t}} Q_{jk}^{\pi^{i}}(a_{t}^{j}, a_{t}^{k} \vert s_{t}) + \sum_{j \in \mathcal{N}_{t}} Q_{j}^{\pi^{i}}(a_{t}^{j} \vert s_{t}) \\
                    = \sum_{j \in \mathcal{N}_{t}} &\Big\{ \sum_{(j,k) \in \mathcal{E}_{t}} Q_{jk}^{\pi^{i}}(a_{t}^{j}, a_{t}^{k} \vert s_{t}) + Q_{j}^{\pi^{i}}(a_{t}^{j} \vert s_{t}) \Big\} \\
                    := \sum_{j \in \mathcal{N}_{t}} &Q_{j}^{\pi^{i}}(a_{t} \vert s_{t}),
                \end{split}
            \end{equation*}
            where $Q_{jk}^{\pi^{i}}(a_{t}^{j}, a_{t}^{k} \vert s_{t}) = \mathbb{E}_{\pi^{i}}[\sum_{\tau=t}^{\infty} \gamma^{\tau - t} \alpha_{jk}(a_{\tau}^{j}, a_{\tau}^{k} \vert s_{\tau})]$ and $Q_{j}^{\pi^{i}}(a_{t}^{j} \vert s_{t}) = \mathbb{E}_{\pi^{i}}[\sum_{\tau=t}^{\infty} \gamma^{\tau - t} R_{j}(a_{\tau}^{j} \vert s_{\tau})]$.
        \end{theorem}
        \endgroup
        \begin{proof}
            By Assumption \ref{assm:agent_leaves_env} and the result of Lemma \ref{lemm:basic_agent_truncated_episode}, for any state $s_{t} \in \mathcal{S}$ and any joint action $a_{t} \in \mathcal{A}_{\scriptscriptstyle \mathcal{N}_{t}}$, we can represent the joint Q-value under any learner $i$'s policy $\pi^{i}$ such as $Q^{\pi^{i}}(s_{t}, a_{t})$ as follows:
            \begin{equation}
            \label{eq:joint_team_preference_Q}
                \begin{split}
                    Q^{\pi^{i}}(s_{t}, a_{t}) &= \mathbb{E}_{\pi^{i}} \Big[ \sum_{\tau=t}^{\infty} \gamma^{\tau-t} R(s_{\tau}, a_{\tau}) \Big] \\
                    &= \mathbb{E}_{\pi^{i}} \Big[ \sum_{\tau=t}^{\infty} \gamma^{\tau-t} \sum_{j \in \mathcal{N}_{\tau}} R_{j}(a_{\tau} \vert s_{\tau}) \Big] \\
                    &= \mathbb{E}_{\pi^{i}} \Big[ \sum_{\tau=t}^{\infty} \gamma^{\tau-t} \sum_{j \in \mathcal{N}_{\tau}} \Big( \sum_{(j, k) \in \mathcal{E}_{\tau}} \alpha_{jk}(a_{\tau}^{j}, a_{\tau}^{k} \vert s_{\tau}) + R_{j}(a_{\tau}^{j} \vert s_{\tau}) \Big) \Big] \\
                    &= \mathbb{E}_{\pi^{i}} \bigg[ \sum_{\tau=t}^{\infty} \gamma^{\tau-t} \bigg( \sum_{j \in \mathcal{N}_{\tau}} \Big( \sum_{(j, k) \in \mathcal{E}_{\tau}} \alpha_{jk}(a_{\tau}^{j}, a_{\tau}^{k} \vert s_{\tau}) + R_{j}(a_{\tau}^{j} \vert s_{\tau}) \Big) \\
                    &\qquad\qquad\qquad\qquad + \underbrace{\sum_{j \in \mathcal{N}_{t} \backslash \mathcal{N}_{\tau}} \Big( \sum_{(j, k) \in \mathcal{E}_{t} \backslash \mathcal{E}_{\tau}} \alpha_{jk}(a_{\tau}^{j}, a_{\tau}^{k} \vert s_{\tau}) + R_{j}(a_{\tau}^{j} \vert s_{\tau}) \Big)}_{\qquad\qquad\quad =0 \text{ by Assumption \ref{assm:agent_leaves_env}.}} \bigg) \bigg] \\
                    &= \mathbb{E}_{\pi^{i}} \Big[ \sum_{\tau=t}^{\infty} \gamma^{\tau-t} \sum_{j \in \mathcal{N}_{t}} \Big( \sum_{(j, k) \in \mathcal{E}_{t}} \alpha_{jk}(a_{\tau}^{j}, a_{\tau}^{k} \vert s_{\tau}) + R_{j}(a_{\tau}^{j} \vert s_{\tau}) \Big) \Big] \\
                    &= \sum_{j \in \mathcal{N}_{t}} \bigg\{ \sum_{(j, k) \in \mathcal{E}_{t}} \underbrace{\mathbb{E}_{\pi^{i}}\Big[ \sum_{\tau=t}^{\infty} \gamma^{\tau-t} \alpha_{jk}(a_{\tau}^{j}, a_{\tau}^{k} \vert s_{\tau})}_{= Q_{jk}^{\pi^{i}}(a_{t}^{j}, a_{t}^{k} \vert s_{t}) \text{ by Lemma \ref{lemm:basic_agent_truncated_episode}.}} \Big] + \underbrace{\mathbb{E}_{\pi^{i}}\Big[ \sum_{\tau=t}^{\infty} \gamma^{\tau-t} R_{j}(a_{\tau}^{j} \vert s_{\tau}) \Big]}_{= Q_{j}^{\pi^{i}}(a_{t}^{j} \vert s_{t}) \text{ by Lemma \ref{lemm:basic_agent_truncated_episode}.}} \bigg\} \\
                    &= \sum_{j \in \mathcal{N}_{t}} \bigg\{ \sum_{(j, k) \in \mathcal{E}_{t}}  Q_{jk}^{\pi^{i}}(a_{t}^{j}, a_{t}^{k} \vert s_{t}) + Q_{j}^{\pi^{i}}(a_{t}^{j} \vert s_{t}) \bigg\} \\
                    &= \sum_{j \in \mathcal{N}_{t}} \sum_{(j, k) \in \mathcal{E}_{t}}  Q_{jk}^{\pi^{i}}(a_{t}^{j}, a_{t}^{k} \vert s_{t}) + \sum_{j \in \mathcal{N}_{t}} Q_{j}^{\pi^{i}}(a_{t}^{j} \vert s_{t}) \\
                    &= \sum_{(j, k) \in \mathcal{E}_{t}}  Q_{jk}^{\pi^{i}}(a_{t}^{j}, a_{t}^{k} \vert s_{t}) + \sum_{j \in \mathcal{N}_{t}} Q_{j}^{\pi^{i}}(a_{t}^{j} \vert s_{t}).
                \end{split}
            \end{equation}

            By the fashion of Bellman optimality equation, for any state $s_{t} \in \mathcal{S}$ and any joint action $a_{t} \in \mathcal{A}_{\scriptscriptstyle \mathcal{N}_{t}}$, we can write out each agent $j$'s preference Q-value under the learner $i$'s policy $\pi^{i}$, $Q_{j}^{\pi^{i}}(a_{t} \vert s_{t})$, as follows: 
            \begin{equation}
            \label{eq:single_agent_team_preference_Q}
                \begin{split}
                    Q_{j}^{\pi^{i}}(s_{t}, a_{t}) &= \mathbb{E}_{\pi^{i}} \Big[ \sum_{\tau=t}^{\infty} \gamma^{\tau-t} R_{j}(a_{\tau} \vert s_{\tau}) \Big] \\
                    &= \mathbb{E}_{\pi^{i}} \Big[ \sum_{\tau=t}^{\infty} \gamma^{\tau-t} \Big( \sum_{(j, k) \in \mathcal{E}_{\tau}} \alpha_{jk}(a_{\tau}^{j}, a_{\tau}^{k} \vert s_{\tau}) + R_{j}(a_{\tau}^{j} \vert s_{\tau}) \Big) \Big] \\
                    &= \mathbb{E}_{\pi^{i}} \Big[ \sum_{\tau=t}^{\infty} \gamma^{\tau-t} \Big( \sum_{(j, k) \in \mathcal{E}_{\tau}} \alpha_{jk}(a_{\tau}^{j}, a_{\tau}^{k} \vert s_{\tau}) \Big) \Big] + \mathbb{E}_{\pi^{i}} \Big[ \sum_{\tau=t}^{\infty} \gamma^{\tau-t} R_{j}(a_{\tau}^{j} \vert s_{\tau}) \Big] \\
                    &= \mathbb{E}_{\pi^{i}} \Big[ \sum_{\tau=t}^{\infty} \gamma^{\tau-t} \Big( \sum_{(j, k) \in \mathcal{E}_{\tau}} \alpha_{jk}(a_{\tau}^{j}, a_{\tau}^{k} \vert s_{\tau}) + \underbrace{\sum_{(j, k) \in \mathcal{E}_{t} \backslash \mathcal{E}_{\tau}} \alpha_{jk}(a_{\tau}^{j}, a_{\tau}^{k} \vert s_{\tau})}_{\qquad\qquad\quad =0 \text{ by Assumption \ref{assm:agent_leaves_env}.}} \Big) \Big] + \mathbb{E}_{\pi^{i}} \Big[ \sum_{\tau=t}^{\infty} \gamma^{\tau-t} R_{j}(a_{\tau}^{j} \vert s_{\tau}) \Big] \\
                    &= \mathbb{E}_{\pi^{i}} \Big[ \sum_{\tau=t}^{\infty} \gamma^{\tau-t} \Big( \sum_{(j, k) \in \mathcal{E}_{t}} \alpha_{jk}(a_{\tau}^{j}, a_{\tau}^{k} \vert s_{\tau}) \Big) \Big] + \mathbb{E}_{\pi^{i}} \Big[ \sum_{\tau=t}^{\infty} \gamma^{\tau-t} R_{j}(a_{\tau}^{j} \vert s_{\tau}) \Big] \\
                    &= \sum_{(j, k) \in \mathcal{E}_{t}} \underbrace{\mathbb{E}_{\pi^{i}} \Big[ \sum_{\tau=t}^{\infty} \gamma^{\tau-t} \alpha_{jk}(a_{\tau}^{j}, a_{\tau}^{k} \vert s_{\tau}) \Big]}_{= Q_{jk}^{\pi^{i}}(a_{t}^{j}, a_{t}^{k} \vert s_{t}) \text{ by Lemma \ref{lemm:basic_agent_truncated_episode}.}} + \underbrace{\mathbb{E}_{\pi^{i}} \Big[ \sum_{\tau=t}^{\infty} \gamma^{\tau-t} R_{j}(a_{\tau}^{j} \vert s_{\tau}) \Big]}_{= Q_{j}^{\pi^{i}}(a_{t}^{j} \vert s_{t}) \text{ by Lemma \ref{lemm:basic_agent_truncated_episode}.}} \\
                    &= \sum_{(j, k) \in \mathcal{E}_{t}}  Q_{jk}^{\pi^{i}}(a_{t}^{j}, a_{t}^{k} \vert s_{t}) + Q_{j}^{\pi^{i}}(a_{t}^{j} \vert s_{t}).
                \end{split}
            \end{equation}
            By substituting the expression of $Q_{j}^{\pi^{i}}(s_{t}, a_{t})$ derived in Eq.~\eqref{eq:single_agent_team_preference_Q} into Eq.~\eqref{eq:joint_team_preference_Q}, we can get the following relation:
            \begin{equation}
                Q^{\pi^{i}}(s_{t}, a_{t}) = \sum_{j \in \mathcal{N}_{t}} Q_{j}^{\pi^{i}}(a_{t} \vert s_{t}).
            \end{equation}
        \end{proof}
        
    \subsection{The Proof of the Conditions of Symmetry for Various Dynamic Affinity Graphs}
    \label{subsec:the_proof_of_construction_of_variant_dynamic_affinity_graphs}
        \begingroup
        \def\theproposition{\ref{prop:condition_of_symmetry_star_graph}}
            \begin{proposition}
                For the learner $i$ and any teammate $j$ or $k$, the constraints $R_{i}(a_{t}^{i} \vert s_{t}) = \sum_{j \in -i} R_{j}(a_{t}^{j} \vert s_{t})$ and $\alpha_{jk}(a_{t}^{j}, a_{t}^{k} \vert s_{t}) = \alpha_{kj}(a_{t}^{k}, a_{t}^{j} \vert s_{t})$, for any $a_{t} \in \mathcal{A}_{\scriptscriptstyle \mathcal{N}_{t}}$ and $s_{t} \in \mathcal{S}$, are necessary for a star dynamic affinity graph to be symmetric.
            \end{proposition}
        \endgroup
        \begin{proof}
            Recall that a symmetric dynamic affinity graph $G_{t} = \langle \mathcal{N}_{t}, \mathcal{E}_{t} \rangle$ needs to satisfy the following condition that $w_{jk}(a_{t}^{j}, a_{t}^{k} \vert s_{t}) = w_{kj}(a_{t}^{k}, a_{t}^{j} \vert s_{t})$, for all $(j, k) \in \mathcal{E}_{t}$, for any state $s_{t} \in \mathcal{S}$ and any joint action $a_{t} \in \mathcal{A}_{\scriptscriptstyle \mathcal{N}_{t}}$. In the dynamic affinity graph as a star graph, the affinity weights of any $(i, j) \in \mathcal{E}_{t}$ or $(j, i) \in \mathcal{E}_{t}$ can be represented as follows:
            \begin{equation*}
                \begin{split}
                    w_{ij}(a_{t}^{i}, a_{t}^{j} \vert s_{t}) = \alpha_{ij}(a_{t}^{i}, a_{t}^{j} \vert s_{t}) + \beta_{ij}(a_{t}^{i} \vert s_{t}), \text{where }R_{i}(a_{t}^{i} \vert s_{t}) = \sum_{j \in -i} \beta_{ij}(a_{t}^{i} \vert s_{t}),\\
                    w_{ji}(a_{t}^{j}, a_{t}^{i} \vert s_{t}) = \alpha_{ji}(a_{t}^{j}, a_{t}^{i} \vert s_{t}) + \beta_{ji}(a_{t}^{j} \vert s_{t}), \text{where }R_{j}(a_{t}^{j} \vert s_{t}) = \beta_{ji}(a_{t}^{j} \vert s_{t}).
                \end{split}
            \end{equation*}
            It is not difficult to observe that for all $\forall s_{t} \in \mathcal{S}$ and $a_{t} \in \mathcal{A}_{\scriptscriptstyle \mathcal{N}_{t}}$ the following conditions that
            \begin{equation*}
                \begin{split}
                    \alpha_{ij}(a_{t}^{i}, a_{t}^{j} \vert s_{t}) = \alpha_{ji}(a_{t}^{j}, a_{t}^{i} \vert s_{t}), \\
                    R_{i}(a_{t}^{i} \vert s_{t}) = \sum_{j \in -i} R_{j}(a_{t}^{j} \vert s_{t}),
                \end{split}
            \end{equation*}
            are necessary for that the star dynamic affinity graph is symmetric. In more details, that $R_{i}(a_{t}^{i} \vert s_{t}) = \sum_{j \in -i} R_{j}(a_{t}^{j} \vert s_{t})$ is a necessary condition for the existence of the one-to-one correspondence that $\beta_{ij}(a_{t}^{i} \vert s_{t}) = \beta_{ji}(a_{t}^{j} \vert s_{t}) = R_{j}(a_{t}^{j} \vert s_{t})$.
        \end{proof}
        
        \begingroup
        \def\theproposition{\ref{prop:condition_of_symmetry_full_graph}}
            \begin{proposition}
                For any two agents $j$ or $k$, the constraints $R_{j}(a_{t}^{j} \vert s_{t}) = R_{k}(a_{t}^{k} \vert s_{t})$ and $\alpha_{jk}(a_{t}^{j}, a_{t}^{k} \vert s_{t}) = \alpha_{kj}(a_{t}^{k}, a_{t}^{j} \vert s_{t})$, for any $a_{t} \in \mathcal{A}_{\scriptscriptstyle \mathcal{N}_{t}}$ and $s_{t} \in \mathcal{S}$, are necessary for the complete dynamic affinity graph to be symmetric.
            \end{proposition}
        \endgroup
        \begin{proof}
            Recall that a symmetric dynamic affinity graph $G_{t} = \langle \mathcal{N}_{t}, \mathcal{E}_{t} \rangle$ needs to satisfy the following condition that $w_{jk}(a_{t}^{j}, a_{t}^{k} \vert s_{t}) = w_{kj}(a_{t}^{k}, a_{t}^{j} \vert s_{t})$, for all $(j, k) \in \mathcal{E}_{t}$, for any state $s_{t} \in \mathcal{S}$ and any joint action $a_{t} \in \mathcal{A}_{\scriptscriptstyle \mathcal{N}_{t}}$. In the dynamic affinity graph as a complete graph, the affinity weights of any $(j, k) \in \mathcal{E}_{t}$ can be represented as follows:
            \begin{equation*}
                w_{jk}(a_{t}^{j}, a_{t}^{k} \vert s_{t}) = \alpha_{jk}(a_{t}^{j}, a_{t}^{k} \vert s_{t}) + \beta_{jk}(a_{t}^{j} \vert s_{t}), \text{where }R_{j}(a_{t}^{j} \vert s_{t}) = \sum_{k \in -j} \beta_{jk}(a_{t}^{j} \vert s_{t}).
            \end{equation*}
            It is not difficult to observe that for all $\forall s_{t} \in \mathcal{S}$ and $a_{t} \in \mathcal{A}_{\scriptscriptstyle \mathcal{N}_{t}}$ the following conditions that
            \begin{equation*}
                \begin{split}
                    \alpha_{jk}(a_{t}^{j}, a_{t}^{k} \vert s_{t}) = \alpha_{kj}(a_{t}^{k}, a_{t}^{j} \vert s_{t}), \\
                    R_{j}(a_{t}^{j} \vert s_{t}) = R_{k}(a_{t}^{k} \vert s_{t}),
                \end{split}
            \end{equation*}
            are necessary for that the complete dynamic affinity graph is symmetric. In more details, that $R_{j}(a_{t}^{j} \vert s_{t}) = \sum_{k \in -j} \beta_{jk}(a_{t}^{j} \vert s_{t}) = \sum_{j \in -k} \beta_{kj}(a_{t}^{k} \vert s_{t}) = R_{k}(a_{t}^{k} \vert s_{t})$ is a necessary condition for the existence of the one-to-one correspondence that $\beta_{jk}(a_{t}^{j} \vert s_{t}) = \beta_{kj}(a_{t}^{k} \vert s_{t})$.
        \end{proof}
        
\subsection{The Proof of Theorem \ref{thm:open_team_ad_hoc_bellman_optimality}}
\label{subsec:proof_of_theorem_thm:open_team_ad_hoc_bellman_optimality}
        \begingroup
        \def\thetheorem{\ref{thm:open_team_ad_hoc_bellman_optimality}}
            \begin{theorem}
                Under Assumption \ref{assm:agent_leaves_env} and an arbitrary learner's deterministic stationary policy $\pi^{i}$, the Bellman equation for the OSB-CAG with DVSC as a solution concept is expressed as follows: $Q^{\pi^{i}}(s_{t}, a_{t}) = R(s_{t}, a_{t}) + \gamma \mathbb{E}_{{\scriptscriptstyle \mathcal{N}_{t+1}}, s_{t+1} \sim P_{O}} \Big[ 
\mathbb{E}_{\substack{\theta_{t+1} \sim P_{E}, \ a_{t+1} \sim \pi_{t+1}}} \big[ Q^{\pi^{i}}(s_{t+1}, a_{t+1}) \big] \Big]$.
            \end{theorem}
        \endgroup
        \begin{proof}
            We derive Eq.~\eqref{eq:open_team_ad_hoc_bellman_optimality} as follows.
            
            By the result of Theorem \ref{thm:joint_q_representation}, we can represent the joint Q-value under an arbitrary learner's deterministic stationary policy $\pi^{i}$ referred to as $Q^{\pi^{i}}(s_{t}, a_{t})$ as follows:
            \begin{equation}
            \label{eq:condition0}
                Q^{\pi^{i}}(s_{t}, a_{t}) = \sum_{j \in \mathcal{N}_{t}} Q^{\pi^{i}}_{j}(a_{t} \vert s_{t}),
            \end{equation}
            Next, we can expand the preference Q-value of each agent $j \in \mathcal{N}_{t}$ following the fashion of the Bellman equation such that
            \begin{equation}
            \label{eq:single_agent_optimality_equation}
                Q^{\pi^{i}}_{j}(a_{t} \vert s_{t}) = R_{j}(a_{t} \vert s_{t}) + \gamma \mathbb{E}_{{\scriptscriptstyle \mathcal{N}_{t+1}}, s_{t+1} \sim P_{O}} \Big[ \mathbb{E}_{\substack{\theta_{t+1} \sim P_{E},\\ a_{t+1} \sim \pi_{t+1}}} \big[ Q^{\pi^{i}}_{j}(a_{t+1} \vert s_{t+1}) \big] \Big].
            \end{equation}
            Then, we can sum up Eq.~\eqref{eq:single_agent_optimality_equation} for all possible agents belonging to the temporary team $\mathcal{N}_{t}$ and get an equation to evaluate the influence of the learner's policy $\pi^{i}$ to a temporary team $\mathcal{N}_{t}$ such that
            \begin{equation}
            \label{eq:sum_agent_optimality_equation}
                \begin{split}
                    Q^{\pi^{i}}(s_{t}, a_{t}) &= \sum_{j \in \mathcal{N}_{t}} Q^{\pi^{i}}_{j}(a_{t} \vert s_{t}) \\
                    &= \sum_{j \in \mathcal{N}_{t}} R_{j}(a_{t} \vert s_{t}) + \sum_{j \in \mathcal{N}_{t}} \gamma \mathbb{E}_{{\scriptscriptstyle \mathcal{N}_{t+1}}, s_{t+1} \sim P_{O}} \Big[ \mathbb{E}_{\substack{\theta_{t+1} \sim P_{E},\\ a_{t+1} \sim \pi_{t+1}}} \big[ Q^{\pi^{i}}_{j}(a_{t+1} \vert s_{t+1}) \big] \Big] \\
                    &= R(s_{t}, a_{t}) + \underbrace{\sum_{j \in \mathcal{N}_{t+1}} \gamma \mathbb{E}_{{\scriptscriptstyle \mathcal{N}_{t+1}}, s_{t+1} \sim P_{O}} \Big[ \mathbb{E}_{\substack{\theta_{t+1} \sim P_{E},\\ a_{t+1} \sim \pi_{t+1}}} \big[ Q^{\pi^{i}}_{j}(a_{t+1} \vert s_{t+1}) \big] \Big]}_{\text{Since $Q^{\pi^{i}}_{j}(a_{t+1} \vert s_{t+1}) = 0$ for agent $j \in \mathcal{N}_{t} \backslash \mathcal{N}_{t+1}$ by Assumption \ref{assm:agent_leaves_env}.}} \\
                    &= R(s_{t}, a_{t}) + \gamma \mathbb{E}_{{\scriptscriptstyle \mathcal{N}_{t+1}}, s_{t+1} \sim P_{O}} \Big[ \mathbb{E}_{\substack{\theta_{t+1} \sim P_{E},\\ a_{t+1} \sim \pi_{t+1}}} \big[ \sum_{j \in \mathcal{N}_{t+1}} Q^{\pi^{i}}_{j}(a_{t+1} \vert s_{t+1}) \big] \Big] \\
                    &= R(s_{t}, a_{t}) + \gamma \mathbb{E}_{{\scriptscriptstyle \mathcal{N}_{t+1}}, s_{t+1} \sim P_{O}} \Big[ 
 \mathbb{E}_{\substack{\theta_{t+1} \sim P_{E},\\ a_{t+1} \sim \pi_{t+1}}} \big[ Q^{\pi^{i}}(s_{t+1}, a_{t+1}) \big] \Big].
                \end{split}
            \end{equation}
            Note that Eq.~\eqref{eq:sum_agent_optimality_equation} does not hold if $\mathcal{N}_{t} \ \mathlarger{\mathlarger{\subset}} \  \mathcal{N}_{t+1}$, since it is problematic to expand the preference Q-value of an agent $k \in \mathcal{N}_{t+1}$ but $\notin \mathcal{N}_{t}$ at timestep $t$, which can be seen as a singularity of this equation. More specifically, $0 = Q^{\pi^{i}}_{k}(a_{t} \vert s_{t}) = R_{k}(a_{t} \vert s_{t}) + \gamma \mathbb{E}_{{\scriptscriptstyle \mathcal{N}_{t+1}}, s_{t+1} \sim P_{O}} \Big[ \mathbb{E}_{\substack{\theta_{t+1} \sim P_{E},\\ a_{t+1} \sim \pi_{t+1}}} \big[ Q^{\pi^{i}}_{k}(a_{t+1} \vert s_{t+1}) \big] \Big] > 0$ is impossible, given that at least $R_{k}(a_{t'} \vert s_{t'}) > 0$, implying agent $k$'s preference for collaborating with other agents, at a timestep $t' \geq t$.
        \end{proof}
        
\section{Experimental Settings}
\label{sec:experimental_settings}
    We evaluate our proposed CIAO in two existing environments, LBF and Wolfpack, both configured with open team settings \citep{rahman2021towards}. In these settings, teammates are randomly selected to enter the environment and remain for a specified number of timesteps. If a teammate surpasses its allocated lifetime, it is removed from the environment and placed in a re-entry queue with a randomly assigned waiting time. The randomized re-entry queue results in varied compositions of teammates in a temporary team. When the number of agents in the environment does not reach its maximum, agents in the re-entry queue are introduced to the environment. Specifically, in the Wolfpack environment, we uniformly determine the active duration by selecting a value between 25 and 35 timesteps, while the dead duration is uniformly sampled between 15 and 25 timesteps. Conversely, the durations for LBF are somewhat shorter, with the active duration uniformly sampled between 15 and 25 timesteps, and the dead duration between 10 and 20 timesteps.
    
    The teammate policies adhere to the experimental settings used for testing GPL~\citep{rahman2021towards}, which encompass a range of heuristic policies and pre-trained policies. Specifically, for Wolfpack, the teammate set includes the following agents: random agent, greedy agent, greedy probabilistic agent, teammate-aware agents, GNN-Based teammate-aware agents, graph DQN agents, greedy waiting agents, greedy probabilistic waiting agents, and greedy team-aware waiting agents. In the case of LBF, a combination of heuristics and A2C agents is employed as the teammate policy set. For more detailed information about teammate policies, we recommend referring to Appendix B.4 of GPL's paper.
    
    In our investigation of different agent-type sets within LBF experiments (see Appendix \ref{subsec:sensitivity_to_different_teammates}), we deliberately exclude the A2C agent from the original agent-type set, thereby establishing a distinct agent-type subset. It's crucial to acknowledge that the A2C agent provided by GPL is designed for scenarios with a maximum of 5 agents. Tailored to scenarios involving a greater number of agents, specifically up to 9, we undertake the additional step of training an A2C agent tailored to these expanded requirements.

    In our experiments of studying the generalizability of CIAO, we constructed the agent-type sets for training and testing, respectively, for Wolfpack and LBF. The details are shown in Tab.~\ref{tab:variant_agent-type_sets}.

    \begin{table}[htbp]
      \centering
      \caption{Variant agent-type sets for training and testing in experiments for evaluating generalizability of CIAO. The shorthand ``Int'' stands for the scenario where agent-type sets for training have intersection with testing. The shrothand ``Exc'' stands for the scenarios where agent-type sets for training are mutually exclusive to testing.}
      \label{tab:variant_agent-type_sets}
      \scalebox{0.95}{
      \begin{tabular}{ccc}
        \toprule
        \textbf{Scenario Name} & \textbf{Training} & \textbf{Testing} \\
        \midrule
        Wolfpack-Int & \texttt{\shortstack{GreedyPredatorAgent, \\ GreedyProbabilisticAgent, \\ TeammateAwarePredator, \\ DistilledCoopAStarAgent, \\ GraphDQNAgent}} & \texttt{\shortstack{GraphDQNAgent, \\ RandomAgent, \\ GreedyWaitingAgent, \\ 
        GreedyProbabilisticWaitingAgent, \\
                          TeammateAwareWaitingAgent}} \\
        \midrule
        Wolfpack-Exc & \texttt{\shortstack{GreedyPredatorAgent, \\ GreedyProbabilisticAgent, \\
                          TeammateAwarePredator, \\ DistilledCoopAStarAgent, \\ GraphDQNAgent}} & \texttt{\shortstack{RandomAgent, \\ GreedyWaitingAgent, \\ GreedyProbabilisticWaitingAgent, \\
                          TeammateAwareWaitingAgent}} \\
        \midrule
        LBF-Int & \texttt{H8, H7, H6, H5, A2C0} & \texttt{A2C0, H1, H2, H3, H4} \\
        \midrule
        LBF-Exc & \texttt{H8, H7, H6, H5, A2C0} & \texttt{H1, H2, H3, H4} \\
        \bottomrule
      \end{tabular}
      }
    \end{table}
    
    \subsection{Detailed Hyperparameters and Computing Resources}
    \label{subsec:detailed_hyperparameters}
        We summarize the values of the common hyperparameters of algorithms that are used in our experiments, as shown in Tabs.~\ref{tab:lbf} and \ref{tab:wolfpack}. The optimizer we use during training is Adam \citep{kingma2014adam}, with the default hyperparameters except learning rate. All algorithms in experiments are implemented in PyTorch \citep{paszke2019pytorch}.
        \begin{table}[htbp]
          \centering
          \caption{Shared hyperparameters for LBF. Note that the arguments \texttt{intersection\_generalization}, \texttt{exclusion\_generalization} and \texttt{exclude\_A2Cagent} cannot be simultaneously set to be \texttt{True}.}
          \label{tab:lbf}
          \scalebox{0.95}{
          \begin{tabular}{cc}
            \toprule
            \textbf{Hyperparameter} & \textbf{Value} \\
            \midrule
            lr & \texttt{0.00025} \\
            gamma & \texttt{0.99} \\
            max\_num\_steps & \texttt{1000000} \\
            eps\_length & \texttt{200} \\
            update\_frequency & \texttt{4} \\
            saving\_frequency & \texttt{50} \\
            pair\_comp & bmm \\
            num\_envs & \texttt{16} \\
            tau & \texttt{0.001} \\
            eval\_eps & \texttt{5} \\
            weight\_predict & \texttt{1.0} \\
            num\_players\_train & \texttt{3} \\
            \multirow{2}{*}{num\_players\_test} & \texttt{5} for a maximum of 5 agents \\
            & \texttt{9} for a maximum of 9 agents \\
            \multirow{2}{*}{exclude\_A2Cagent} & \texttt{True} for the agent-type set excluding A2C agent \\
            & \texttt{False} for the default agent-type sets \\
            \multirow{2}{*}{intersection\_generalization} & \texttt{True} for the agent-type sets for training and testing are intersected \\
            & \texttt{False} for the default agent-type sets \\
            \multirow{2}{*}{exclusion\_generalization} & \texttt{True} for the agent-type sets for training and testing are mutually exclusive \\
            & \texttt{False} for the default agent-type sets \\
            seed & \texttt{0} \\
            eval\_init\_seed & \texttt{2500} \\
            \bottomrule
          \end{tabular}
          }
        \end{table}
    
        \begin{table}[ht!]
          \centering
          \caption{Shared hyperparameters for Wolfpack. Note that the arguments \texttt{intersection\_generalization} and \texttt{exclusion\_generalization} cannot be simultaneously set to be \texttt{True}.}
          \label{tab:wolfpack}
          \scalebox{0.95}{
          \begin{tabular}{cc}
            \toprule
            \textbf{Hyperparameter} & \textbf{Value} \\
            \midrule
            lr & \texttt{0.00025} \\
            gamma & \texttt{0.99} \\
            num\_episodes & \texttt{4000} \\
            update\_frequency & \texttt{4} \\
            saving\_frequency & \texttt{50} \\
            pair\_comp & bmm \\
            num\_envs & \texttt{16} \\
            tau & \texttt{0.001} \\
            eval\_eps & \texttt{5} \\
            weight\_predict & \texttt{1.0} \\
            num\_players\_train & \texttt{3} \\
            \multirow{2}{*}{num\_players\_test} & \texttt{5} for a maximum of 5 agents \\
            & \texttt{9} for a maximum of 9 agents \\
            \multirow{2}{*}{intersection\_generalization} & \texttt{True} for the agent-type sets for training and testing are intersected \\
            & \texttt{False} for the default agent-type sets \\
            \multirow{2}{*}{exclusion\_generalization} & \texttt{True} for the agent-type sets for training and testing are mutually exclusive \\
            & \texttt{False} for the default agent-type sets \\
            seed & \texttt{0} \\
            eval\_init\_seed & \texttt{2500} \\
            close\_penalty & \texttt{0.5} \\
            \bottomrule
          \end{tabular}
          }
        \end{table}
    
        Then, we list the exclusive hyperparameters of all algorithms implemented in this work, as shown in Tab.~\ref{tab:exclusive_args}.
        \begin{table}[ht!]
          \centering
          \caption{Exclusive hyperparameters of all algorithms implemented in this paper.}
          \label{tab:exclusive_args}
          \scalebox{0.95}{
          \begin{tabular}{ccccc}
            \toprule
            \textbf{Algorithm} & \textbf{weight\_regularizer} & \textbf{graph} & \textbf{pair\_range} & \textbf{indiv\_range} \\
            \midrule
            GPL & \texttt{0.0} & \texttt{complete} & \texttt{free} & \texttt{free} \\
            CIAO-S & \texttt{0.5} & \texttt{star} & \texttt{pos} & \texttt{pos} \\
            CIAO-S-NP & \texttt{0.5} & \texttt{star} & \texttt{neg} & \texttt{pos} \\
            CIAO-S-FI & \texttt{0.5} & \texttt{star} & \texttt{pos} & \texttt{free} \\
            CIAO-S-ZI & \texttt{0.5} & \texttt{star} & \texttt{pos} & \texttt{zero} \\
            CIAO-S-NI & \texttt{0.5} & \texttt{star} & \texttt{pos} & \texttt{neg} \\
            CIAO-C & \texttt{0.5} & \texttt{complete} & \texttt{pos} & \texttt{pos} \\
            CIAO-C-NP & \texttt{0.5} & \texttt{complete} & \texttt{neg} & \texttt{pos} \\
            CIAO-C-FI & \texttt{0.5} & \texttt{complete} & \texttt{pos} & \texttt{free} \\
            CIAO-C-ZI & \texttt{0.5} & \texttt{complete} & \texttt{pos} & \texttt{zero} \\
            CIAO-C-NI & \texttt{0.5} & \texttt{complete} & \texttt{pos} & \texttt{neg} \\
            \bottomrule
          \end{tabular}
          }
        \end{table}
        All experiments have been run on Xeon Gold 6230 with 30 CPU cores and 30 GB primary memory. An experiment conducted on Wolfpack requires approximately 11 hours, whereas an experiment on LBF typically takes around 12 hours.
    
    \section{Additional Experimental Results}
    \label{sec:additional_experimental_results}
        \subsection{Additional Evaluation on Small Number of Agents}
        \label{subsec:small_number_of_agents}
            \begin{figure}[ht!]
                \centering
                \begin{subfigure}[b]{0.33\textwidth}
                    \centering
                    \includegraphics[width=\textwidth]{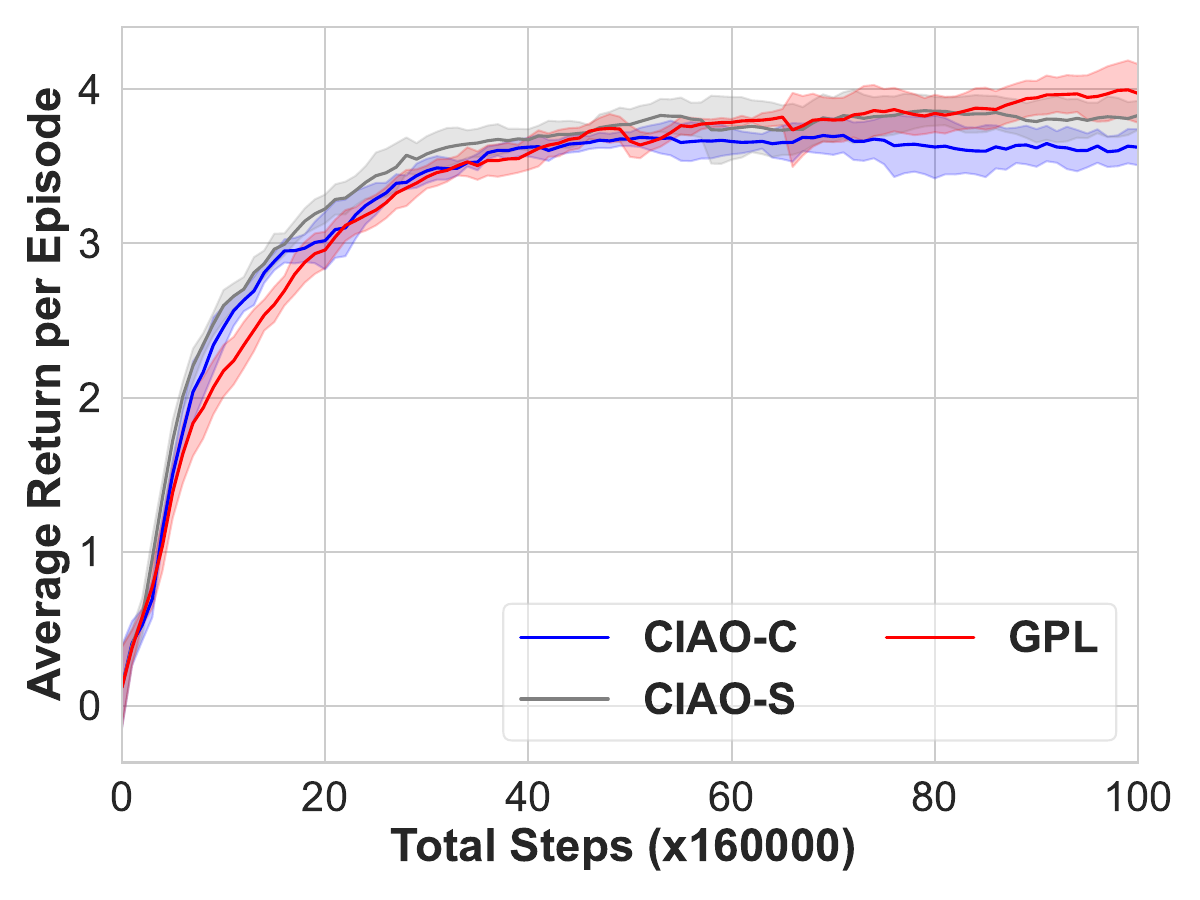}
                    \caption{LBF including A2C agent.}
                \end{subfigure}
                \hfill
                \begin{subfigure}[b]{0.33\textwidth}
                    \centering
                    \includegraphics[width=\textwidth]{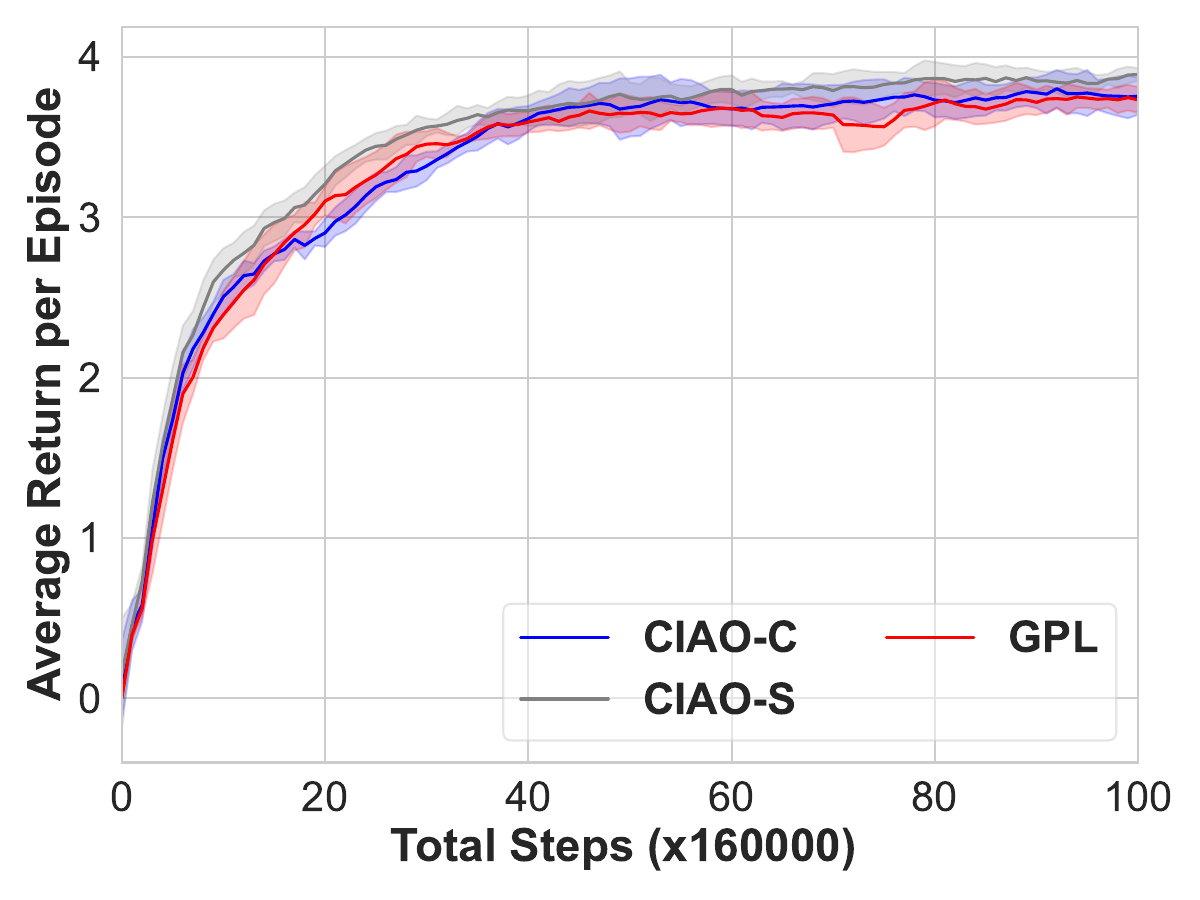}
                    \caption{LBF excluding A2C agent.}
                \end{subfigure}
                \hfill
                \begin{subfigure}[b]{0.33\textwidth}
                    \centering
                    \includegraphics[width=\textwidth]{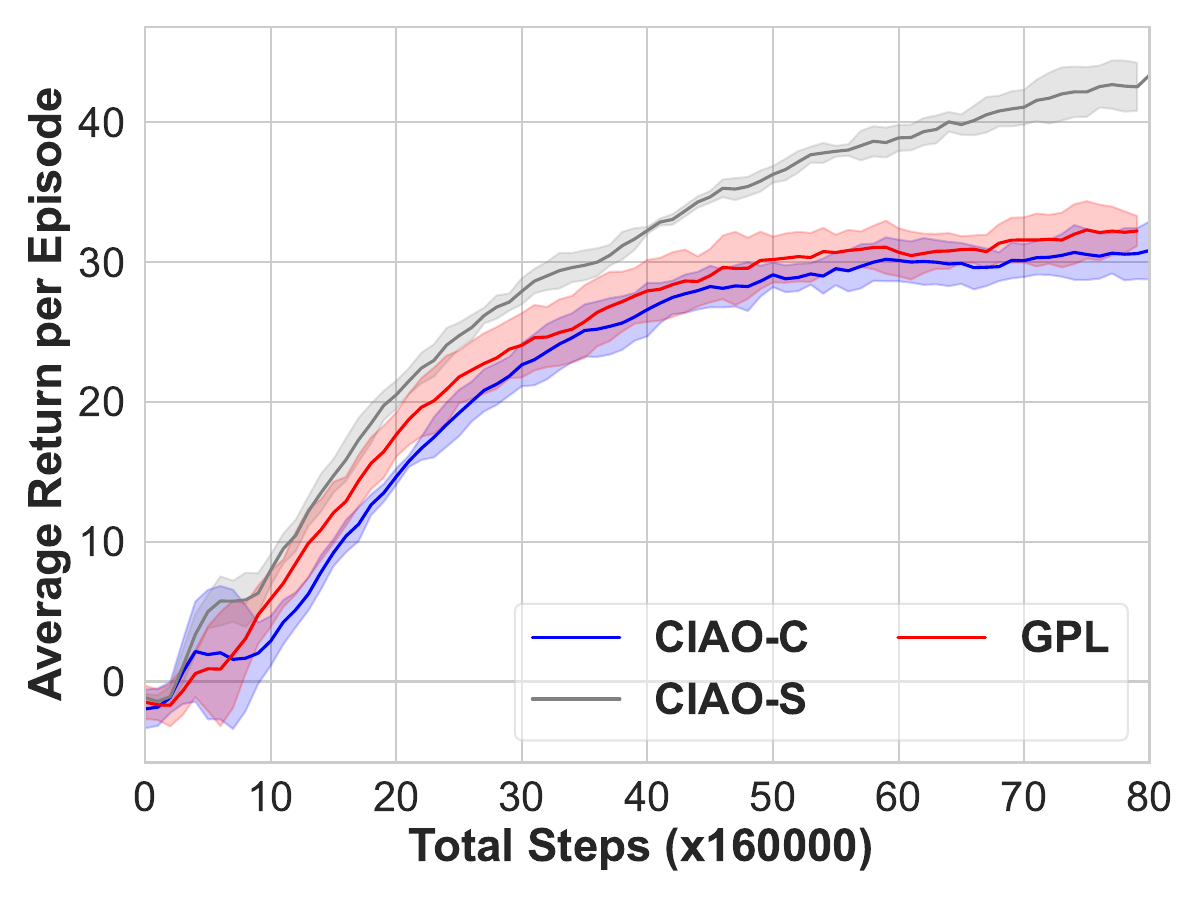}
                    \caption{Wolfpack.}
                \end{subfigure}
            \caption{Comparison between CIAO and GPL in evaluation, across different scenarios of a maximum of 3 agents.}
            \label{fig:test-3-agents}
            \end{figure}
            We present a performance comparison between CIAO and GPL across various scenarios involving a maximum of 3 agents, as illustrated in Fig.~\ref{fig:test-3-agents}. The results indicate comparable performances on LBF, while CIAO-S significantly outperforms the other algorithms in the Wolfpack scenario. This observation leads to the conclusion that the star graph structure is better suited for Wolfpack. The rationale behind this outcome is that, in instances with a small number of agents in Wolfpack, conveying the learner's 'instructions' through one teammate to another is less effective. This contrasts with the scenario depicted in Fig.~\ref{fig:wolfpack-test-9-normal}, where a larger number of agents necessitates transmitting the learner's instructions through an intermediary teammate. The consistency of these findings reinforces the argument for the star graph structure's superiority in Wolfpack scenarios.

        \subsection{LBF with Agent-Type Sets Excluding A2C Agent}
        \label{subsec:sensitivity_to_different_teammates}
            \begin{figure}[ht!]
            \centering
                \begin{subfigure}[b]{0.45\textwidth}
                \centering
                    \includegraphics[width=\textwidth]{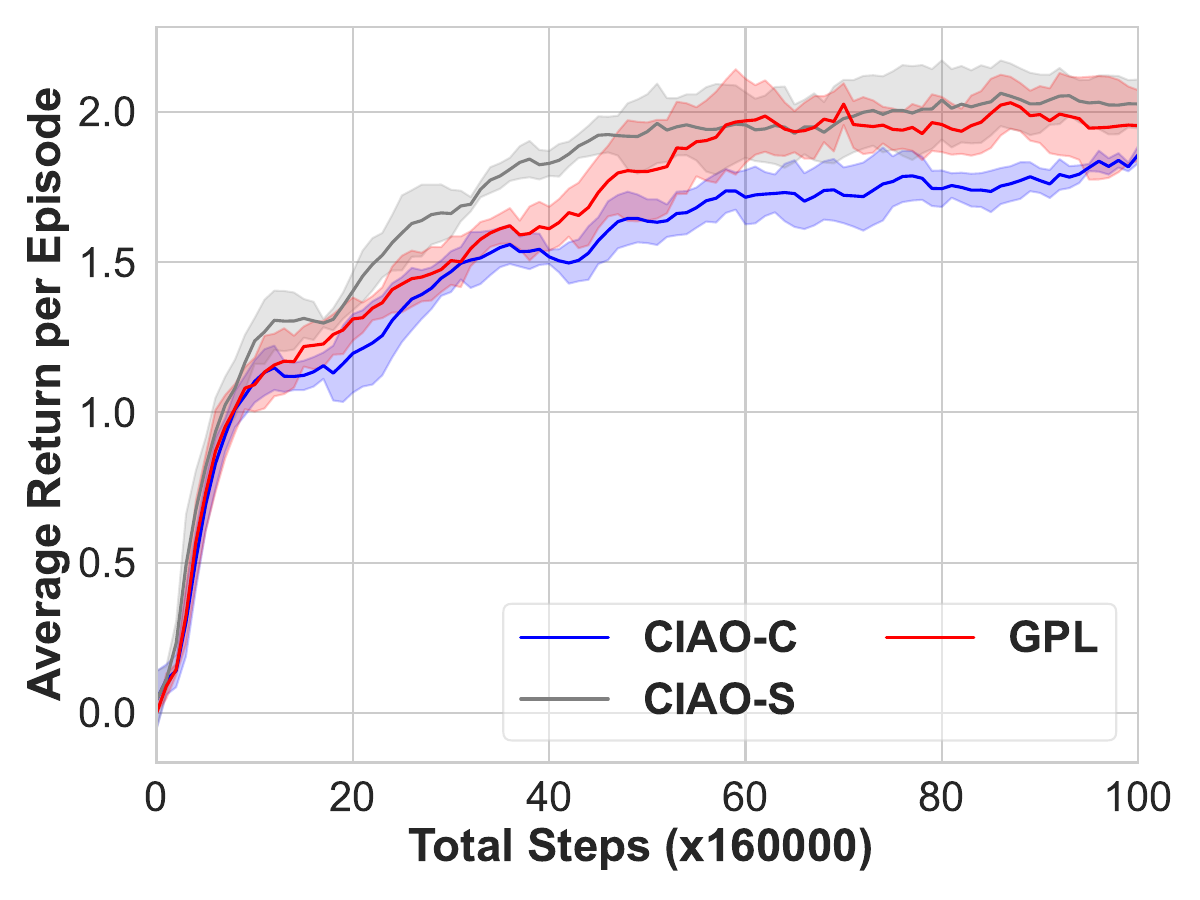}
                    \caption{LBF: max. of 5 agents.}
                    \label{fig:lbf-test-5-noA2C}
                \end{subfigure}
                \begin{subfigure}[b]{0.45\textwidth}
                \centering
                    \includegraphics[width=\textwidth]{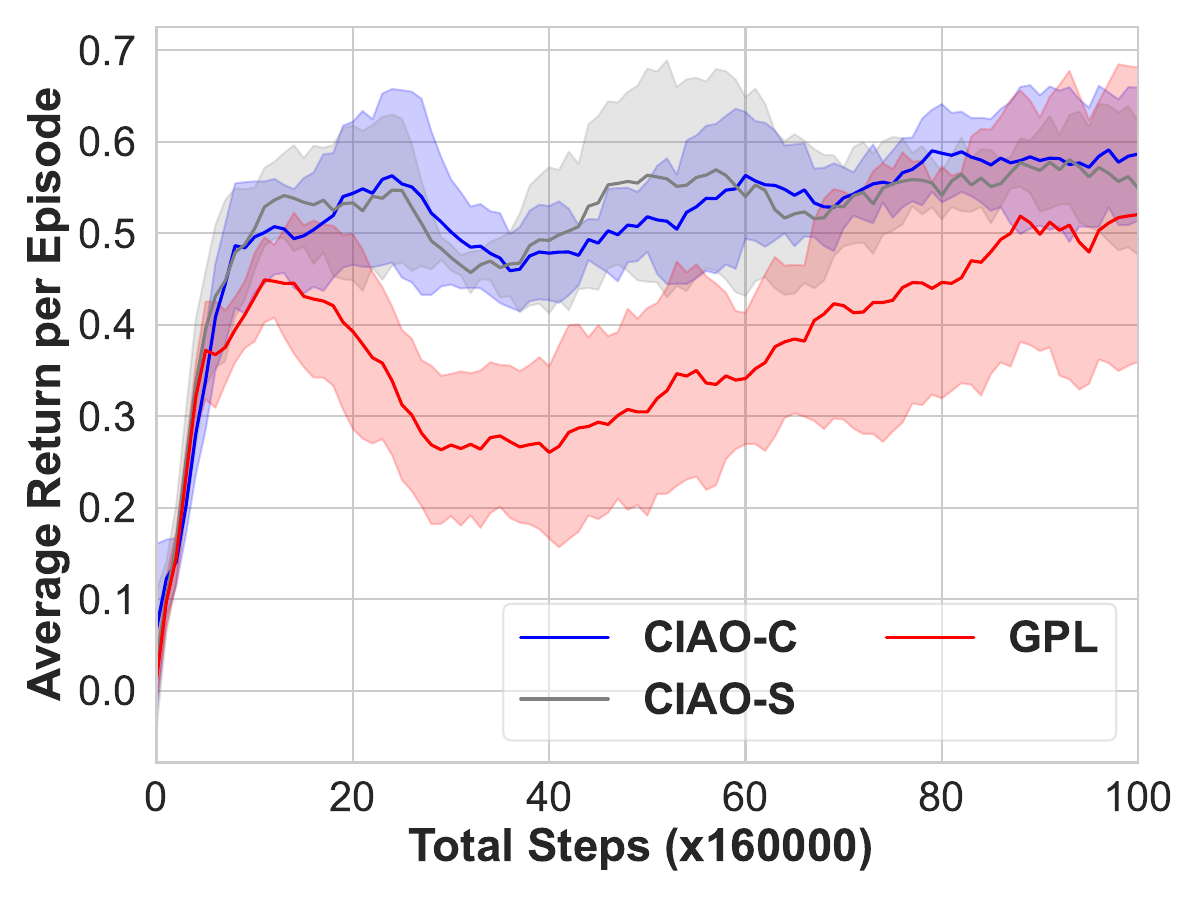}
                    \caption{LBF: max. of 9 agents.}
                    \label{fig:lbf-test-9-noA2C}
                \end{subfigure}
                \caption{Comparison between CIAO and GPL on LBF with the agent-type set excluding the agent-type generated by RL algorithms, in scenarios with a maximum of 5 and 9 agents.}
            \label{fig:LBF_noA2C}
            \end{figure}
            We extend our evaluation of CIAO to LBF considering the agent-type set without the agent-type trained by RL (A2C agent), as depicted in Fig.~\ref{fig:LBF_noA2C}. A comparison between Fig.~\ref{fig:LBF_noA2C} and Fig.~\ref{fig:main_results} leads to the conclusion that CIAO-S exhibits comparatively robust performance across different agent-type sets, whereas CIAO-C demonstrates robustness primarily in scenarios with a larger number of agents. The underlying reasons for CIAO-C's limited robustness in situations with a small number of agents remain a topic for future investigation. Additionally, exploring the correlation between the performance of these algorithms in testing and RL-based agent-types is a valuable topic for further research. 
        
        \subsection{Additional Ablation Study on LBF with Agent-Type Sets Excluding A2C Agent}
        \label{subsec:ablation_lbf_no_a2c}
            \begin{figure}[ht!]
                \centering
                \begin{subfigure}[b]{0.33\textwidth}
                    \centering
                    \includegraphics[width=\textwidth]{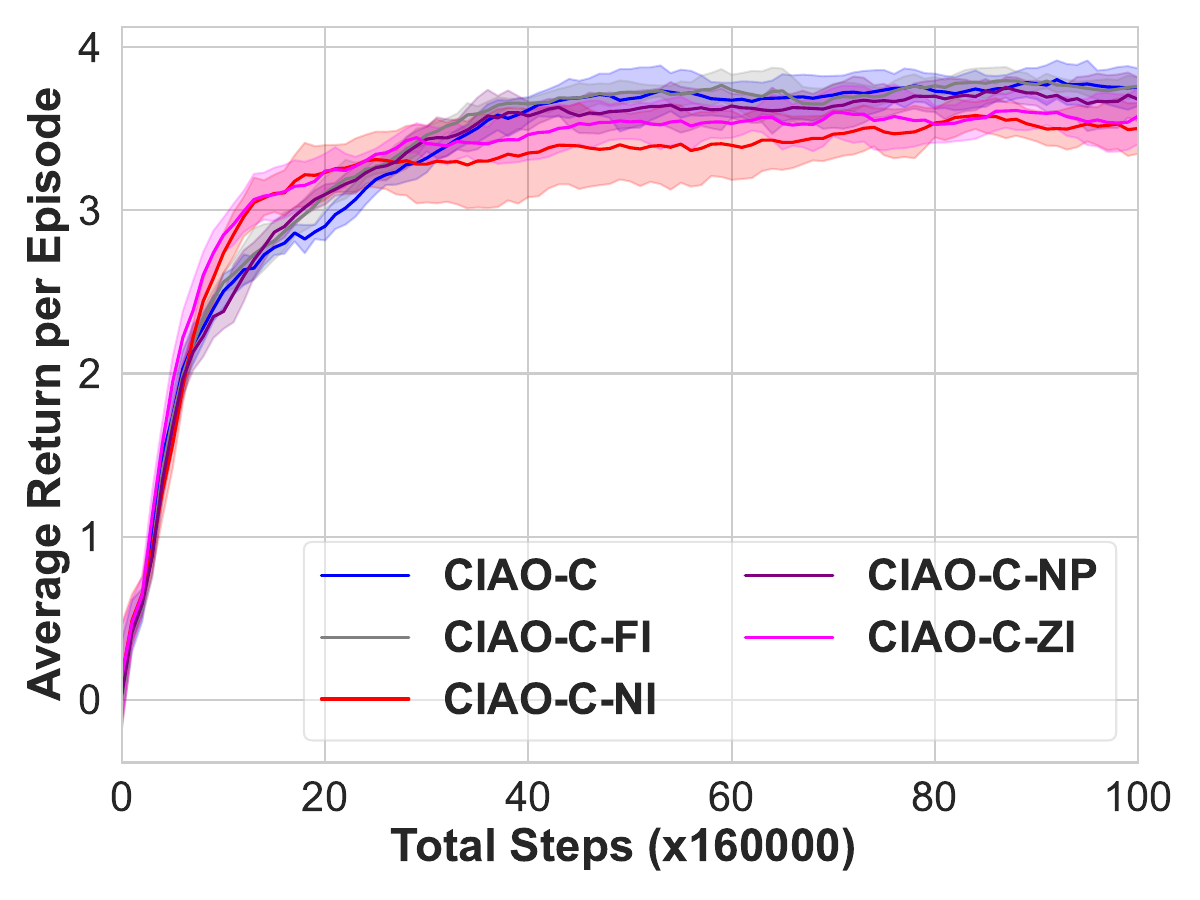}
                    \caption{Maximum of 3 agents.}
                \end{subfigure}
                \hfill
                \begin{subfigure}[b]{0.33\textwidth}
                    \centering
                    \includegraphics[width=\textwidth]{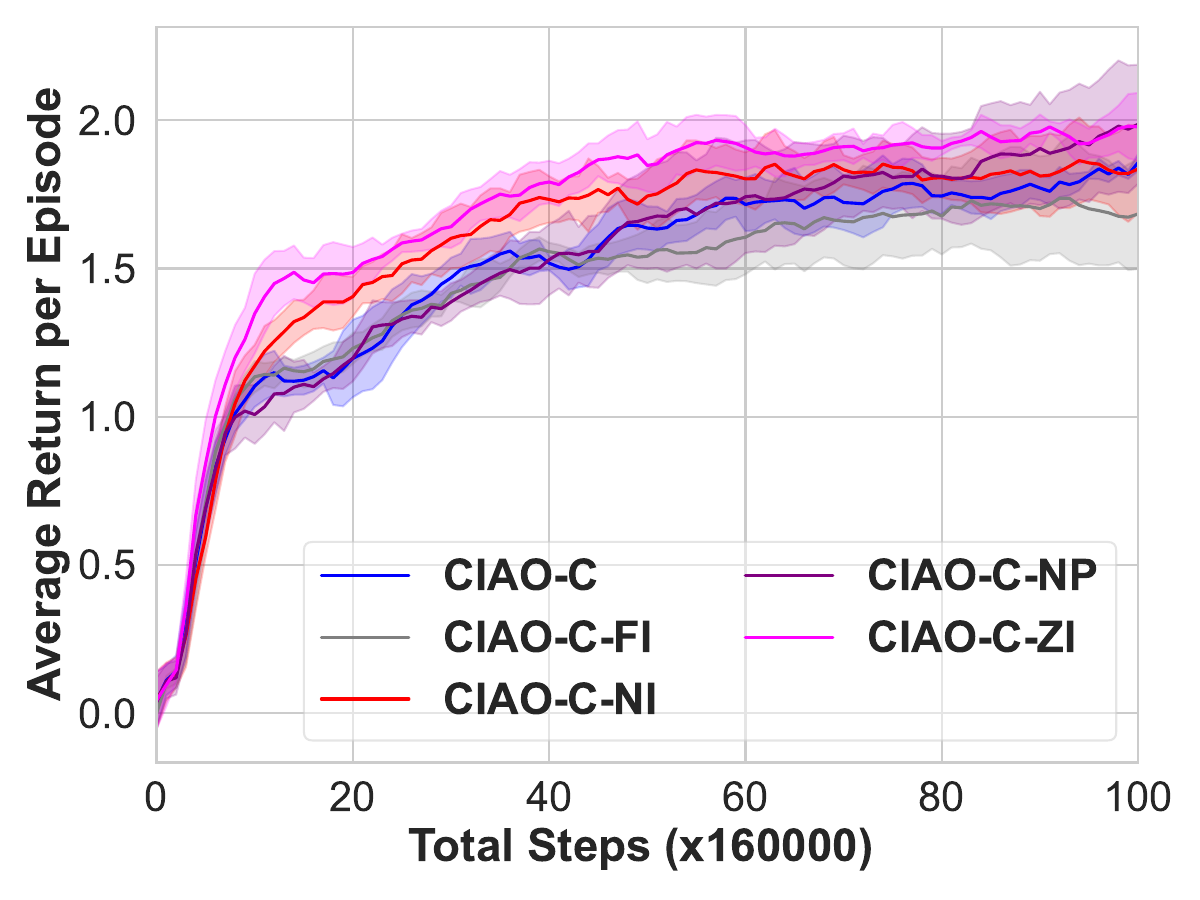}
                    \caption{Maximum of 5 agents.}
                \end{subfigure}
                \hfill
                \begin{subfigure}[b]{0.33\textwidth}
                    \centering
                    \includegraphics[width=\textwidth]{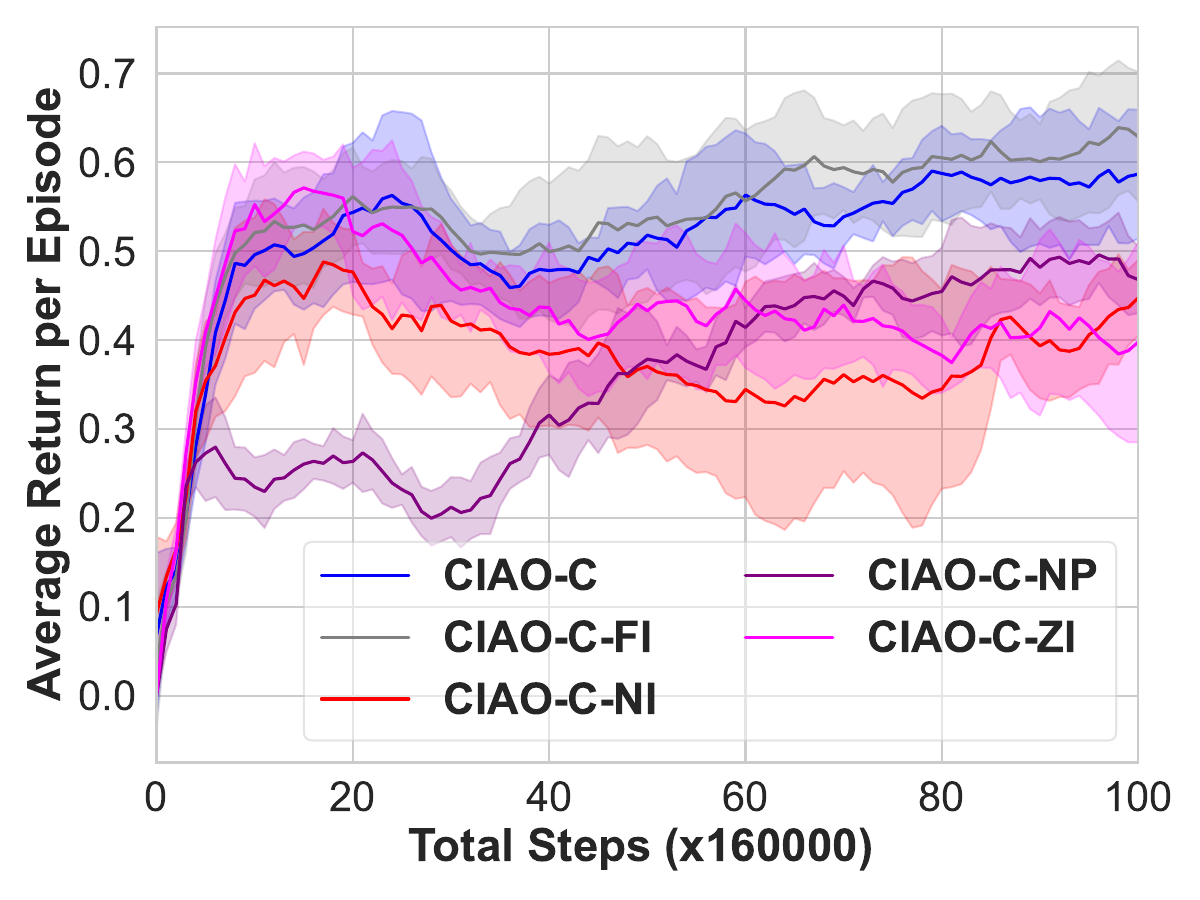}
                    \caption{Maximum of 9 agents.}
                \end{subfigure}
                \caption{Comparison between CIAO-C and its ablations in evaluation, on LBF excluding A2C agent, across scenarios of various maximum numbers of agents as 3, 5 and 9, respectively.}
                \label{fig:ablation-complete-noA2C}
            \end{figure}

            \begin{figure}[ht!]
                \centering
                \begin{subfigure}[b]{0.33\textwidth}
                    \centering
                    \includegraphics[width=\textwidth]{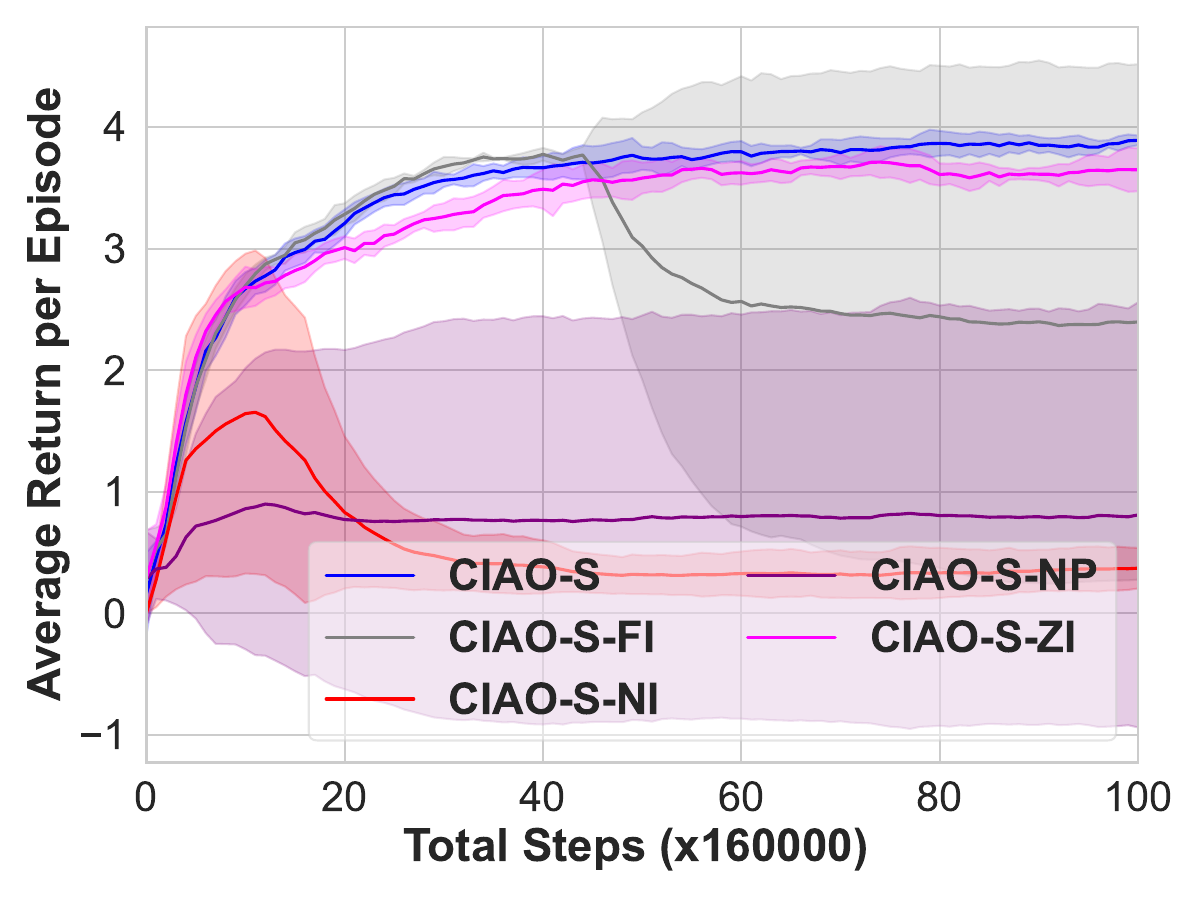}
                    \caption{Maximum of 3 agents.}
                \end{subfigure}
                \hfill
                \begin{subfigure}[b]{0.33\textwidth}
                    \centering
                    \includegraphics[width=\textwidth]{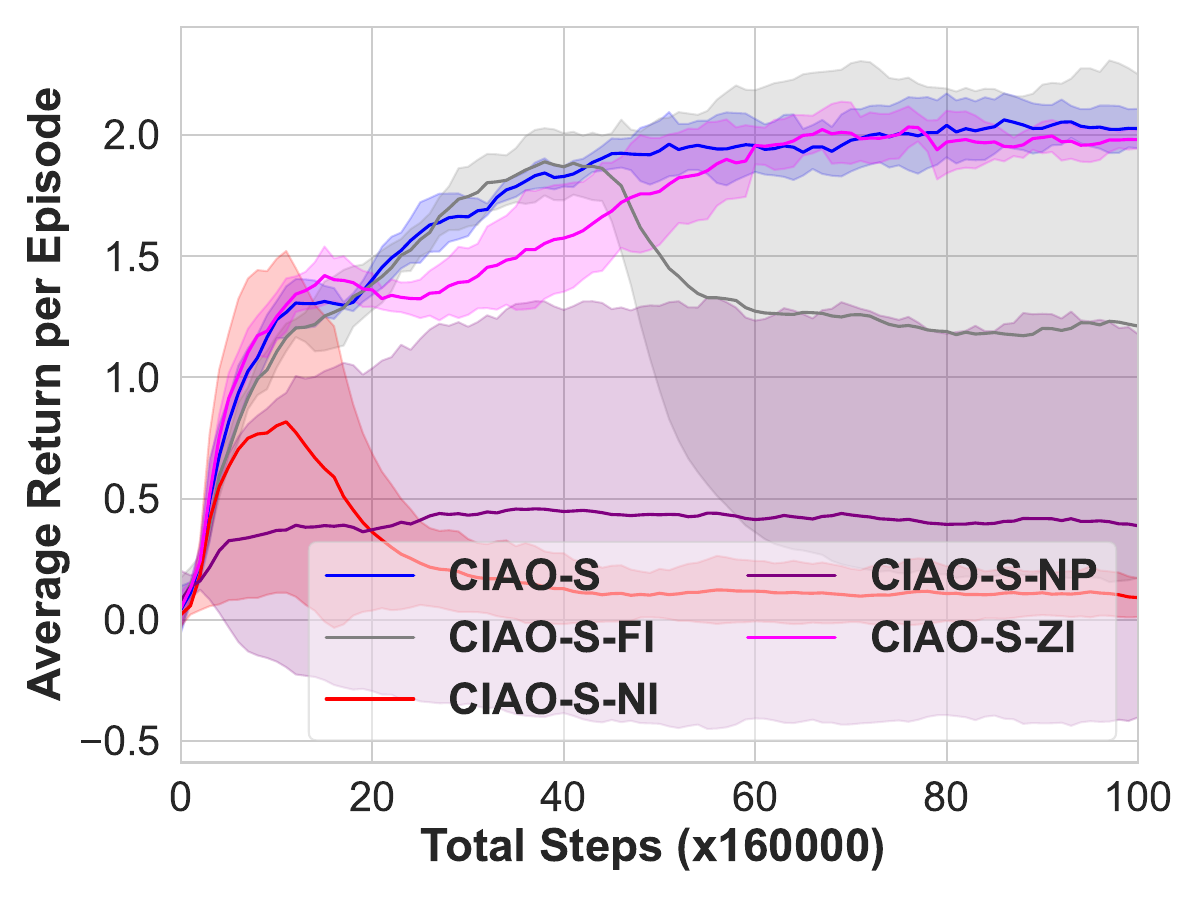}
                    \caption{Maximum of 5 agents.}
                \end{subfigure}
                \hfill
                \begin{subfigure}[b]{0.33\textwidth}
                    \centering
                    \includegraphics[width=\textwidth]{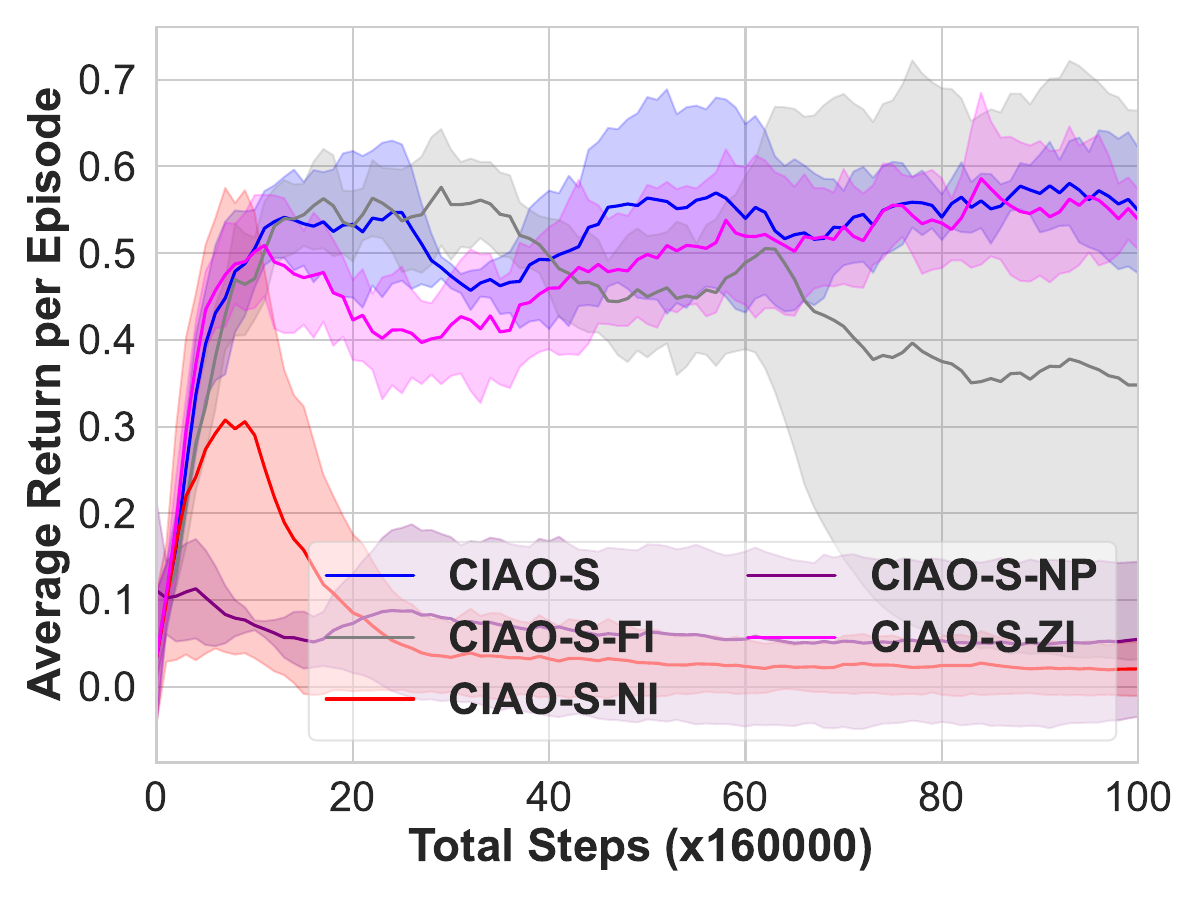}
                    \caption{Maximum of 9 agents.}
                \end{subfigure}
                \caption{Comparison between CIAO-S and its ablations in evaluation, on LBF excluding A2C agent, across scenarios of various maximum numbers of agents.}
                \label{fig:ablation-star-noA2C}
            \end{figure}
            We present a comprehensive performance comparison among CIAO-C, CIAO-S, and their respective ablation variants on LBF, excluding the A2C agent. Figs. \ref{fig:ablation-complete-noA2C} and \ref{fig:ablation-star-noA2C} illustrate the results for CIAO-C and CIAO-S, respectively. In the majority of situations, our hypothesis regarding the non-negative individual utility range is validated. However, we note that the unregularized individual utility exhibits satisfactory performance but is prone to instability. Additionally, our theoretical expectation of a non-negative pairwise utility range is violated for CIAO-C in scenarios involving a maximum of 3 and 5 agents. The root cause of this deviation requires further investigation, suggesting a potential avenue for future research into dynamic affinity graph structures.

        \subsection{Additional Ablation Study on CIAO with No Regularizers}
        \label{subsec:ablation_noReg}
            \begin{figure}[ht!]
                \centering
                \begin{subfigure}[b]{0.33\textwidth}
                    \centering
                    \includegraphics[width=\textwidth]{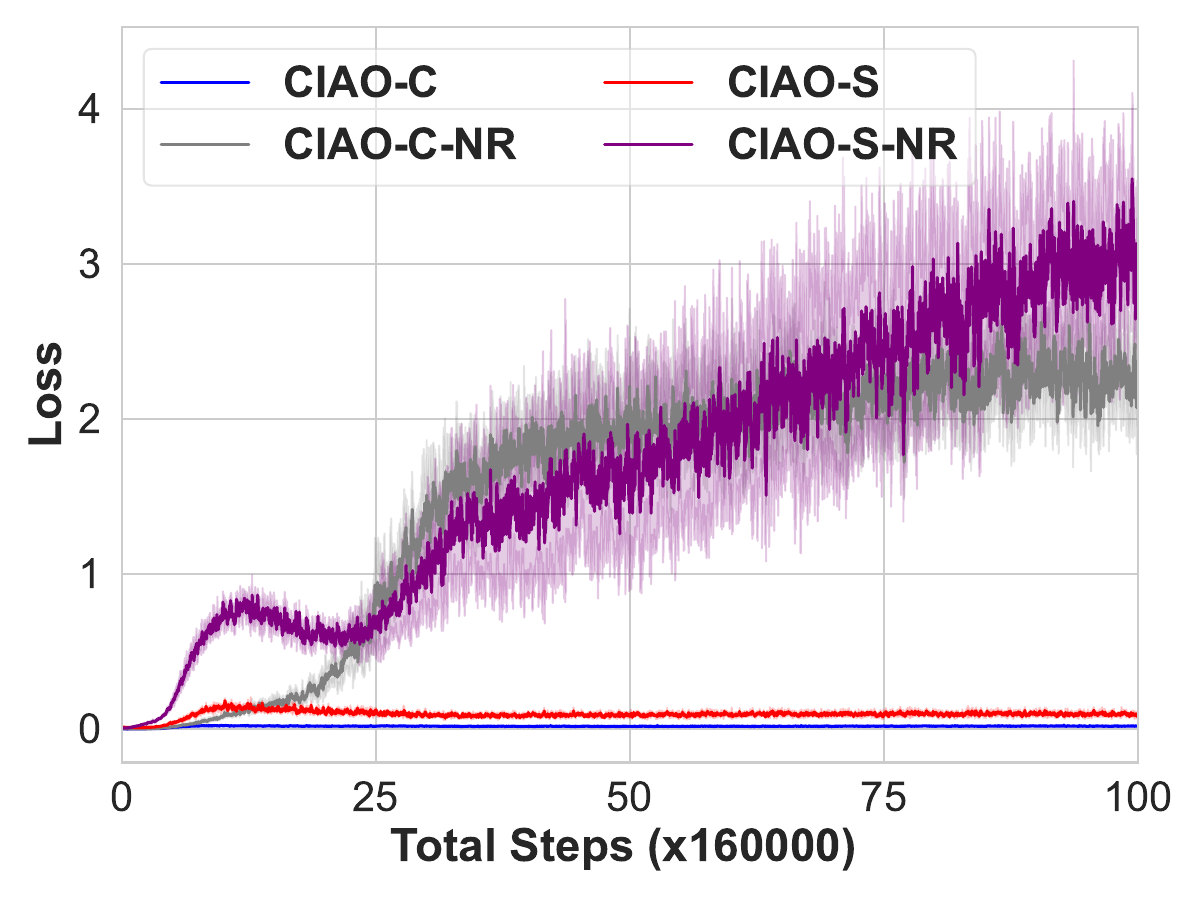}
                    \caption{LBF including A2C agent.}
                \end{subfigure}
                \hfill
                \begin{subfigure}[b]{0.33\textwidth}
                    \centering
                    \includegraphics[width=\textwidth]{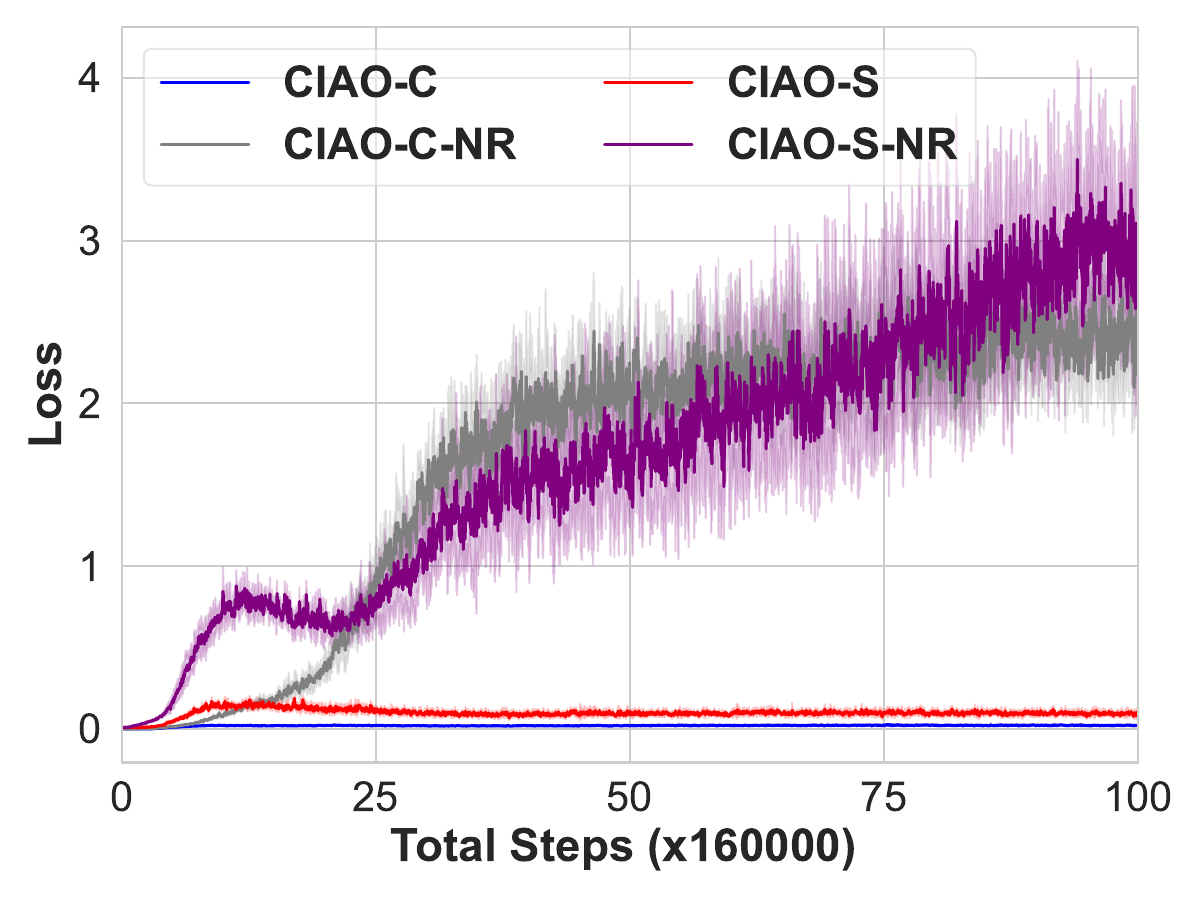}
                    \caption{LBF excluding A2C agent.}
                \end{subfigure}
                \hfill
                \begin{subfigure}[b]{0.33\textwidth}
                    \centering
                    \includegraphics[width=\textwidth]{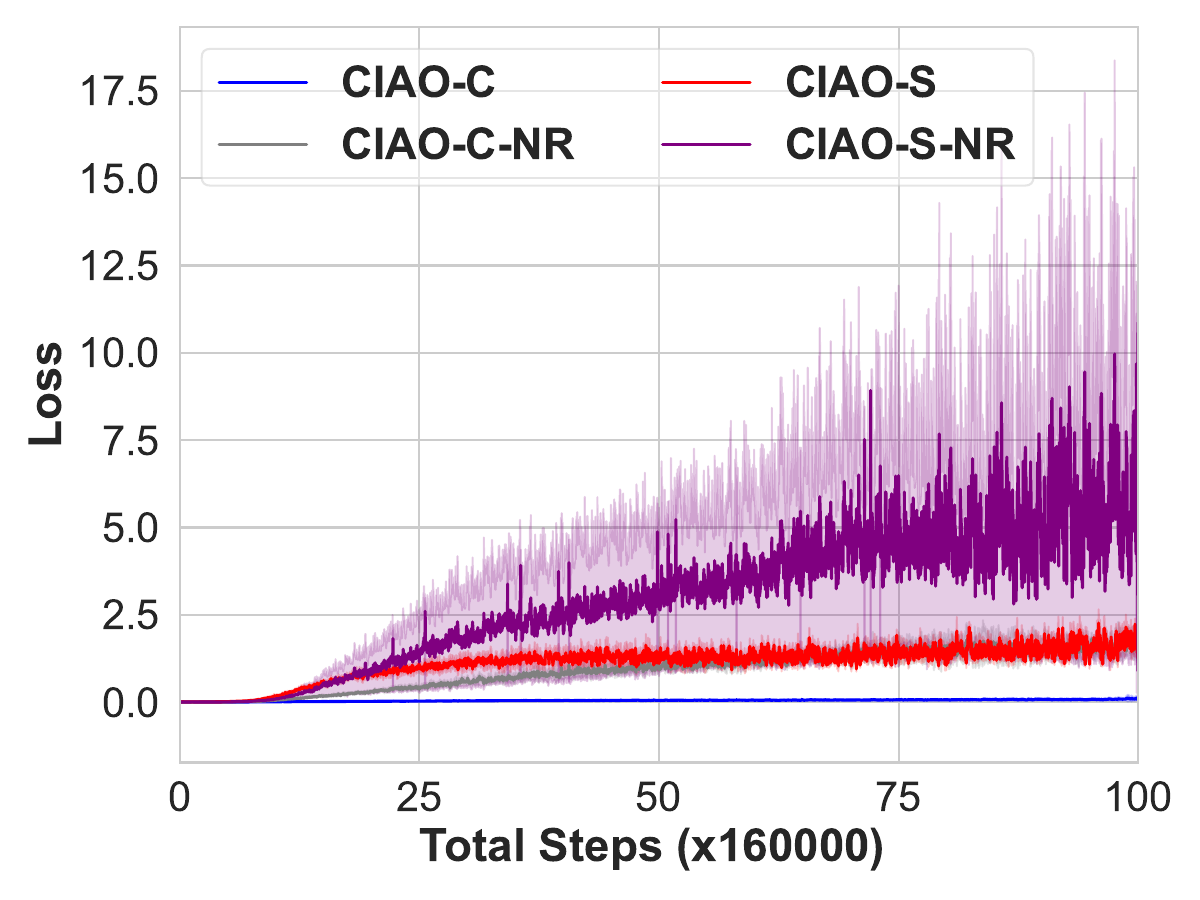}
                    \caption{Wolfpack.}
                \end{subfigure}
                \caption{Comparison between CIAO and its ablation variant with no consideration of regularizers, denoted as CIAO-X-NR (``X'' is either ``C'' or ``S''), on the regularization loss during training, across different scenarios where the training is conducted with a maximum of 3 agents.}
                \label{fig:ablation-noReg-loss}
            \end{figure}

            \begin{figure}[ht!]
                \centering
                \begin{subfigure}[b]{0.33\textwidth}
                    \centering
                    \includegraphics[width=\textwidth]{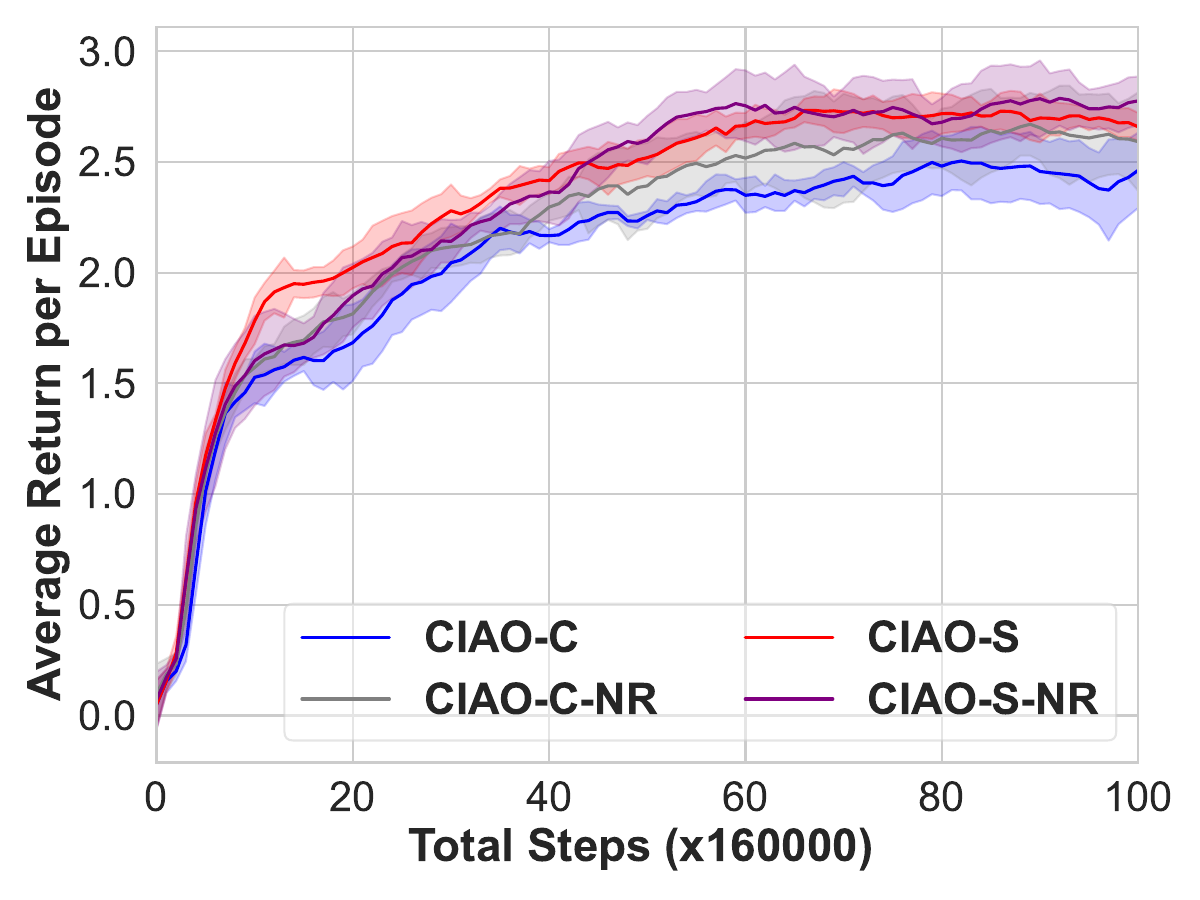}
                    \caption{LBF including A2C agent.}
                \end{subfigure}
                \hfill
                \begin{subfigure}[b]{0.33\textwidth}
                    \centering
                    \includegraphics[width=\textwidth]{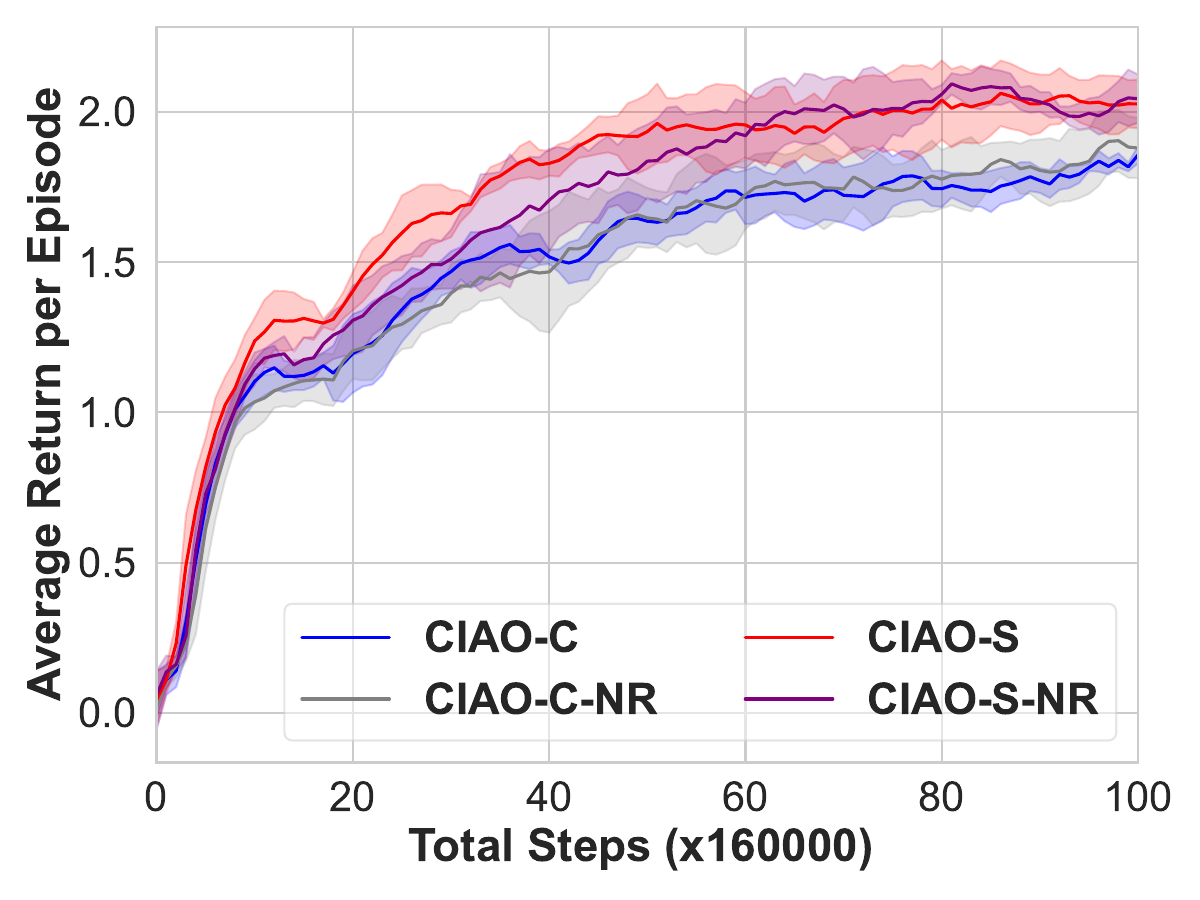}
                    \caption{LBF excluding A2C agent.}
                \end{subfigure}
                \hfill
                \begin{subfigure}[b]{0.33\textwidth}
                    \centering
                    \includegraphics[width=\textwidth]{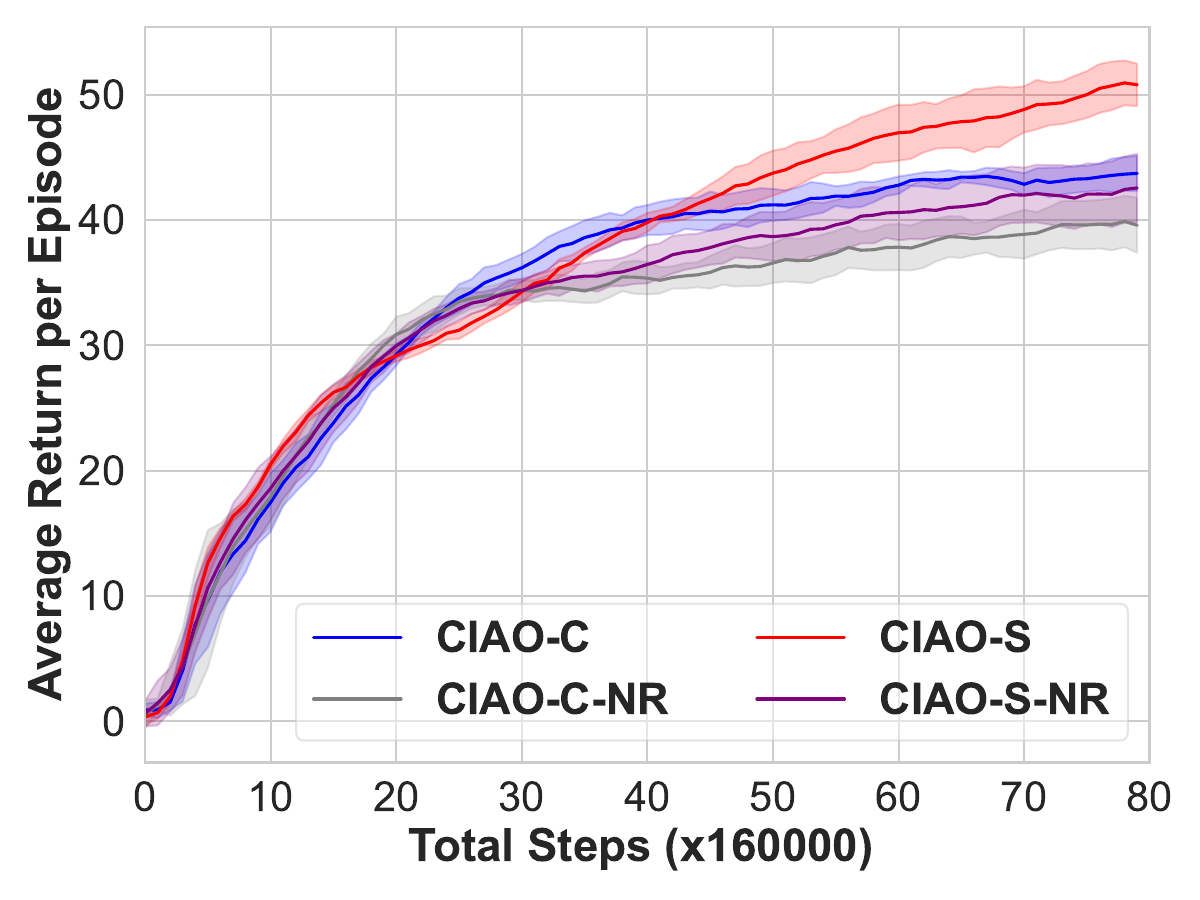}
                    \caption{Wolfpack.}
                \end{subfigure}
                \caption{Comparison between CIAO and its ablation variant with no consideration of regularizers, denoted as CIAO-X-NR (``X'' is either ``C'' or ``S''), across different scenarios where the evaluation is conducted with a maximum of 5 agents.}
                \label{fig:ablation-noReg-5-eval}
            \end{figure}

            \begin{figure}[ht!]
                \centering
                \begin{subfigure}[b]{0.33\textwidth}
                    \centering
                    \includegraphics[width=\textwidth]{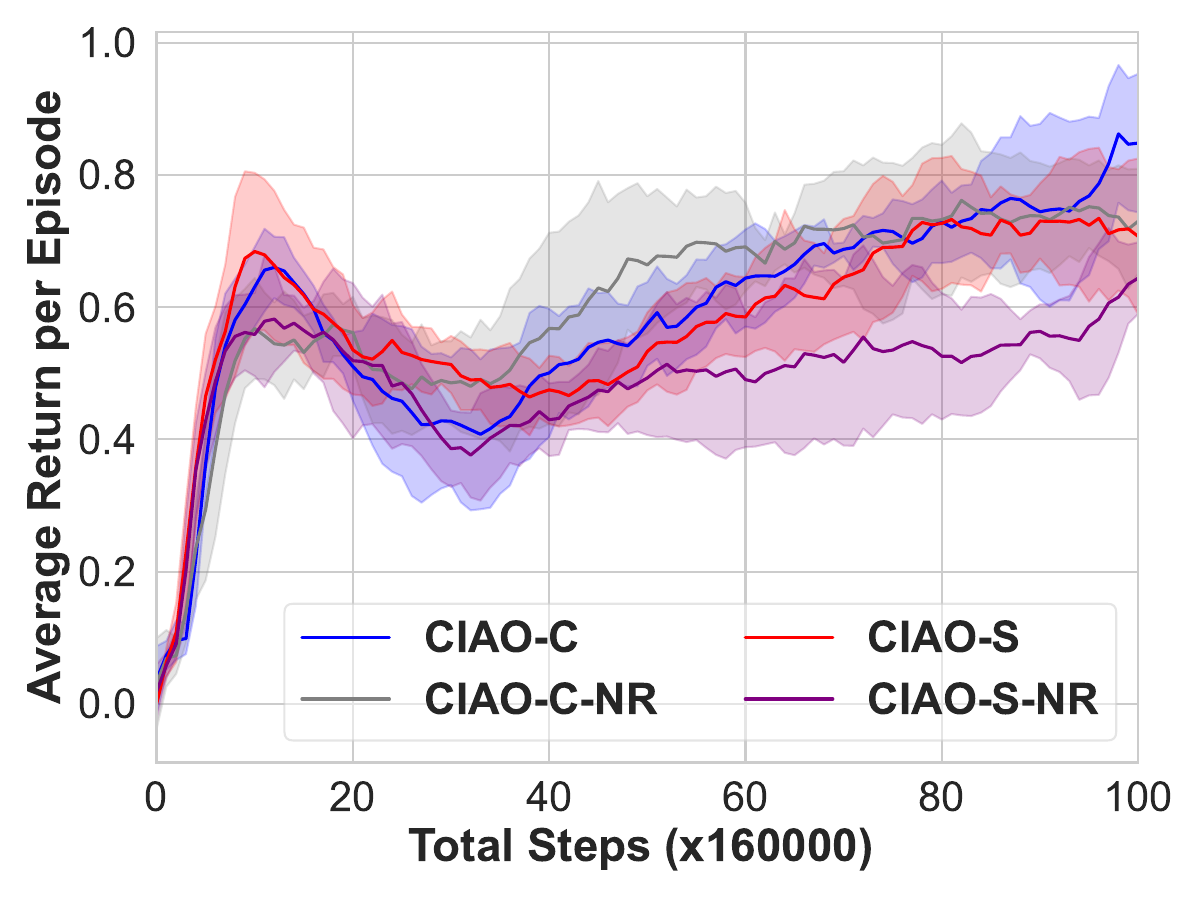}
                    \caption{LBF including A2C agent.}
                \end{subfigure}
                \hfill
                \begin{subfigure}[b]{0.33\textwidth}
                    \centering
                    \includegraphics[width=\textwidth]{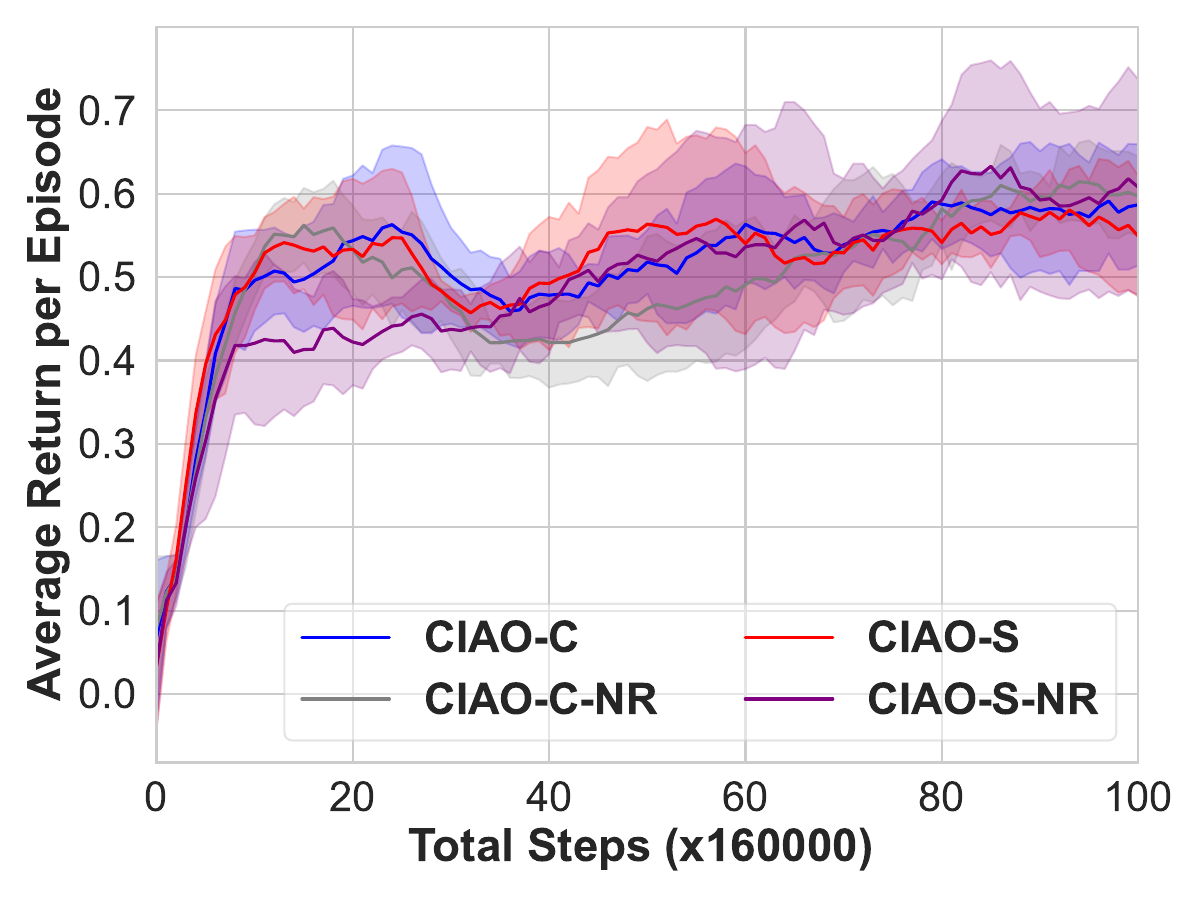}
                    \caption{LBF excluding A2C agent.}
                \end{subfigure}
                \hfill
                \begin{subfigure}[b]{0.33\textwidth}
                    \centering
                    \includegraphics[width=\textwidth]{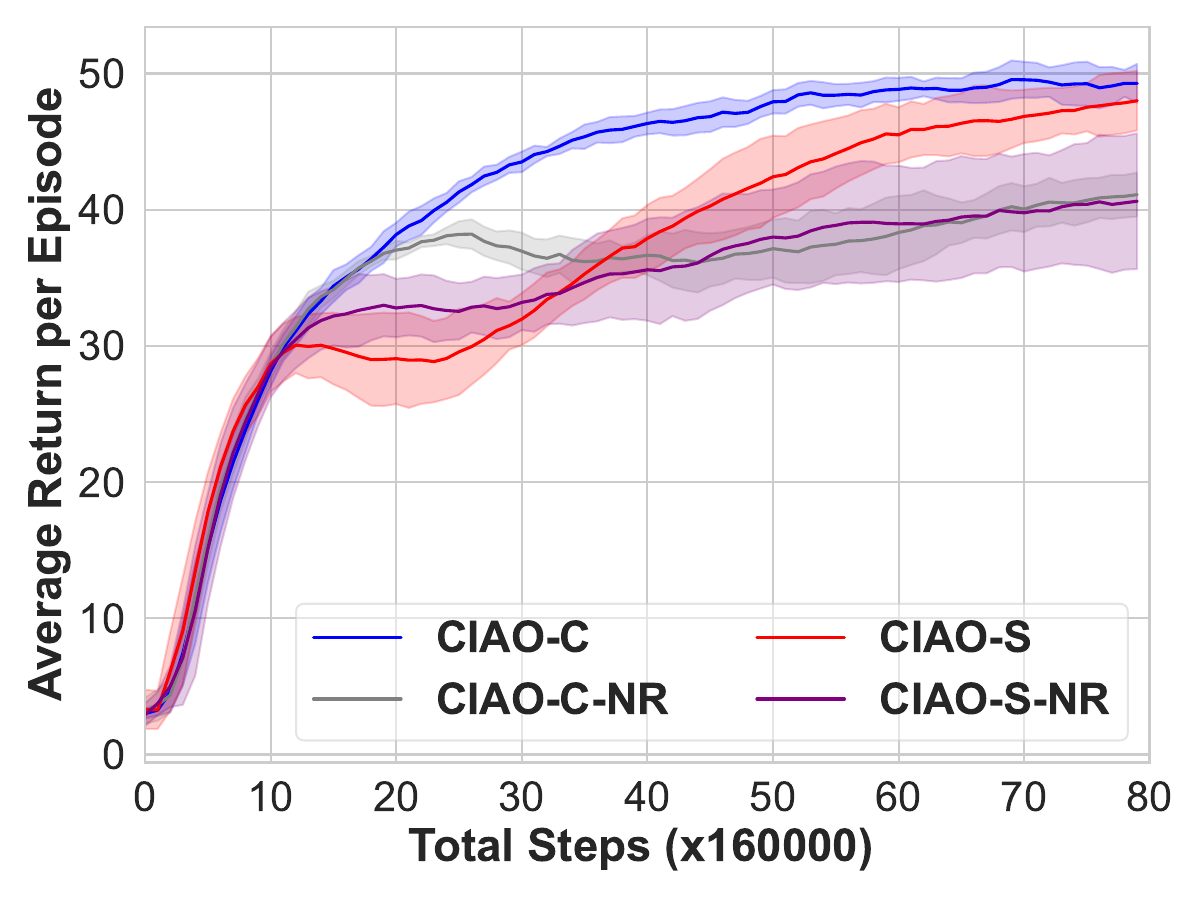}
                    \caption{Wolfpack.}
                \end{subfigure}
                \caption{Comparison between CIAO and its ablation variant with no consideration of regularizers, denoted as CIAO-X-NR (``X'' is either ``C'' or ``S''), across different scenarios where the evaluation is conducted with a maximum of 9 agents.}
                \label{fig:ablation-noReg-9-eval}
            \end{figure}

            We conduct a performance comparison between CIAO and its ablation variant, excluding considerations of regularizers. In Fig. \ref{fig:ablation-noReg-loss}, the regularization losses during training are depicted, affirming the importance of incorporating regularizers. Notably, the effectiveness of regularizers is not consistently robust in the context of LBF, as shown in Figs.~\ref{fig:ablation-noReg-5-eval} and \ref{fig:ablation-noReg-9-eval}. Two potential explanations arise: (1) unique properties of the LBF environment may diminish the impact of regularizers, and (2) the regularization, driven by a sufficient condition to address DVSC as an RL problem, may lack consideration of other eligible conditions. The exploration of these possibilities is deferred to the future research.

\end{document}